\documentclass[11pt,a4paper]{report}
\usepackage[ngerman,english]{babel}
\usepackage{geometry}
\usepackage{latexsym,amsmath,amstext,amsthm,amscd,amsopn,verbatim,amsrefs}
\usepackage{graphicx}
\usepackage[dvips]{color}
\usepackage{epsfig}
\usepackage[all]{xy}
\usepackage{amscd}
\usepackage{amssymb}
\usepackage{latexsym}
\usepackage{amsfonts}
\usepackage{hyperref}
\hypersetup{colorlinks,citecolor=black,filecolor=black,linkcolor=black,urlcolor=black}

\renewcommand{\d}{{\rm d}}

\newcommand{\I}{\mathrm{i}}
\DeclareMathOperator{\id}{id}

\newtheorem{satz}{Satz}[subsection]
\newtheorem{theorem}[satz]{Theorem}
\newtheorem*{theorem*}{Theorem}

\newtheorem{proposition}[satz]{Proposition}

\newtheorem{definition}[satz]{Definition}
\newtheorem*{definition*}{Definition}

\makeatletter
\newcommand*{\toccontents}{\@starttoc{toc}}
\makeatother

\begin{document}
\thispagestyle{empty}
\begin{center}\begin{huge}\textbf{\textit{Singular propagators in deformation quantization and Shoikhet-Tsygan formality}}\end{huge}\\\vspace{0.1cm}
\begin{large}{\textit{Johannes L\"offler}}\end{large}\\\vspace{0.1cm}\end{center}

\section*{Abstract}
This paper adds some details to the seminal approach to logarithmic formality \cite{AWRT} and interpolation formality \cite{WR} by Alekseev, Rossi, Torossian and Willwacher: We prove that the interpolation family of Kontsevich formality maps extends to Shoikhet-Tsygan formality and a complex interpolation parameter. We show some elementary relations satisfied by this polynomials. We also compute some Kontsevich integral weights and reason on the number theoretic meaning of the invariance of Kontsevich's propagator under real translations and scalings in the case of the Merkulov $n$-wheels.

\section*{Contents}

\toccontents\vspace{1cm}

{\large \textbf{Acknowledgements}} First I want to heartily thank C.-A. Rossi. The discussions with him set the ground for some of the issues considered here and most of the work on the logarithmic propagator are joint with him, A. Alekseev and C. Torossian and T. Willwacher. I am deeply grateful to Willwacher for many useful discussions and comments. I thank S. Merkulov for useful comments. I want to thank P. Moore for a motivating discussion and the comment to compare my result with an established formula. For helpful discussions of my notes I thank H. Furusho and P. Mnev. 

Last but not least I thank the MPIM \textsc{Max-Planck-Institute For Mathematics} in Bonn for financial support. johannes@mpim-bonn.mpg.de

\clearpage\thispagestyle{empty}

\chapter{\textit{Introduction}}
This paper is concerned with recent developments in the framework of deformation quantization. In deformation quantization the central definition, established by Moyal, Weyl \cite{Moy}, Berezin \cite{BeQu,BeGCoQu}, Bayen, Flato, Fr\o{}nsdal, Lichnerowicz and Sternheimer \cite{BFFLS1,BFFLS2}, is the definition of a star product $\star$. To give the definition of $\star$ products we first introduce the notion of a Poisson manifold and it is not difficult to prove that every $\star$ product always defines a so-called Poisson bracket by its skew-symmetric part in the first order of a formal parameter $\hbar$. A Poisson bracket on a manifold $M$ can be identified with a skew-symmetric $2$-vector field $\Pi\in\Gamma(M,\wedge^2TM)$ that satisfies
\begin{equation}\label{Poisson} 
\Pi^{ia}\frac{\partial\Pi^{bc}}{\partial{x}^i}+\Pi^{ib}\frac{\partial\Pi^{ca}}{\partial{x}^i}+\Pi^{ic}\frac{\partial\Pi^{ab}}{\partial{x}^i}=0
\end{equation}
Equivalent we can say that the bilinear map $\{\cdot,\cdot\}:C^{\infty}(M)\times{C}^{\infty}(M)\rightarrow{C}^{\infty}(M)$ defined by
$\{f,g\}:=\Pi(\d{f},\d{g})=\Pi^{kl}\frac{\partial{f}}{\partial{x}^k}\frac{\partial{g}}{\partial{x}^l}$
satisfies the Jacobi identity
$\{f,\{g,h\}\}+\{g,\{h,f\}\}+\{h,\{f,g\}\}=0$.
Examples of Poisson manifolds $(M,\Pi)$ are symplectic manifolds where we have a non-degenerate $2$-vector field or in contrast another natural family of Poisson manifolds are the duals of finite-dimensional Lie algebras $\mathfrak{g}$ where a in some sense linear Poisson bracket is canonically defined with help of the structure constants of $\mathfrak{g}$.


\begin{definition}\label{Star}
A $\star$ product on $(M,\Pi)$ is a $\mathbb{C}[\![\hbar]\!]$-bilinear associative operation
$\star:{C^{\infty}(M)[\![\hbar]\!]}\times{C^{\infty}(M)[\![\hbar]\!]}\rightarrow{C^{\infty}(M)[\![\hbar]\!]}$
of the shape
$f\star{g}=\sum_{n=0}^{\infty}\hbar^{n}B_{n}(f,g)$
with bi-differential operators $B_n$ and with the properties $1\star{f}=f=f\star1$, $B_{0}(f,g)=fg$ and $B_{1}(f,g)-B_{1}(g,f)=\I\lbrace{f,g}\rbrace$.
Two star products on $(M,\Pi)$ are equivalent if there exists a formal series
$S=id+\sum_{l=1}^{\infty}\hbar^{l}S_{l}$
of $\mathbb{C}[[\hbar]]$-linear operators $S_{l}:{C^{\infty}(M)[[\hbar]]}\rightarrow{C^{\infty}(M)[[\hbar]]}$ with $S_{l}(1)=0$ for $l\geq1$ and
$f\star{g}=S^{-1}(Sf\star'Sg)\;\forall\;f,g\in{C^{\infty}(M)[[\hbar]]}$
\end{definition}

In the mentioned case of a finite dimensional Lie algebra it is well-known that for the canonical linear Poisson structure the deformation quantization method of Kontsevich \cite{K} essentially produces the deformed universal enveloping algebra of $\mathfrak{g}$.

The equivalence just identifies $\star$ products that are easy to construct out of each other by the previous formula, we refer the reader to the reference \cite{WaB} for physical arguments why we put the other conditions in definition \ref{Star}. Of course the definition of a $\star$ product is not a complete way to quantize a specific physical system, the representation of the $\star$ product algebra as operators acting on a Hilbert space $(\mathcal{H},\langle,\cdot,\cdot\rangle)$, here $\langle\cdot,\cdot\rangle:\mathcal{H}\times\mathcal{H}\rightarrow\mathcal{C}$ is a complex positive-definite inner product, are necessary because of the superposition principle for states of the quantum system \cite{Dir}. In the realm of the theory of deformation quantization a step forward to the Hilbert space representation of a $\star$ product algebra can be achieved with help of a Gel'fand-Naimark-Segal-construction in the sense of Bordemann and Waldmann \cite{BoWaGNS}, \cite{WaB}.

In the symplectic case the existence of $\star$ products has been first proved by DeWilde and Lecomte \cite{DeWL} and soon later Fedosov gave a quite explicit, elegant and very geometric construction \cite{FeS}. We will sk{\em etc}h this construction in the following to give the reader examples of $\star$ products where the construction is much easier compared to the general case \cite{K}. This write up differs from the original construction because we replace an essential fixed point equation by its unique solution and give the examples for K\"ahler manifolds of constant sectional curvature. We also consider another aspect of Fedosov's construction in this paper: Classical counterparts of Fedosov Taylor expansions generate the formal exponential map, as first noticed by Emmrich and Weinstein \cite{W}. We will give a constructive proof of this fact, this is done by just computing the Taylor expansion of local coordinate functions and comparison with another formula for the formal exponential map that we get by induction. We also demonstrate how this formula for the exponential map works in the two generic examples $2$-sphere and Poincar\'e half-plane.


Finally in 1997 a celebrated paper \cite{K} of Kontsevich gave an impressive proof of the existence of $\star$ products on any Poisson manifold. Kontsevich's constructive solution is based on his famous formality theorem and this stronger statement is a continuation of the Hochschild-Kostant-Rosenberg map: Let $D_{poly}(M)$ be defined by $D_{poly}(M):=\oplus_{k=-1}^{\infty}\lbrace{\text{Polydifferential operators}:{C}^{\infty}(M)^{\otimes{k+1}}\rightarrow{C}^{\infty}(M)\rbrace}$ and denote by $T_{poly}(M):=\oplus_{k=-1}^{\infty}{\Gamma}^{\infty}({\wedge}^{k+1}TM)$ the space of polyvector fields. The Kontsevich formality map is a $L_{\infty}$-quasi-isomorphism
$\mathcal{U}$ between the differential graded Lie algebra $T_\mathrm{poly}(\mathbb{R}^d)$ equipped with trivial differential and Schouten-Nijenhuis bracket, and the dg Lie algebra $D_\mathrm{poly}(\mathbb{R}^d)$ equipped with Hochschild differential and Gerstenhaber bracket. $\mathcal U$ has the $\mathrm{HKR}$ quasi-isomorphism of dg vector spaces as first Taylor component, and enjoys additional properties, which are relevant for the Fedosov like globalisation of this local result to an arbitrary smooth or complex manifold or algebraic variety. Cattaneo and Felder showed \cite{CaFe}, \cite{CF} that the graphical expansion in Kontsevich's proof can be interpreted by a Feynman diagram expansion of a special $2$-dimensional Poisson $\sigma$-model of Schaller and Strobl \cite{SS}. Let us  mention that Tamarkin gave an independent proof of the formality theorem for $\mathbb{R}^d$ with operadic methods \cite{TaF}, the formality of little discs and Drinfeld associators are some of the main ingredients in Tamarkin's alternative proof.

Kontsevich's formality map is defined with help of certain graphs and associated integral weights.
In the article \textit{Operads and Motives in Deformation Quantization} \cite{KoM} Kontsevich stated (without proof) a formula that yields a second formality map $\mathcal{U}^{\ln}$ where one replaces in the construction of the weights the argument function $\arg(\cdot)$ in the definition of the Kontsevich propagator by a logarithm $\ln(\cdot)$. The proof that this $\ln(\cdot)$ formula indeed enjoys all properties is the aforementioned logarithmic formality theorem, a result joint with Alekseev, Rossi, Torossian and Willwacher.

Further, in \cite{WR} Rossi and Willwacher generalised this result and realised that the standard formality quasi-isomorphism $\mathcal{U}$ and the logarithmic one $\mathcal{U}^{\ln}$ are special cases of a more general formula defining formality maps, in that one may replace the $\arg(\cdot)$-function, in the formula of the Kontsevich $L_\infty$-quasi-isomorphism $\mathcal{U}$, by the family
$\left[(1-\lambda)\ln(\cdot)-\lambda\ln(\overline{\cdot})\right]/{2\pi\I}$
with $\lambda$ any real number and as we will see the case where $\lambda$ is a complex number can be proved analogous. This yields a family of formality maps $\mathcal U^\lambda$ and the original Kontsevich formality corresponds to $\lambda=1/2$, this value and $\lambda=0,1$ are the choices of $\lambda$ with the most symmetry. Rossi and Willwacher originally introduced $\mathcal U^{\mathbb{R}}$ to produce a family of Drinfeld associators \cite{Dri} named $\Phi^\lambda$. This family interpolates between the associators $\Phi_{\mathrm{AT}}$, $\Phi_{\mathrm{KZ}}$ and $\Phi_{\overline{\mathrm{KZ}}}$ (The Alekseev-Torossian associator, the Knizhnik-Zamolodchikov associator and the anti-$\mathrm{KZ}$ associator respectively \cite{WR}).

 By $\Omega^{-\bullet}_A$ we denote the exterior algebra of the $\mathbb K$-module of K\"ahler differentials with reversed grading and by $C_{-\bullet}(A,A)$ the Hochschild chain complex of $A$ with reversed grading. In another chapter we give a brief introduction to Tsygan's formality conjecture and Shoikhet's proof. Tsygan formality conjecture states the existence of an $L_\infty$-quasi-isomorphism
$\mathcal{S}^{\lambda}:C_{-\bullet}(A)\leadsto\Omega^\bullet_A$
of $L_\infty$-modules over $T_\mathrm{poly}(A)$ compatible with $\mathcal{U}$. We show here that Shoikhet's construction works compatible with the interpolation family of propagators if we also replace the $\arg(\cdot)$ function in the Shoikhet integral weights by the interpolation family $\left[(1-\lambda)\ln(\cdot)-\lambda\ln(\overline{\cdot})\right]/2\pi\I$ with $\lambda\in\mathbb{C}$. Willwacher in \cite{WillChains} proved another formality conjecture raised by Tsygan \cite{TsyganChains}. His elegant proof relies on a nice compatibility of Shoikhet's formality with the deRham differential and we briefly comment why we can also adapt his result and proof for the interpolation propagator.

Kontsevich's and Shoikhet's maps are defined with help of certain integral weights. This weights on the one hand satisfy a remarkable sequence of quadratic identities and also have an obvious invariance in their definition, on the other hand they are able to generate important sequences like the values of the Riemann $\zeta$ function evaluated at the positive integers (it is maybe interesting that the series $\zeta(2n)$ also satisfies a certain quadratic identity, see for example the proof of the rationality of $\zeta(2n)/\pi^{2n}$) in \cite{ZA}.

Unfortunately there are not many examples where such weights have been calculated and not so much is known about this integrals appearing in Kontsevich's $\star$ product, {\em e.g.} it is also not yet known whether they are all rational \cite{FW}, but there is some progress: In~\cite{VdB} the weights of wheels with spokes pointing outwards from the centre are computed explicitly. Furthermore, in~\cite[Chapter 7]{BCKT}, it is shown how to compute weights using the generalised Stokes Theorem for integration along the fiber. The results \cite{AWRT} and \cite{WR} show that $\mathcal{U}^{\ln}=\mathcal{U}^{1}$ has nicer number theoretic properties compared to $\mathcal{U}$, because $\mathcal{U}^{\ln}$ corresponds to the Knizhnik-Zamolodchikov associator \cite{Dri} and hence to multiple $\zeta$ values. Some pages written here are also concerned with the computation of this integrals and related number theoretic identities:

The interpolation weights are polynomials in $\lambda$. Because we prove the formality theorems for a complex interpolation parameter $\lambda$ the interpolation weights for a given graph are either constant or they have to admit roots for certain values of $\lambda$ by the fundamental theorem of algebra. This fact can be restated saying that there are two classes of weights: On the one hand the class of weights that are constant and do not depend on the interpolation parameter, on the other hand the class of weights that admit zeros for certain values of the interpolation parameter $\lambda\in\mathbb{C}$. Abusing physical language one could consider this eliminating of a certain graph as some kind of gauge fixing, this choice seems to be a slight advantage of our interpolation family of formality morphisms where we use a complex interpolation parameter.

We take care about the three types of vertices with valency $2$ for the interpolation. To summarise what happens when we eliminate a type I vertex with valency $2$ of a Kontsevich weight by performing a $2$-dimensional integral: If the vertex has one incoming and one departing arrow or two incoming arrows the integration vanishes. If the vertex has two departing arrows the targets couple, graphically
$$\xymatrix{ w_1&w\ar@/^1pc/[l]\ar@/^1pc/[r]&w_2\;\ar@{=>}[rr]^{\quad\int\hspace{-0.1cm}{{\d}w}{\overline{{\d}w}}\quad}\;&&{w_1}\ar@{<:>}[r]&{w_2}}$$
where the dashed double arrow represents the difference of two Kontsevich propagator functions and not an exterior derivative of a Kontsevich propagator function.

In the appendix we give an explicit computation by ``elementary'' methods, {\em i.e.} Stokes' Theorem and Residue Theorem, of the integral weight of the $2$-wheel graph appearing in the terms of order $2$ w.r.t.\ $\hbar$ in Kontsevich's $\star$ product \cite{K}.
The ``method" for the calculation of this special weight is just to try to apply Stokes theorem several times, this method is not very intuitive, technical integrals remain and it may not work for arbitrary weights. However for the calculation of arbitrary integral weights there is a universal, but less elegant and tedious recipe that we will demonstrate in the calculation of the Merkulov $n$-wheels. Merkulov showed in \cite{M} that the so-called characteristic class of the logarithmic formality map $\mathcal{U}^{\ln}$ is given by
$\exp\left(\sum_{n=2}^{\infty}\frac{\zeta(n)}{n}\left(\frac{x}{2\pi\I}\right)^n\right)$.
Here $\zeta(n)$ denotes the Riemann zeta function evaluated at positive integers $2\leq{n}$ and the $n$-wheel graph corresponding to the weight $\zeta({n})$ is pictured below
\vspace*{+0.3cm} \begin{align}\label{Wheel}
\SelectTips{cm}{}
\xymatrix{
w_1\ar@/^1.2pc/[rr]&&w_2\ar@/^1.2pc/[dd]\\
&w\ar@{->}[ul]\ar@{->}[dl]\ar@{->}[ur]\ar@{->}[dr]\\
w_n\ar@/^1.2pc/[uu]&&w_3\ar@/^1.2pc/[ll]_{\cdots}\\
}
\end{align}
In the appendix we also demonstrate how to compute this important wheel weights, but we do not fix the central vertex for the computation and calculate with the universal recipe. Our upgraded computation where we do not use the freedom to fix a vertex shows some rather involved number theoretic identities that correspond to the obvious invariance of the Kontsevich propagator under rescalings and real translations.\\

\chapter{Fedosov's quantization of symplectic manifolds}
In this section we sketch Fedosov's construction of $\star$ products on a symplectic manifold, the description slightly differs from Fedosov's original notation and construction \cite{FeS} and follows \cite{WaB}, \cite{D}. In the literature Fedosov's construction is sometimes referred to as the Gelfand-Fuchs trick of formal geometry or mixed resolutions.

Let in the following $(M,\omega)$ be a symplectic manifold of necessary even dimension $d$. Let $\mathcal{S}M$ denote the formally completed symmetric algebra of the cotangent space $T^{\ast}M$, defined as the bundle whose sections are infinite collections of symmetric covariant tensors $a_{i_1\cdots{i}_p}(x)$ where the indices $i_k$ run from $1$ to $d$ and $p$ runs from $0$ to $\infty$. Let  $\mathrm{v}^{i}$ be variables, which transform as contravariant vectors, we can think of $\mathrm{v}^{i}$ as formal coordinates of the fibers of $T^\ast{M}$ and write a section of $\Gamma(\mathcal{S}M)$ as the formal series
$$a(x,\mathrm{v})=\sum_{p=0}^{\infty}a_{i_1\cdots{i}_p}(x)\mathrm{v}^{i_1}\cdots\mathrm{v}^{i_p}$$
Next we tensor $\mathcal{S}M$ with the algebra of exterior forms to obtain a super-commutative algebra $\Omega^\bullet(M,\mathcal{S}M)$, called the formal Weyl algebra
$$\Omega^\bullet\left(M,\mathcal{S}M\right):=\Biggr\{\;a(x,\mathrm{v},{\d}x)=\sum_{0\leq{p,q}}a_{i_1\cdots{i}_pj_1...j_q}(x){\d}x^{j_1}\cdots{\d}x^{j_q}\mathrm{v}^{i_1}...\mathrm{v}^{i_p}\;\Biggr\}$$
where the $a_{i_1\cdots{i}_pj_1\cdots{j}_q}$ transform as tensors and are symmetric in the indices $i_1,\cdots,i_p$ and antisymmetric in the indices $j_1,\cdots,j_q$. Here and in the following we often will suppress in our notation wedge and tensor products and also always think of complexified bundles and formal series in $\hbar$.

The bundle $\Omega^\bullet(M,\mathcal{S}M)$ naturally admits the gradings $\deg_\mathrm{v}$, $\deg_a$ and $\deg_\hbar$ with respect to the symmetric $\mathrm{v}$ degree, the anti-symmetric $\d{x}$ form degree and the $\hbar$ degree.

Let the operator
$\delta:\Omega^{q}(M,\mathcal{S}M)\rightarrow\Omega^{q+1}(M,\mathcal{S}M)$
be defined by
$$\delta={\d}x^i\frac{\partial}{\partial{\mathrm{v}^i}}$$

We denote by ${i\left(\frac{\partial}{\partial{x}^{k}}\right)}$ the interior derivative of exterior forms by the vector field $\frac{\partial}{\partial{x}^{k}}$ and define the homotopy operator $\delta^{-1}:\Omega^{q}(M,\mathcal{S}M)\rightarrow\Omega^{q-1}(M,\mathcal{S}M)$ by the formula 
$$\delta^{-1}a=\mathrm{v}^{k}{i\left(\frac{\partial}{\partial{x}^{k}}\right)}\int_{0}^{1}\frac{{\d}t}{t}a(x,t\mathrm{v},t{\d}x)$$
Symmetry and grading arguments show that we have
$\delta^2=(\delta^{-1})^2=0$ and this clearly implies that $\delta$ and $\delta^{-1}$ are not invertible and the traditional notation $\delta^{-1}$ may be a bit misleading. We have the so-called \textit{Poincar\'{e}-Lemma}
$$\delta\delta^{-1}+\delta^{-1}\delta+\sigma=\id$$
with the projection $\sigma$ on the functions $C^\infty(M)$ given by
$$\sigma{a}=a\vert_{{\d}x^i=\mathrm{v}^i=0}$$
By this we can identify the zeroth cohomology
$H^0(\Omega(M,\mathcal{S}M),\delta)=C^\infty(M)$.

The space of sections $\Gamma(\mathcal{S}M)$ is naturally endowed with a commutative fiberwise product $\cdot$ induced by the multiplication of formal power series.
This fiberwise product and the wedge product induce in a natural way an associative graded product $\cdot$ that inherits the super-commutativity of the wedge product
$a_1\cdot{a_2}=(-1)^{k_1{k_2}}a_2\cdot{a_1}$
for $a_i\in\Omega^{k_i}\left(M,\mathcal{S}M\right)$, {\em i.e.} $\deg_a(a_i)=k_i$.

Similar to the usual Moyal-Weyl $\star$ product for $\mathbb{R}^d$ we can deform the product $\cdot$ to the so-called fiberwise Weyl product $\circ_{\Pi}$ with help of the commuting derivations $\frac{\partial}{\partial{\mathrm{v}^i}}$ by the following formula
$$a\circ_{\Pi}b:=\cdot{\exp}\bigr({\hbar}{P}\bigr)(a\otimes{b})$$
where $P$ is defined in local coordinates by
$${P}:=\frac{\I}{2}\Pi^{kl}\frac{\partial}{\partial{\mathrm{v}^k}}\otimes\frac{\partial}{\partial{\mathrm{v}^l}}\quad\text{or}\quad{P}:=2g^{k\overline{l}}\frac{\partial}{\partial{\mathrm{v}^k}}\otimes\frac{\partial}{\partial{\overline{\mathrm{v}}^l}}$$
with the symplectic Poisson-Tensor in the Weyl- or the symplectic K\"ahler-Tensor in the Wick-case \cite{BoWaWT} of K\"ahler manifolds respectively \cite{Gu}. This method also allows to directly construct star products for the Poisson structures $\frac{1}{2}\sum_{k,l}\pi^{kl}X_k\wedge{X}_l$ where $X_l,1\leq{l}\leq{n}$ are Lie commuting vector fields and $\pi\in{M}_n(\mathbb{R})$ is a skew-symmetric matrix \cite{NeDiss}.

The associative product $\circ_{\Pi}$ induces a super Lie bracket $[\cdot,\cdot]$ by 
$$[a_1,a_2]:=a_1\circ_{\Pi}a_2-(-1)^{k_1k_2}a_2\circ_{\Pi}a_1$$
In other words the previous defined bracket $[\cdot,\cdot]$ satisfies the graded Jacobi identity
$$[a_1,[a_2,a_3]]=[[a_1,a_2],a_3]+(-1)^{k_1k_2}[a_2,[a_1,a_3]]$$ The maps $[b,\cdot]$ also are inner super derivations of $\circ_{\Pi}$ of degree $\deg_a{b}$.

In the following it is convenient to introduce a degree combination that respects the specific structure of the Weyl product: The so-called total degree $Deg$ is defined by
$$Deg:=\deg_\mathrm{v}+2\deg_\hbar$$
and we denote by
$$\Omega^{\bullet}(M,\mathcal{S}M)=\Omega^{\bullet(0)}(M,\mathcal{S}M)\supseteq\Omega^{\bullet(1)}(M,\mathcal{S}M)\supseteq\Omega^{\bullet(2)}(M,\mathcal{S}M)\supseteq\cdots\supseteq\lbrace0\rbrace\;$$
with
$\bigcap_{d=0}^{\infty}\Omega^{\bullet(d)}\left(M,\mathcal{S}M\right)=\lbrace0\rbrace$
the  filtration corresponding to the total degree $Deg$.

It is a well-known fact that symplectic manifolds $(M,\omega)$ always admit torsion-free symplectic connections \cite{HSC}, {\em i.e.} connections that do respect the symplectic structure: We have a compatible connection $\nabla$ by the explicit formula
$$\omega(\nabla_X{Y},Z)=\omega(\nabla'_X{Y},Z)+(\nabla'_X\omega)(Y,Z)/3+(\nabla'_Y\omega)(X,Z)/3$$
where $\nabla'$ is any torsion-free connection. The operator $\nabla:\Omega^{\bullet}\left(M,\mathcal{S}M\right)\rightarrow\Omega^{\bullet+1}\left(M,\mathcal{S}M\right)$ is defined with a symplectic connection $\nabla$ by
$$\nabla={\d}x^i\frac{\partial}{\partial{x^i}}-\Gamma^{k}_{ij}{\d}x^i\mathrm{v}^{j}\frac{\partial}{\partial\mathrm{v}^{k}}$$
The operator $\nabla$ is a $\circ_\Pi$ super derivation of degree $1$, {\em i.e.} we have the graded Leibniz rule
$$\nabla(a_1\circ_\Pi{a_2})=\nabla{a}_1\circ_\Pi{a_2}+(-1)^{k_1}a_1\circ_\Pi{\nabla{a}_2}$$
The operators $\nabla$ and $\delta$ anti commute
$$\delta\nabla+\nabla\delta=0$$
Because symplectic Poisson tensors are by definition non-degenerate it is not a difficult exercise to verify that the square $\nabla^2$ can be written as an inner derivation $[R,\;.\;]:\Omega^{\bullet}\left(M,\mathcal{S}M\right)\rightarrow\Omega^{\bullet+2}\left(M,\mathcal{S}M\right)$ where we denote by
$$R=\frac{1}{4}{\omega}_{kr}R^{r}_{lij}\mathrm{v}^{k}\mathrm{v}^l{{\d}x}^{i}{{\d}x}^j\quad\text{or}\quad{R}=\frac{\I}{2}g_{k\overline{r}}R^{\overline{r}}_{\overline{l}i\overline{j}}\mathrm{v}^{k}\overline{\mathrm{v}}^l{\d}z^{i}{{\d}\overline{z}}^j$$
the curvature in the Weyl- or the Wick-Case respectively.

The Bianchi identities for the curvature of a torsion-free connection imply
$$\delta{R}=\nabla{R}=0$$


We call a differential ${\mathcal{D}}:\Omega^{\bullet}\left(M,\mathcal{S}M\right)\longrightarrow\Omega^{\bullet+1}\left(M,\mathcal{S}M\right)$ of the shape
$$\mathcal{D}=-\delta+\nabla+\frac{\I}{\hbar}[{\mathcal{R}},\;\cdot\;]$$
with ${\mathcal{R}}\in\Omega^{1(2)}\left(M,\mathcal{S}M\right)$ a Fedosov connection. 

With the $Deg$-adic topology $\Omega^{1(2)}\left(M,\mathcal{S}M\right)$ is a complete ultrametric space where we have
$$d(a,b)\leq{\max}(d(a,c),d(c,b))\;\forall\;a,b,c\in\Omega^{1(2)}\left(M,\mathcal{S}M\right)$$
and Banach's famous fixed point theorem holds for contraction mappings.

For every closed two form $\Omega\in\hbar\Gamma^\infty(\Lambda^{2}T^\ast{M})[[\hbar]]$ we have a flat
$\circ_{\Pi}$ super derivation
$$\mathcal{D}_{\Omega}=-\delta+\nabla+\frac{\I}{\hbar}[\mathcal{R}_{\Omega},\;\cdot\;]$$
The existence of a unique solution ${\mathcal{R}}\in\Omega^{1(2)}\left(M,\mathcal{S}M\right)$ of the fixed point equation
$$\mathcal{R}_\Omega=\delta^{-1}(\Omega+R+\nabla\mathcal{R}_\Omega+\frac{\I}{\hbar}\mathcal{R}_\Omega\circ_\Pi\mathcal{R}_\Omega)$$
and the normalisation $\delta^{-1}\mathcal{R}_{\Omega}=0$
is due to Fedosov. For a interpretation of the condition $\delta^{-1}\mathcal{R}_{\Omega}=0$ we refer the reader to the article \cite{NeDiss}.

Moreover we can write
\begin{equation}\label{Quadratic}
{\mathcal{R}}_{\Omega}=\sum_{n=1}^{\infty}\,\sum_{{p}=1}^{{C_{n}}}\left(\frac{1}{1-\delta^{-1}\nabla}\delta^{-1}(\Omega+R)\right)^{\left(\frac{\I}{2\hbar}\frac{1}{1-\delta^{-1}\nabla}\delta^{-1}[\;\cdot\;,\;\cdot\;]\right){n}}_{p}
\end{equation}
where we sum over all
$$C_{n}=\frac{1}{n}\binom{2n-2}{n-1}\;\forall\;{n}\in\mathbb{N}^+$$
{\em a priori} different parenthesis of $n$ times the element
$$\frac{1}{1-\delta^{-1}\nabla}\delta^{-1}(\Omega+R)\in\Omega^{1(2)}\left(M,\mathcal{S}M\right)$$
combined with $n-1$ non-associative, but in this case commutative binary operations
$${\frac{\I}{2\hbar}\frac{1}{1-\delta^{-1}\nabla}\delta^{-1}[\;\cdot\;,\;\cdot\;]}:\Omega^{1(i)}\left(M,\mathcal{S}M\right)\times\Omega^{1(j)}\left(M,\mathcal{S}M\right)\rightarrow\Omega^{1(i+j-1)}\left(M,\mathcal{S}M\right)$$
The validity of the formula for $\mathcal{R}_\Omega$ involving the Catalan numbers $$C_{n}:=\frac{1}{n!}\frac{\partial^{n}}{\partial{z}^{n}}\frac{1-\sqrt{1-4z}}{2}\Big\vert_0$$ was proven by the author in \cite{J}. We comment a bit how one can rephrase this formula: In the language of trees parenthesis correspond to rooted full binary trees with $n$ leaves, for example the five trees for $n=4$ are pictured below: First the most symmetric graph
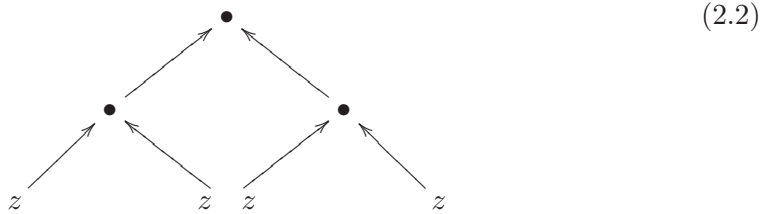
\begin{equation}\label{CGraph1}
\hspace{1.7cm}\xymatrix{&&{\bullet}&&&\\
&{\bullet}\ar[ur]&&{\bullet}\ar[ul]&&\\
{z}\ar[ur]&&{z}\ar[ur]\quad{z}\ar[ul]&&{z}\ar[ul]\\
}
\end{equation}
and second the four graphs
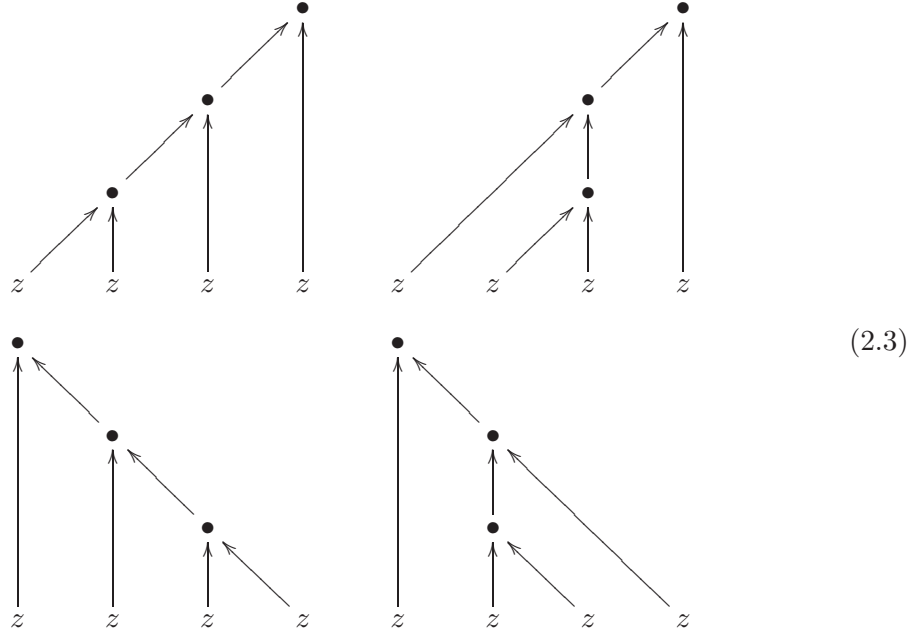
\begin{equation*}\label{CGraph2}
\xymatrix{&&&\bullet&&&&\bullet\\&&\bullet\ar[ur]&&&&\bullet\ar[ur]&\\&\bullet\ar[ur]&&&&&\bullet\ar[u]&\\{z}\ar[ur]&{z}\ar[u]&{z}\ar[uu]&{z}\ar[uuu]&{z}\ar[uurr]&{z}\ar[ur]&{z}\ar[u]&{z}\ar[uuu]\\
}
\end{equation*}
\begin{equation}\label{CGraph2}
\xymatrix{\bullet&&&&\bullet&&&\\&\bullet\ar[ul]&&&&\bullet\ar[ul]&&\\&&\bullet\ar[ul]&&&\bullet\ar[u]&&\\{z}\ar[uuu]&{z}\ar[uu]&{z}\ar[u]&{z}\ar[ul]&{z}\ar[uuu]&{z}\ar[u]&{z}\ar[ul]&{z}\ar[uull]\\
}
\end{equation}
where $z$ stands for $z=\frac{1}{1-\delta^{-1}\nabla}\delta^{-1}(\Omega+R)$ and $\bullet$ for ${\frac{\I}{2\hbar}\frac{1}{1-\delta^{-1}\nabla}\delta^{-1}[\;\cdot\;,\;\cdot\;]}:\Omega^{1(i)}\left(M,\mathcal{S}M\right)\times\Omega^{1(j)}\left(M,\mathcal{S}M\right)\rightarrow\Omega^{1(i+j-1)}\left(M,\mathcal{S}M\right)$. Notice that we plug in the parenthesis tree always the same argument $z$, hence because of commutativity the four graphs pictured in \ref{CGraph2} will clearly contribute with the same term in the sum \ref{Quadratic} over all parenthesis. However, the general reduction of the computation to the set of graphs obtained by dividing out commutativity seems to be a rather deep question.

The endomorphism $\mathcal{D}_{\Omega}^{-1}:\Omega^{\bullet}\left(M,\mathcal{S}M\right)\rightarrow\Omega^{\bullet-1}\left(M,\mathcal{S}M\right)$ defined by
$${\mathcal{D}_{\Omega}}^{-1}\;:=-\frac{1}{1-\delta^{-1}(\nabla+\frac{\I}{\hbar}[{\mathcal{R}}_{\Omega},\;\cdot\;])}\delta^{-1}$$
is well-defined in the $Deg$-adic topology. Also the quantum Fedosov-Taylor expansion
$$\tau_{\Omega}:{C}^{\infty}(M)[[\hbar]]\rightarrow\Omega^{0}\left(M,\mathcal{S}M\right)\cap{\ker}\;\mathcal{D}_{\Omega}$$
is well-defined by the formula
$$\tau_{\Omega}:=\frac{1}{1-\delta^{-1}(\nabla+\frac{\I}{\hbar}[{\mathcal{R}}_{\Omega},\;\cdot\;])}$$
The two previous maps are well-defined since $\delta^{-1}(\nabla+\frac{\I}{\hbar}[{\mathcal{R}}_{\Omega},\;.\;])$ is a linear contraction mapping in the $Deg$-adic topology. The Fedosov Taylor expansion can also be understood as the unique fixed point of the fixed point equation $$\tau_\Omega(f)=f+\delta^{-1}(\nabla+\frac{\I}{\hbar}[{\mathcal{R}}_{\Omega}\;\cdot\;])\tau_\Omega(f)$$

The previous defined maps satisfy the deformed \textit{Poincar\'{e}-Lemma}
$${\mathcal{D}_{\Omega}}^{-1}\mathcal{D}_{\Omega}+\mathcal{D}_{\Omega}{\mathcal{D}_{\Omega}}^{-1}+\tau_\Omega\sigma=id$$
Hence the linear map
$$\tau_{\Omega}\hspace{-0.05cm}:\hspace{-0.05cm}{C}^{\infty}(M)[[\hbar]]\hspace{-0.05cm}\rightarrow\hspace{-0.05cm}\Omega^{0}\left(M,\mathcal{S}M\right)\cap{\ker}\mathcal{D}_{\Omega}$$ is an isomorphism with inverse $\sigma$.
\begin{theorem}
For every closed two form $\Omega\in\hbar\Gamma^\infty(\Lambda^{2}T^\ast{M})[[\hbar]]$ on a symplectic manifold $(M,\omega)$ the Fedosov $\star$ product
$$f\star_{\Omega}{g}:=\sigma\left(\tau_{\Omega}(f)\circ_{\Pi}\tau_{\Omega}(g)\right)$$
is an associative $\star$ product. Two Fedosov products $\star_\Omega$ and $\star_{\Omega'}$ on $M$ are equivalent if and only if
$$[\Omega]=[\Omega']\in\hbar{H}^2_{dR}(M,\mathbb{C})[[\hbar]]$$
\end{theorem}
It has been proved that on symplectic manifolds Fedosov's construction yields $\star$ products of all equivalence classes, the original proof of this classification result is due to Gutt and Lichnerowicz. One can also show that the Fedosov $\star$ product is hermitian if $\Omega\in\hbar\Gamma^\infty(\Lambda^{2}T^\ast{M})[[\hbar]]$ is real \cite{WaB}.

Regular Poisson manifolds equipped with a connection compatible with the Poisson tensor can be quantized via Fedosov as well, but the lack of compatible connections for arbitrary Poisson tensors is a problem for the construction to apply. Although we cannot apply the construction directly all known constructions of globalising the local Kontsevich $\star$ product or the Kontsevich formality map use in some sense Fedosov constructions as a key stone in the globalisation of the local result and we refer the reader to \cite{D}, Dolgushev's elegant approach to global formality with more references. There is also an even more Fedosov like globalisation procedure for $\star$ products on any Poisson manifold due to Cattaneo, Felder and Tomassini \cite{CFT}.

Bezrukavnikov and Kaledin showed that the Fedosov quantization can also work under some assumptions in the algebraic setting \cite{Ka1} and in positive characteristic \cite{Ka2}.
\section{Example: Fedosov quantization of $\mathbb{D}^n$ and ${\mathbb{C}}{\mathbb{P}}^n$}
The following example is a result of the author's diploma thesis. Here a lot of the terms appearing in the general formula vanish. We refer the reader to \cite{J} for the detailed proof where we use the adapted Fedosov quantization of K\"ahler manifolds as described in \cite{BoWaWT} and \cite{NeDiss}:

It is well-known \cite{KaN} that the curvature tensor of a simply connected K\"ahler manifold $M$ of constant holomorphic sectional curvature $\mathrm{C}$ satisfies
$$R_{k\overline{l}i\overline{j}}=-\mathrm{C}(g_{k\overline{l}}g_{i\overline{j}}+g_{k\overline{j}}g_{\overline{l}i})$$
and $M$ is either holomorphically isometric to ${\mathbb{C}}^n$, $\mathbb{D}^n$ or to the complex projective space
${\mathbb{C}}{\mathbb{P}}^n$ depending on $\mathrm{C}=0$, $\mathrm{C}<0$ or $\mathrm{C}>0$ respectively. From this formula it is clear that the curvature in this case is covariantly constant, {\em i.e.}
$\nabla_X{R}=0\quad\forall{X}\in\Gamma^{\infty}(TM)$.
By calculation for any formal power series $F(\hbar)\in\mathbb{C}[[\hbar]]$ the Fedosov differential resulting out of the construction for
$\Omega=-4\hbar{F}(\hbar)\omega$
where $\omega$ denotes the symplectic form defined with the complex structure $J$ by $\omega(x,y)=g(x,Jy)$, is
\begin{equation}\label{FDCC}\mathcal{D}=\nabla-\frac{1}{2\hbar}\Biggr[\sum_{n=0}^{\infty}(-1)^{n}\Biggr(\frac{\sum_{l=0}^{n}(-1)^{l}\binom{n}{l}\sqrt{1-4\hbar(F(\hbar)+{l}\mathrm{C}})}{{n!}(2\hbar)^{n}}\Biggr)g^{n}\rho,\quad\;\cdot\quad\;\Biggr]\;\end{equation}
where $\rho=g_{i\overline{j}}(\mathrm{v}^{i}\overline{{\d}z}^j-\overline{\mathrm{v}}^j{\d}z^{i})$ and
$g^{n}\rho=\underbrace{g...{g}}_{n}\cdot\rho$
with $g=g_{i\overline{j}}\mathrm{v}^i\overline{\mathrm{v}}^j$.

We split $\nabla$ and $\delta^{-1}$ in its purely holomorphic and purely anti-holomorphic part and as proved in \cite{NeDiss} we only need the purely holomorphic part $\pi_{z}\tau$ and the purely anti-holomorphic part $\pi_{\overline{z}}\tau$ of the Fedosov taylor series $\tau$ to construct the $\star$ products of Wick type. With \ref{FDCC} we yield for the projections $\pi_{z}\tau$ and $\pi_{\overline{z}}\tau$ the equations
$$\pi_{z}\tau=\sum_{n=0}^{\infty}\Biggr(\prod_{l=0}^{n-1}\frac{1}{\sqrt{1-4\hbar(F(\hbar)+{l}\mathrm{C})}}\Biggr)(\delta_{z}^{-1}\nabla_{z})^{n}$$
$$\pi_{\overline{z}}\tau_=\sum_{n=0}^{\infty}\Biggr(\prod_{l=0}^{n-1}\frac{1}{\sqrt{1-4\hbar(F(\hbar)+{l}\mathrm{C}})}\Biggr)(\delta_{\overline{z}}^{-1}\nabla_{\overline{z}})^{n}\;$$
Finally with this formulas it is quite obvious that on spaces with constant sectional curvature $\mathrm{C}$ the Fedosov Wick type $\star$ product for $\Omega=-4\hbar{F}(\hbar)\omega$ can be written as
$$f\star{g}=\sum_{n=0}^{\infty}\frac{(2\hbar)^n}{n!}\Biggr(\prod_{l=0}^{n-1}\frac{1}{1-4\hbar(F(\hbar)+{l}\mathrm{C})}\Biggr)\cdot{P^{n}}\Bigr((\delta_{z}^{-1}\nabla_{z})^{n}f\otimes(\delta_{\overline{z}}^{-1}\nabla_{\overline{z}})^{n}g\Bigr)$$

\section{Fedosov Taylor expansion and the $\exp$ map}
The works \cite{W}, \cite{CF} and  \cite{G} pointed out the connection between formal exponential maps and Fedosov Taylor expansions of coordinate functions. More precisely the classical Fedosov Taylor expansions of local coordinate functions generate the exponential map. We will show a direct approach to a formal solution that to the best of our knowledge is missing in the literature. By formal we mean that we assume that the solution of the geodesic equations in our local coordinates is analytic in the time variable. Our proof just uses self-consistency and the inductive step provides the calculus. We consider this formula from a constructive point of view and will also demonstrate how the general construction works by considering two generic examples. 

Let $M$ be an analytic manifold and let us denote by $\nabla$ a torsion-free connection. The geodesic equations are a system of differential equations, in local coordinates they read
$$\frac{{\d}^2{\phi}}{{\d}t^2}^i+\Gamma^{i}_{c_1{c_2}}({X})\frac{{\d}{\phi}}{{\d}t}^{c_1}\frac{{\d}{\phi}}{{\d}t}^{c_2}=0\quad\forall{i=1,...,\dim{M}}$$
Equivalent the curve $\gamma(t)$ corresponding to $\phi(t)$ satisfies 
the parallel transport equation
$\nabla_{\gamma'(t)}\gamma'(t)=0$
and geodesics are the generalisation of the notion ``{straight line}" to curved spaces. If
$\Gamma^{i}_{k{l}}=\frac{1}{2}g^{im}\left(\partial_l{g}_{km}+\partial_k{g}_{lm}-\partial_m{g}_{kl}\right)$
are the Christoffel symbols of the Levi-Civita connection $\nabla$ of a Riemann metric $g$ this equations are the Euler-Lagrange equations
$\frac{{\d}}{{\d}t}\frac{\partial{L}}{\partial{\dot{x}^{i}}}-\frac{\partial{L}}{\partial{x^{i}}}=0\quad\forall{i=1,\cdots,\dim{M}}$
corresponding to the variation problem
$\delta\int{{\d}{t}{g_{c_1{c_2}}}\bigr({\phi(t)}\bigr)\frac{{\d}{\phi}}{{\d}t}^{c_1}(t)\frac{{\d}{\phi}}{{\d}t}^{c_2}(t)}\stackrel{!}{=}0$.
In general relativity a fundamental postulate states that objects in free motion move along geodesics in a space time determined by the Einstein-Hilbert equations.

Its well-known that we can do a reduction of the geodesic equations to a first order system of differential equations. In this reduced situation, under mild conditions on the Christoffel symbols, the Picard-Lindel\"of theorem states locally the existence of a unique solution ${\phi}(t)$ for small $t$ with
${\phi}^i(0)=x^i$ and $\frac{d{\phi^i}(0)}{dt}=\mathrm{v}^i$ and uniqueness allows to glue together geodesics in an overlap of two coordinate charts.

Let us denote by
${\phi}(x,\mathrm{v}):={\phi}(x,\mathrm{v},1):{U}\rightarrow{M}$ the formal exponential map, where $U$ is a neighbourhood of the zero section of the tangent bundle ${T_x{M}}$, by formal we again mean the assumption that this map is analytic in $\mathrm{v}$. A few steps and the Picard-Lindel\"of integral iteration gets difficult to handle in general because one has to integrate the Christoffel symbols.

It is well-known that in Riemann normal coordinates the solution of the geodesic equations are just straight lines. Quite obviously the Taylor expansion should only depend on the start point and start velocity because of the Picard-Lindel\"of theorem, more precise the solutions of the geodesic equations satisfy
${\phi}(x,\mathrm{v},\lambda{t})={\phi}(x,\lambda\mathrm{v},t)$
because of the invariance of the geodesic equations under affine reparametrisations. 

Of course the Christoffel symbols are not tensors, but the geodesic equations provide some sort of tensor calculus that does not care about the bound variable in the geodesic equations: Let on smooth functions with a finite number of upper and lower indices for variables $\in\lbrace{1,\cdots,\dim{M}}\rbrace$ of the form $a^{i}_{{v_n}\cdots{v_1}}$ a lower index raising operation $\nabla$ be defined by
$$\nabla_{v_{n+1}}a^{(i)}_{{v_n}\cdots{v_1}}:=\partial_{v_{n+1}}a^{(i)}_{{v_n}\cdots{v_1}}-\sum_{l=1}^{n}\Gamma^{c}_{{v_{n+1}}v_{l}}a^{(i)}_{{v_n}\cdots{v_{l+1}}c{v_{l-1}}\cdots{v_1}}$$
with the usual covariant tensor calculus notation in the lower index, but without covariant differentiations with respect to the upper index denoted by $(i)$ and as usual no respect to indices of Einstein contractions. With this notation its possible to write $\phi$ as a contraction with the by iteration increased number of lower indices, to the best of our knowledge this formula never before appeared in the literature:
\begin{proposition}\label{TH}
Let $\nabla$ be a torsion-free connection and $\phi$ be the formal exponential map. We have
\begin{align*}{\phi}^{i}(x,\mathrm{v})=&\;x^{i}+\mathrm{v}^{i}-\sum_{n=0}^{\infty}\frac{1}{(n+2)!}\hspace{+0.00cm}\nabla^{{n}}_{c_{n+2}\cdots{c}_{3}}\Gamma^{(i)}_{c_{2}c_{1}}\hspace{-0.0cm}(x)\hspace{+0.15cm}{{\mathrm{v}^{c_{n+2}}}\hspace{-0.2cm}\dots{\mathrm{v}^{c_1}}}\end{align*}
for $\forall{i=1,\cdots,\dim{M}}$.
\end{proposition}
\begin{proof}[Proof:] We assume that the unique solution of the geodesic equations is analytic in the time variable. We have
$$\frac{{\d}^3{\phi}}{{\d}t^3}^i=-\nabla_{c_3}\Gamma^{(i)}_{{c_2}{c_1}}\frac{{\d}{\phi}}{{\d}t}^{c_3}\frac{{\d}{\phi}}{{\d}t}^{c_2}\frac{{\d}{\phi}}{{\d}t}^{c_1}$$
$\forall{i=1,\cdots,\dim{M}}$ and the claim is that 
$$\frac{{\d}^n{\phi}}{{\d}t^n}^i=-\nabla_{{c_n}\cdots{c_3}}^{{n-2}}\Gamma^{(i)}_{{c_2}{c_1}}\frac{{\d}{\phi}}{{\d}t}^{c_n}\hspace{-0.2cm}\dots\frac{{\d}{\phi}}{{\d}t}^{c_1}$$
$\forall{i=1,\cdots,\dim{M}}$ if $n>1$ holds.

The inductive step is analogous to the calculation for $n=3$ and the notation in proposition \ref{TH}, where we respect the order of $\nabla$ operations, is motivated by this inductive step:
\begin{align*}\frac{{\d}^{n+1}{\phi}}{{\d}t^{n+1}}^i=-\Biggr(&\partial_{c_{n+1}}\nabla_{{c_n}\cdots{c_1}}^{{n}}\Gamma^{(i)}_{{c_2}{c_1}}-\sum_{l=3}^{n}\Gamma^{c}_{{c_{n+1}}c_l}\nabla_{{c_n}\cdots{c_{l+1}}c{c_{l-1}}\cdots{c_3}}^{{n}}\Gamma_{c_2c_1}^{(i)}\\&-\Gamma^{c}_{c_{n+1}c_2}\nabla_{{c_n}\cdots{c_3}}^{{n}}\Gamma_{cc_1}^{(i)}-\Gamma^{c}_{c_{n+1}c_1}\nabla_{{c_n}\cdots{c_3}}^{{n}}\Gamma_{c_2c}^{(i)}\Biggr)\frac{{\d}{\phi}}{{\d}t}^{c_{n+1}}\hspace{-0.4cm}\dots\frac{{\d}{\phi}}{{\d}t}^{c_1}\end{align*}
for $\forall{i=1,\cdots,\dim{M}}$, because self-consistency allows to substitute all second order derivatives by first order derivatives and Christoffel symbols and we can also relabel contraction indices. Notice that the proof has an analogy with \cite{MA} and\end{proof}
Notice only the totally symmetric part
$$\frac{1}{(n+2)!}\sum_{\sigma\in{S_{n+2}}}\nabla_{c_{\sigma(n+2)}\cdots{c}_{\sigma(3)}}^{{n}}\Gamma^{(i)}_{c_{\sigma(2)}c_{\sigma(1)}}$$
contributes, because of the contractions with the velocities and relabelling.

If one rewrites the equations as a Hamiltonian flow on the tangent bundle the method of Lie transforms works quite similar, as pointed out to us by A. Weinstein.

As proved in \cite{W} the classical Fedosov Taylor expansions of local coordinate functions $x^i$ regenerate the exponential map. This fact was stated there without a formula for geodesics and the proof in \cite{W} is quite involved. We will show this here in a different way by computation, the correspondence with \ref{TH} finally yields the identification of the Fedosov taylor expansion of local coordinate functions with the exponential map:

For the definition of the classical Fedosov-Dolgushev Taylor expansion we refer the reader to \cite{W}. The computation of the classical Fedosov Taylor expansion is
\begin{align*}\tau(x^i)=&\sum_{n=0}^{\infty}\frac{1}{n!}\nabla_{c_n\cdots{c}_1}(x^i)\mathrm{v}^{c_n}\cdots\mathrm{v}^{c_1}\\&=x^i+\mathrm{v}^i-\sum_{n=0}^{\infty}\frac{1}{(n+2)!}\nabla_{c_{n+2}\cdots{c}_3}\Gamma^{(i)}_{c_2c_1}\mathrm{v}^{c_{n+2}}\cdots\mathrm{v}^{c_1}\end{align*}
The first equals in this computation is due to a formula contained in \cite{MA} for the classical Fedosov-Dolgushev Taylor expansion of functions \cite{D}, \cite{W}. The formula in \cite{MA} uses an analog of the proof of proposition \ref{TH} and now we can argue also as follows: The crucial difference between the classical and the quantum Fedosov Taylor expansions is that the classical Taylor series are only derivations of the natural product $\cdot$ induced by formal power series multiplication and the wedge product. As usual we define a derivation $\delta^{\ast}$ of the natural product $\cdot$ by
$\delta^{\ast}:=\mathrm{v}^{i}{i_{a}}(\partial_{i})$ and rewrite
$$\delta^{-1}\alpha:=\begin{cases}\frac{1}{s+k}\delta^{\ast}\alpha&\text{if}\;\deg_{s}\alpha=s,\;\deg_{a}\alpha=k\;\text{and}\;s+k\neq0\\0&\text{if}\;\deg_{a}\alpha=0\end{cases}$$
We obviously have $\delta^{-1}\tau(f)=0$ for a function $f\in{C}^\infty(M)$ and because we have put the normalization condition $\delta^{-1}\mathcal{R}=0$ the term containing $\delta^{-1}\mathcal{R}\tau(f)$ cancels in the fixed point equation
$\tau(f)=f+\delta^{-1}\left(\nabla+{\mathcal{R}}\right)\tau(f)$
{\em i.e.} we have
$\tau(f)=\left(1-\delta^{-1}\nabla\right)^{-1}\tau(f)$.
Here the reader should be aware that this formula is not true for the quantum Fedosov Taylor series for symplectic manifolds because $\delta^{\ast}$ is not a derivation of $\circ_\Pi$. The second equality is just a computation of the first orders of the locally defined expansion in the special case of coordinate functions. By comparison we see that the classical Fedosov Taylor expansion of coordinate functions coincides with the formula for the exponential map, {\em i.e.}
$$\tau(x^i)=\phi^i(x,\mathrm{v})$$
\section{Example calculations of some geodesics}
 On the one hand in specific examples the choice of coordinates is of course essential to make the formula computable and on the other hand the construction is tautological if one computes in geodesic normal coordinates.

Unlike the other calculations we found in the literature a construction of geodesics with theorem  \ref{TH} does not essentially use integration of combinations of the equations, this seems to be a slight advantage of our in some sense more recursive approximation: In general it is not {\em a priori} clear how to use integration methods, if one can presuppose invertible relations and it needs intuition to arrange the equations in a promisingly way, compute the integrals and set integration constants. 

\subsection*{$2$-sphere} We choose spherical coordinates  $x^1:=\theta$ an $x^2:=\varphi$. Explicit the metric coefficients and Christoffel symbols are given by the formulas
$$g_{kl}=\delta_{k}^{1}\delta_{l}^{1}+sin^2(\theta)\delta_{k}^{2}\delta_{l}^{2}$$
$$\Gamma^{k}_{ij}=-\sin(\theta)\cos(\theta)\delta_{1}^{k}\delta_{j}^{2}\delta_{i}^{2}+\frac{\cos(\theta)}{\sin(\theta)}\delta_{2}^{k}\left(\delta_{i}^{2}\delta_{j}^{1}+\delta_{j}^{2}\delta_{i}^{1}\right)$$
Hence the only Christoffel symbols that do not vanish identically are
$$\Gamma^{1}_{22}(\theta)=-\sin(\theta)\quad\text{and}\quad\Gamma^{2}_{12}=\Gamma^{2}_{21}(\theta)=\frac{\cos(\theta)}{\sin(\theta)}$$

By induction because of $\Gamma^{1}_{11}=\Gamma^{2}_{11}=0$ we have
$\nabla_1^{{n+1}}\Gamma^{(i)}_{11}=\partial_1\left(\nabla^{{n}}\Gamma\right)^{(i)}_{1...1}=0$.
With this result proposition \ref{TH} for $\phi^1(0)=\theta$, $\phi^2(0)=\varphi$ and $\mathrm{v}_1=1$, $\mathrm{v}_2=0$ implies the time evolution
$\phi^2(t)=\varphi$ and ${\phi^1}(t)=\theta+t$
and the geometric shape of all geodesics is clear in this example by symmetry.

\subsection*{Poincar\'e half-plane}
We choose Cartesian coordinates $x_1$ an $x_2$. The metric and Christoffel symbols are
$$g_{kl}=\left({1}/{x_2}\right)^2(\delta_{k}^{1}\delta_{l}^{1}+\delta_{k}^{2}\delta_{l}^{2})$$
$$\Gamma^{k}_{ij}=\frac{-1}{x_2}\Bigr(\delta_{2}^{k}\left(-\delta_{j}^{1}\delta_{i}^{1}+\delta_{j}^{2}\delta_{i}^{2}\right)+\delta_{1}^{k}(\delta_{i}^{1}\delta_{j}^{2}+\delta_{j}^{1}\delta_{i}^{2})\Bigr)$$
Hence the only Christoffel symbols that do not vanish identically are
$$\Gamma^{1}_{21}=\Gamma^{1}_{12}=\Gamma^{2}_{22}=-\Gamma^{2}_{11}=-{1}/{x_2}$$

The Cartesian coordinates are a good choice to calculate the vertical, with metric unit speed at $\phi^1(0)=x_1$, $\phi^2(0)=x_2$ starting geodesics. In this notation this vertical emanating geodesics correspond to $\mathrm{v}_1=0$, $\mathrm{v}_2={x_2}$. By induction, because of $\Gamma^{1}_{22}=0$ we have
$\nabla_2^{{n+1}}\Gamma^{(1)}_{22}=\partial_2\left(\nabla_2^{{n}}\Gamma^{(1)}_{22}\right)-(n+2)\Gamma^{2}_{{{2}}{2}}\nabla_2^{{n}}\Gamma
^{(1)}_{22}=0$.
This implies that
$\phi^1(t)=x_1$
is constant. The same argument shows that we have the recursion
$$\Gamma^{{2}}_{2{2}}=-\frac{1}{x_2}\quad\text{and}\quad\nabla_2^{{n+1}}\Gamma^{(2)}_{2\cdots2}=\partial_2\left(\nabla_2^{{n}}\Gamma^{(2)}_{22}\right)-(n+2)\Gamma^{2}_{{{2}}{2}}\nabla_2^{{n}}\Gamma^{(2)}_{22}$$
This recursion is solved by
$\nabla_2^{{n}}\Gamma_{22}^{(2)}=-{1}/{x_2}^{n+1}$
and because of the contractions with $\mathrm{v}_2^{n+2}=x_2^{n+2}$ proposition \ref{TH} leads to
$\phi^2(t)=x_2\exp(t)$
hence we get the unique speed parametrisation of the vertical geodesics.

The calculation of all geodesics in Cartesian coordinates is quite complicated. With polar coordinates instead it's not very difficult to figure out the geometric shape of horizontal geodesics, but the more conceptual approach seems to be the following \cite{Mi}: It is a well-known fact that local isometries map geodesics to geodesics. Moreover it is also well-known \cite{Mi} that for all $2\times2$ matrices with determinant $1$, {\em i.e.}
$$\left(\begin{matrix}a & b \\ c & d \end{matrix}\right)\in{SL(2,\mathbb{R}})$$
the map that sends $(x_1,x_2)$ to
$$\Biggr(\frac{ac\left(x_1^2+x_2^2\right)+(ad+bc)x_1+db}{\left(cx_1+d\right)^2+\left(cx_2\right)^2},\frac{(ad-bc)x_2}{\left(cx_1+d\right)^2+\left(cx_2\right)^2}\Biggr)$$
is a local isometry, in other words its pull back respects the metric tensor as a fixed point.

\chapter{\textit{Kontsevich's quantization of Poisson manifolds}}\label{GenQua}
In the following sections we outline the construction of $\star$ products that applies on any Poisson manifold due to Kontsevich. Kontsevich has constructed in~\cite{K} a universal formula for the deformation quantization of a general Poisson manifold, we proceed here in the same way and first discuss the local case of a general Poisson structure $\Pi$ on $X=\mathbb R^d$, or on an open subset thereof.

Graphs are an essential ingredient in Kontsevich's construction and the definition of the graphs we will need for Kontsevich's formality theorem is:
\begin{definition} An admissible graph $\Gamma\in\mathcal G_{n,m}$ is an oriented graph with labels such that
\begin{enumerate}
\item The set of vertices $V_\Gamma$ is $\{1,\cdots,n\}\sqcup\{\overline{1},\cdots\overline{m}\}$ with $n,m\in\mathbb{N}$ and $2n+2-m\geq0$. Vertices of the set $\{1,\cdots,n\}$ are called vertices of type I, vertices from $\{\overline{1},\cdots\overline{m}\}$ are called vertices of type II.
\item Every edge $(v_1,v_2)$ starts at a vertex of type I.
\item There are no short loops, {\em i.e.} no edges of the type $(v,v)$.
\item For every vertex $k\in\{1,\cdots,n\}$ of type I, the set of edges
$$\text{Star}(k)=\{(v_1,v_2)\in{E}_\Gamma\vert{v}_1=k\}$$
starting from $k$, is labeled by symbols $(e_k^1,\cdots,e_k^{\#\text{Star}(k)})$.
\end{enumerate}
\end{definition}
In the case $X=(\mathbb{R}^d,\Pi)$ Kontsevich's local formality theorem $[1]$ states as a corollary that we can define a $\star$ product by
$$f_1\star f_2=f_1 f_2+\sum_{n\geq 1}\frac{\hbar^n}{n!}B_n(f_1,f_2)$$
where
$B_n(\cdot,\cdot)=\sum_{\Gamma\in\mathcal G_{n,2}}w_\Gamma B_\Gamma(\cdot,\cdot)$
is a certain weighted sums over admissible graphs of type $(n,2)$. To an admissible graph $\Gamma$ of type $(n,2)$ with the property that every vertex of type I has exactly two departing edges, Kontsevich associates $i)$ a bidifferential operator $B_\Gamma(\cdot,\cdot)$ and $ii)$ an integral weight $w_\Gamma$ to define $B_n(\cdot,\cdot)=\sum_{\Gamma\in\mathcal G_{n,2}}w_\Gamma B_\Gamma(\cdot,\cdot)$ as the mentioned weighted sum. We will turn to the weights in the next section and give the reader a first intuitive definition of the operators in the Kontsevich $\star$ product:

There are (up to isomorphisms) three admissible graphs $(2,2)$ where every vertex of type I has exactly two departing edges, namely first the two graphs
\begin{align}\label{Graph1}
\SelectTips{cm}{}
\xymatrix{
*\txt{}&*\txt{}&z_1\ar@/_/[dl]\ar@/^/[d]&z_2\ar@/_/[dll]\ar@/^/[dl]\\
*\txt{$-\infty\;$}&r_1\ar[l]\ar@{-}[r]_{\substack{{B_\Gamma=}\\{\Pi^{i_1l_1}\Pi^{i_2l_2}\partial^2_{i_1i_2}\otimes\partial^2_{l_1l_2}}}}&r_2\ar[r]&*\txt{$\;+\infty$} }\quad\xymatrix{
*\txt{}&*\txt{}&z_1\ar@/_/[dl]\ar@/^/[d]&z_2\ar@/_/[l]\ar@/^/[dl]\\
*\txt{$-\infty\;$}&r_1\ar[l]\ar@{-}[r]_{\substack{{B_\Gamma=}\\{\partial_{i_2}\Pi^{i_1l_1}\Pi^{i_2l_2}\partial_{i_1}\otimes\partial^2_{l_1l_2}}}}&r_2\ar[r]&*\txt{$\;+\infty$}}\end{align}
where we have written the corresponding bidifferential operator under the previous graph and the following graph to give the reader a first intuitive definition of the operators in the Kontsevich $\star$ product \ref{KSTAR}. 
The bidifferential operators corresponding to this two graphs appear in the formula in~\cite[Subsubsection 1.4.2]{K}. The third wheel like graph of type $(2,2)$, that does not contribute to the formula in~\cite[Subsubsection 1.4.2]{K}, is
\begin{align}\label{Graph2}
\SelectTips{cm}{}
\hspace{-1cm}\xymatrix{*\txt{}&*\txt{}&z_1\ar@/_/[dl]\ar@/^/[r]&z_2\ar@/^/[l]\ar@/^/[dl]\\*\txt{$-\infty\;$}&r_1\ar[l]\ar@{-}[r]_{\substack{{B_\Gamma=}\\{\partial_{i_2}\Pi^{i_1l_1}\partial_{i_1}\Pi^{i_2l_2}\partial_{l_1}\otimes\partial_{l_2}}}}&r_2\ar[r]&*\txt{$\;+\infty$}}\end{align}
In the following we will explain how the integral weights are constructed and in the appendix we also give another proof of the quite well-known fact that the weight of the last pictured graph actually does not vanish. 
\section{The Kontsevich integral weights}
\subsection*{Manifolds with corners}
We need a brief memento of manifolds with corners: A compact $d$ dimensional manifold with corners $X$ is a compact differentiable manifold in the usual sense whose local charts are diffeomorphic to $U_{p,q}=({\mathbb{R}_+})^p\times\mathbb{R}^q$ where $\mathbb{R}_+=\{r\in{\mathbb{R}}, 0\leq{r}\}$ and  $0\leq{p,q}\in\mathbb{N}$ with $p+q=d$.

There is an action of $\mathfrak S_p\times\mathfrak S_q$ on $U_{p,q}$, where $\mathfrak S_n$ denotes the group of permutations of n elements and we consider transition diffeomorphisms $\phi$ of $U_{p,q}$ which preserve the boundary in the sense 
$\phi_i(x_1,\cdots,x_{i-1},0,x_{i+1},\cdots,x_d)=0$
for $i=1,\dots,p$. The set $U_{p,q}$ admits a boundary stratification into boundary strata of codimension $1,2,3,\cdots$. For instance the boundary of codimension $1$ is
$\partial{U}_{p,q}=\cup_{i=1}^{p}({\mathbb{R}_+})^{i-1}\times\{0\}\times({\mathbb{R}_+})^{p-i}\times\mathbb{R}^q$
and ${\mathbb{R}_+}^{i-1}\times\{0\}\times{\mathbb{R}_+}^{p-i}\times\mathbb{R}^q$ is up to permutations $U_{p-1,q}$. This stratification is preserved by transition isomorphisms and therefore the manifold $X$ admits a boundary stratification
$X\supseteq\partial{X}\supseteq\partial^2{X}\supseteq\cdots$
and the orientation on $X$ induces orientations on all boundary strata.

\subsection*{Compactified configuration spaces}
We consider the configuration space $C_{n,m}^+$, resp. $C^+_n$, of $n$ distinct points in $\mathbb H$ and $m$ distinct points on $\mathbb R$, resp.\ $n$ distinct points in $\mathbb C$, modulo the action of the group $\mathbb R^+\ltimes\mathbb R$, resp.\ $\mathbb R^+\ltimes\mathbb C$, acting component wise by rescalings and translations. In formulas we first denote the product of the configuration space of the upper half-plane $\mathbb H$ with the configuration space of the real line $\mathbb{R}$ by
$$\mathrm{Conf}_{n,m}\hspace*{-0.1cm}=\{(z_1,\cdots,z_n,r_1,\cdots,r_m){\mid}z_i\in\mathbb{H},r_k\in\mathbb{R},z_i\neq{z_j},r_i\neq{r_j}\;\text{if}\;i\neq{j}\}$$
By $\mathbb R^+\ltimes\mathbb R$ we denote the group of orientation-preserving transformations of the real line
$$\mathbb R^+\ltimes\mathbb R={\{z{\mapsto}p\cdot{z}+q\mid{p,q\in\mathbb{R}, p>0}}\}$$
The condition $2n+m\geq2$ ensures that the action of $\mathbb R^+\ltimes\mathbb R$ on $\mathrm{Conf}_{n,m}$ is free and we define ${C}_{n,m}$ as the quotient space
$${C}_{n,m}=\mathrm{Conf}_{n,m}\slash\mathbb R^+\ltimes\mathbb R$$

Respectively we define
$$\mathrm{Conf}_{n}\hspace*{-0.1cm}=\{(z_1,\cdots,z_n){\mid}z_i\in\mathbb{C},,z_i\neq{z_j}\;\text{if}\;i\neq{j}\}$$
and use the notation
$$\mathbb R^+\ltimes\mathbb C={\{z{\mapsto}p\cdot{z}+q\mid{p\in\mathbb{R}^+, q\in\mathbb C}}\}$$
to define ${C}_{n}$ as the quotient space
$${C}_{n}=\mathrm{Conf}_{n}\slash\mathbb R^+\ltimes\mathbb C$$

The two spaces $C_{n,m}$ and $C_n$ admit both compactifications $C_{n,m}^+$ and $ C_n^+$ {\em \`a} la Fulton-MacPherson: These are smooth, oriented manifolds with corners. $C_{n,m}^+$, resp.\ $C_n^+$, of dimension $2n+m-2$, resp.\ $2n-3$ (whence, we assume $2n+m-2\geq 0$, resp.\ $2n-3\geq 0$). Observe that $ C_2\cong S^1$, $ C_{0,2}^+\cong\{0,1\}$, $ C_{1,0}^+\cong \{i\}$ and $\mathcal C_{1,1}^+\cong [0,1]$.
We may as well consider $C_{A,B}^+$ and $C_A^+$ for $A$ a finite set and $B$ a finite, ordered set.

We are interested in the combinatorics of the boundary strata of codimension $1$ of the boundary stratification of $ C_{n,m}^+$, which are of two types:
\begin{itemize}\label{Boundary}
\item[S1)] There is a subset $A$ of $[n]=\{1,\dots,n\}$, $2\leq |A|\leq n$, such that the corresponding boundary stratum identifies with
$$ C_A\times C_{ \{\bullet\}\sqcup[n]\setminus{A},m}$$
\item[S2)] There are a subset $A$ of $[n]$ and an ordered subset $B$ of consecutive elements of $[m]$, $0\leq |A|\leq n$, $0\leq |B|\leq m$ and $2\leq\vert{A}+\vert{B}\leq{n}+m-1$ such that the corresponding boundary stratum identifies with
$$ C_{A,B}\times C_{[n]\setminus A,\{\bullet\}\sqcup [m]\setminus B}$$
\end{itemize}
We refer to~\cite[Section 5]{K} and to~\cite[Section I]{AMM} for details on the previous compactified configuration spaces; in particular, the second reference is quite helpful for what concerns orientations for the boundary strata and signs.

In the construction of the Kontsevich formality map we will give an explicit one, but it is convenient to first give an intuitive definition of the weights by saying that for example the weight of the graph \ref{Graph2} is defined by
$$w_\Gamma=\frac{1}{2!}\int_{\overline{C}_{2,2}}\hspace{-0.1cm}{{\d}\phi(z_1,z_2){\d}\phi(z_1,r_1){\d}\phi(z_2,z_1){\d}\phi(z_2,r_2)}$$
Here the Kontsevich propagator $\d\phi(s,t)$ is defined as the exterior derivative of a so-called angle map $\phi(\cdot,\cdot)$: The angle map $\phi(s,t)$ can be defined as the angle at the source $s$ formed by the vertical geodesic passing through $s$ and the geodesic passing through $s$ and the target $t$ where we consider the upper half plane endowed with the Lobachevsky metric and explicit we have
$$\phi(s,t):=\frac{1}{2\pi}\mathrm{\arg}\left(\frac{t-s}{t-\overline{s}}\right)$$
The angle map $\phi(\cdot,\cdot)$ is clearly invariant under the action of $\mathbb R^+\ltimes\mathbb R$, hence the Kontsevich propagator extends to $C_{n,m}^+$ and the Kontsevich integral weights are well-defined, for a detailed convergence discussion we refer to section \ref{Convergence}.
\section{The Kontsevich $L_\infty$-quasi-morphism for $\mathbb{R}^d$}
\subsection{$L_\infty$-algebras and -morphisms}\label{Linfty}
We need a brief memento of $L_\infty$-algebras.
For a field $\mathbb{K}$ of characteristic zero such that $\mathbb{C}\subset\mathbb{K}$ we denote by $A:=\mathbb{K}[\![x_1,\cdots,x_d]\!]$ the ring of formal power series in $d$ variables. By $T_\mathrm{poly}(A)$ we denote the algebra of totally skew-symmetric multi-derivations and by $D_\mathrm{poly}(A)$ we denote the algebra of multi-differential operators on $A$.

Let $V$ be a $\mathbb{Z}$-graded vector space over $\mathbb{K}$, $[\bullet]$ be the degree-shifting functor, $S(V)$ be the (graded) symmetric algebra and
$$S^+(V)=\oplus_{1\leq{n}}S^n(V)$$
be the cofree cocommutative, coassociative coalgebra without counit.

An $L_\infty$-algebra structure on an object $\mathfrak{g}$ of grVect consists of a coderivation $Q$ of degree $1$, which additionally
squares to $0$, on the cocommutative, cofree coassociative coalgebra without counit $S+(\mathfrak{g}[1])$.
The fact that $S^+(\mathfrak{g}[1])$ is a cofree, coassociative, cocommutative coalgebra implies that $Q$ is uniquely specified by
its Taylor components
$$Q_n:S^n(\mathfrak{g}[1])\rightarrow\mathfrak{g}[1]$$
$n\geq1$, via the assignment
\begin{equation}\label{LinftyRel}
Q(x_I)=\sum_{J\subseteq{I},\;\vert{J}\vert \geq1}\epsilon(J,I)Q_{\vert{J}\vert}(x_J)x_{I\setminus{J}}
\end{equation}
where $x_I$ is any monomial of degree $Ã¢\vert{I}\vert\geq1$ in $S^+(g[1])$, for some set of indices $I$, and for a subset $J$ of $I$ of cardinality bigger or equal than $1$ and the sign $\epsilon(J,I)$ in \ref{LinftyRel} is specified by the rule
$x_I =\epsilon(J,I)x_{J}x_{I\setminus{J}}$.

There is an isomorphism of vector spaces over $\mathbb{K}$ between $S^n(\mathfrak{g}[1])$ and $\wedge^n(\mathfrak{g})[n]$ (the \textit{d\'ecalage} isomorphism), through which the Taylor component $Q_n$ of $Q$ may be also regarded as a morphism of degree $2Ã¢ÂÂn$ from $\wedge^n(\mathfrak{g})$ to $\mathfrak{g}$.
The condition that $Q$ squares to $0$ is equivalent to an infinite family of quadratic identities between the Taylor
components of $Q$, {\em i.e.}
$$\sum_{J\subseteq{I},\;\vert{J}\vert\geq1}\epsilon(J,I)Q_{\vert{I}\vert-\vert{J}\vert+1}(Q_{\vert{J}\vert}(x_J)x_{I\setminus{J}})=0$$ for any choice $I$ of cardinality at least $1$.
The condition that $Q$ squares to zero in particular implies the first three equations
\begin{itemize}
\item ${Q}^2_1=0$, {\em i.e.} $Q_1$ is a differential on $\mathfrak{g}$.
\item $Q_2$ is a skew-symmetric bilinear operation on $\mathfrak{g}$, for which $Q_1$ satisfies the graded Leibniz rule.
\item $Q_2$ satisfies the Jacobi identity up to homotopy explicitly described by $Q_3$ with respect to $Q_1$. In other words, the cohomology of an $L_\infty$-algebra $\mathfrak{g}$ with respect to $Q_1$ has a structure of graded Lie algebra. 
\end{itemize}
This shows that the notion of $L_\infty$-algebras generalises the definition of a DGLA where we just have non-trivial Taylor coefficients $Q_n$ for $n=1,2$ but $Q_n=0$ for $n\geq3$.

Finally, given two $L_\infty$-algebras $(\mathfrak{g}_i,Q_i)$, a $L_\infty$-morphism $\mathcal{U}$ from $\mathfrak{g}_1$ to $\mathfrak{g}_2$ is defined as a morphism of degree $0$ of cocommutative, cofree, coassociative coalgebras without counits from $S^+(\mathfrak{g}_1[1])$ to $S^+(\mathfrak{g}_2[1])$, which additionally intertwines the corresponding codifferentials.
The fact that $\mathcal{U}$ is a morphisms of coalgebras from $S^+(\mathfrak{g}_1[1])$ to $S^+(\mathfrak{g}_1[1])$ and the cofreeness of such coalgebras implies that $\mathcal{U}$ is uniquely determined by its Taylor components
\begin{align*}&\mathcal{U}_1:\mathfrak{g}_1\rightarrow\mathfrak{g}_2\\&\mathcal{U}_2:\wedge^2\mathfrak{g}_1\rightarrow\mathfrak{g}_2[-1]\\&\mathcal{U}_3:\wedge^3\mathfrak{g}_1\rightarrow\mathfrak{g}_2[-2]\\&\cdots\end{align*}
Namely we have the identity
$$\mathcal{U}(x_I)=\sum_{ p\geq1,\;J_1\cup\cdots\cup{J_p}=I
,\vert{J_i}\vert\geq1\forall{i}=1,\cdots,p}\epsilon(J_1 , \cdots , J_p , I )\mathcal{U}_{\vert{J_1}\vert} (x_{J_1} )\cdots{\mathcal{U}}_{\vert{J_p}\vert} (x_{J_p} )$$
The condition that $\mathcal{U}$ intertwines the codifferentials is equivalent to an infinite family of polynomial identities involving the Taylor components of $Q_i$, $i=1,2$ and by the \textit{d\'ecalage} isomorphism we can identify an $L_\infty$-morphism also with a collection of maps
$$\mathcal{U}_n:S^n(\mathfrak{g}_1[1])\rightarrow\mathfrak{g}_2[1]$$
for $n\geq1$
that satisfies the identities
\begin{align}\label{LinftyMor}
&\sum\pm\mathcal{U}_{q+1}\left(Q_p(x_{i_1}\wedge\cdots\wedge{x}_{i_p})\wedge{x_{j_1}}\wedge\cdots\wedge{x_{j_q}})\right)\\&=\sum_{n_1+\cdots{n}_k=p+q}\pm{Q}_k\left(\mathcal{U}_{n_1}(x_{i^1_1}\wedge\cdots\wedge{x}_{i^1_{n_1}})\wedge\cdots\wedge\mathcal{U}_{n_1}(x_{i^k_1}\wedge\cdots\wedge{x}_{i^k_{n_k}})\right)\nonumber
\end{align}
We are mainly interested in the case when $\mathfrak{g}_1$ and $\mathfrak{g}_2$ are dg Lie algebras. Here the equations reduce to
\begin{align}&\d\mathcal{U}_n(\gamma_1\wedge\cdots\wedge\gamma_n)-\sum_{i=1}^{n}\pm\mathcal{U}_n(\gamma_1\wedge\cdots\wedge\d\gamma_i\wedge\cdots\wedge\gamma_n)\nonumber\\&=\frac{1}{2}\sum_{\substack{{k,l\geq1}\\{k+l=n}}}\frac{\pm1}{k!l!}\sum_{\sigma\in\sum_n}[\mathcal{U}_k(\gamma_{\sigma_1}\wedge\cdots\wedge\gamma_{\sigma_k}),\mathcal{U}_l(\gamma_{\sigma_{k+1}}\wedge\cdots\wedge\gamma_{\sigma_n})]\nonumber\\&+\sum_{i<j}\pm\mathcal{U}_{n-1}([\gamma_i,\gamma_j]\wedge\cdots\wedge\gamma_n)\nonumber\end{align} 
In particular the first two equations are
\begin{itemize}
\item $\d\mathcal{U}_1=\mathcal{U}_1\d$, {\em i.e.} $\mathcal{U}_1$ is a map of complexes. 
\item $\mathcal{U}_1$ is a map of dg Lie algebras modulo $\mathcal{U}_2$, hence $\mathcal{U}_1$ defines a map of graded Lie algebras on the level of cohomology.
\end{itemize}
The first equation motivates the following definition: An $L_\infty$-quasi-isomorphism $\mathcal{U}$ is a $L_\infty$-morphism, whose first Taylor coefficient $\mathcal{U}_1$ is an isomorphism on the corresponding cohomologies. One of the main properties of $L_\infty$-quasi-isomorphisms is that they define an equivalence relation on the set of $L_\infty$-algebras.
\subsection{The $\mathrm{HKR}$ map}
Let $M$ be a $d$-dimensional manifold. We denote by
$$T_{poly}(M)=\oplus_{k=-1}^{\infty}{\Gamma}^{\infty}({\wedge}^{k+1}TM)$$
the vector space of polyvector fields and by
$$D_{poly}(M)=\oplus_{k=-1}^{\infty}\lbrace{\text{Polydifferential operators}:{C}^{\infty}(M)^{\otimes{k+1}}\rightarrow{C}^{\infty}(M)\rbrace}$$
the vector space of polydifferential operators on $M$. This two bundles are naturally endowed with DGLA structures, induced by the Schouten-Nijenhuis bracket $[\;.\;,\;.\;]_{\mathrm{S-N}}$ and zero differential in $T_{poly}(M)$ and the Gerstenhaber bracket $[\cdot,\cdot]_{\mathrm{G}}$ and the Hochschild differential $\partial_H$ in $D_{poly}(M)$.\\

Explicitly the Schouten-Nijenhuis bracket of factoring multi vector fields can be expressed with the Lie bracket of vector fields
$$[X_1\wedge\cdots\wedge{X}_a,{Y}_1\wedge\cdots\wedge{Y}_b]_{S-N}$$
$$=\sum_{i=1}^{a}\sum_{j=1}^{b}(-1)^{i+j}[X_i,Y_j]\wedge{X}_1\cdots{X}_{i-1}\wedge{X}_{i+1}\cdots{X}_a\wedge{Y}_1\cdots{Y}_{j-1}\wedge{Y}_{j+1}\cdots\wedge{Y}_b$$
$$[f,X]_{S-N}=-i(\d{f})X$$
for a function $f\in{C}^\infty(M)$ and the Leibniz rule.

Maurer-Cartan elements in $T_{poly}(M)$, {\em i.e.} elements of degree $1$ that are zeros of the Maurer-Cartan map $\d +[\cdot,\cdot]/2$ correspond to Poisson $2$-vector fields. Notice that here we consider $T_{poly}(M)$ endowed with trivial differential and the Maurer-Cartan equation reduces to the Poisson equation \ref{Poisson} in the shape
$[\Pi,\Pi]=0$,
equivalent we have for a Poisson structure a flat super derivation $\d_\Pi:=[\Pi,\cdot]_{\mathrm{S-N}}:T_{poly}^\bullet(M)\rightarrow{T}_{poly}^{\bullet+1}(M)$.\\

The Gerstenhaber bracket on $D_{poly}(M)$ can be described as follows: For $\alpha\in{D}^{n}_{poly}(M)$ and $\beta\in{D}^{m}_{poly}(M)$ we define for $0\leq{i}\leq{n}=\deg\alpha$ the insertions
$$(\alpha\circ_i\beta)(a_0,\cdots,a_{n+m})=\alpha(a_0,\cdots,a_{i-1},\beta(a_i,\cdots,a_{i+m}),a_{1+i+m},\cdots,a_{n+m})$$
$\in{D}^{n+m}_{poly}(M)$. With
$\alpha\circ\beta:=\sum_{i=0}^{n}\alpha\circ_i\beta$
we can define the Gerstenhaber bracket by
$$[\alpha,\beta]_{\mathrm{G}}=\alpha\circ\beta-(-1)^{nm}\beta\circ\alpha$$
The Hochschild differential is now just given by
$$\partial_\mathrm{H}=-[\cdot,m]_{\mathrm{G}}$$
where we denote by $m$ the natural multiplication of the algebra. 
Maurer-Cartan elements in $D_{poly}(M)$ correspond to associative deformations of the natural product $m$.

For $0\leq{n}$ in the literature the following canonical map is called the $\mathrm{HKR}$ map
$$\mathrm{HKR}:(\gamma_0\wedge\cdots\wedge\gamma_n)\rightarrow\left(f_0\otimes\cdots\otimes{f}_n\rightarrow\frac{1}{(n+1)!}\sum_{\sigma\in\Sigma_{n+1}}sgn(\sigma)\prod_{i=0}^{n}\gamma_{\sigma_i}(f_i)\right)$$
It is a famous theorem called the $\mathrm{HKR}$ theorem that this map is a quasi isomorphism of complexes $\left(T_\mathrm{poly}(A),0\right)\rightarrow\left(D_\mathrm{poly}(A),\mathrm{d}_\mathrm{H}\right)$. Hochschild Kostant and Rosenberg first proved the difficult, local part of this theorem for $Pol(\mathbb{R}^d)$ and the generalisation to any manifold can be quite easy achieved from the local theorem by a partition of unity, because
$$\mathrm{d}_\mathrm{H}f\phi=f\mathrm{d}_\mathrm{H}\phi$$
for any function $f$. In his famous paper \cite{K} Kontsevich also gave a proof of the $\mathrm{HKR}$ theorem and we mention \cite{CaGu} for more details. Moreover one can consider Kontsevich's formality theorem as a $L_\infty$ continuation of the $\mathrm{HKR}$ map as we will see in the following:
\subsection{Kontsevich's formality map}\label{KoOp}
\subsubsection*{Differential operators corresponding to graphs}
To every admissible graph $\Gamma$ with $n$ vertices of type I and $m$ of the type II and polyvector fields $\gamma_i$ of degree $\#\text{Star}(i)$ we associate a map
$$\mathcal{U}_\Gamma:\gamma_1\otimes\cdots\otimes\gamma_n\rightarrow\{A^{\otimes{m}}\rightarrow{A}\}$$
where we remind the reader that $A:=C^\infty(\mathbb{R}^d)$. The function
$$\Phi=\mathcal{U}_\Gamma(\gamma_1\otimes\cdots\otimes\gamma_n)(f_1\otimes\cdots\otimes{f}_m)$$
is defined to be the sum over all configurations of indices running from $1$ to $d$, labeled by $E_\Gamma$:
$$\Phi=\sum_{I:E_\Gamma\rightarrow\{1,\cdots,d\}}\Phi_I$$
where $\Phi_I$ is the product over all $n+m$ vertices of $\Gamma$ of certain partial derivatives. To be more precise, with each vertex $1\leq{i}\leq{n}$ of type I we associate the function $\Psi_i$ on $\mathbb{R}^d$ which is defined by
$$\Psi_i=<\gamma_i,dx^{I(e_i^1)}\otimes\cdots\otimes{d}x^{I\left(e_i^{\#\text{Star}(k)}\right)}>$$
and with each vertex $1\leq\overline{j}\leq{m}$ we associate $f_j$.

Then we put into each vertex $v$ instead of $\Psi_v$ the partial derivative
$$\tilde{\Psi}_v=\left(\prod_{e\in{E}_\Gamma,\;e=(\cdot,v)}\partial_{I(e)}\right)\Psi_v$$
and define with
$$\Phi_I=\prod_{v\in{V}_\Gamma}\tilde{\Psi}_v$$
finally the function
$$\Phi=\sum_I\Phi_I$$

Now we can describe the local Kontsevich $L_\infty$-quasi-morphism
$$\mathcal{U}:\left(T_\mathrm{poly}(\mathbb{R}^d),0,[\cdot,\cdot]_\mathrm{S-N}\right)\leadsto\left(D_\mathrm{poly}(\mathbb{R}^d),\mathrm{d}_\mathrm{H},[\cdot,\cdot]_\mathrm{G}\right)$$
by the formula of its n-th derivative $\mathcal{U}_n$ considered as a skew-symmetric polylinear map
$$\mathcal{U}_n:\otimes^{n}T_\mathrm{poly}(\mathbb{R}^d)\rightarrow{D}_\mathrm{poly}(\mathbb{R}^d)[1-n]$$
We set
$$\mathcal{U}_n=\sum_{m\geq0}\sum_{\Gamma\in\mathcal{G}(n,m)}\prod_{i=1}^{n}\frac{1}{\vert\text{Star}(i)\vert!}W_\Gamma{U}_\Gamma$$
where the integral weight of a graph $\Gamma\in\mathcal{G}(n,m)$ is defined by
$$W_\Gamma=\int_{C^+_{n,m}}\wedge_{e\in{E}_\Gamma}\d\phi(e)$$
with the Kontsevich propagator
$$\phi(s,t)=\frac{1}{2\pi}\mathrm{\arg}\left(\frac{t-s}{t-\overline{s}}\right)$$
 
It is maybe surprising that the map of Kontsevich is universal in the sense that the weights are defined independent of the dimension $d$ of $\mathbb{R}^d$.

The $L_\infty$-morphism $\mathcal{U}$ also enjoys several additional properties that we will again state in the following section in theorem \ref{InterpolationFormality} for a more general family of $L_\infty$-morphisms
$$\left(T_\mathrm{poly}(A),0,[\cdot,\cdot]_\mathrm{S-N}\right)\leadsto\left(D_\mathrm{poly}(A),\mathrm{d}_\mathrm{H},[\cdot,\cdot]_\mathrm{G}\right)$$
For instance we comment on the additional properties: $\mathcal{U}$ is $GL(d,\mathbb{K})$-equivariant, one can replace $\mathbb{K}^d$ in the construction by its formal completion $\mathbb{K}^d_{formal}$ at the origin. This first two properties rely on Kontsevich's universal graph construction, but there are two other properties that are relevant if we replace the integral weights: For vector fields $\gamma_i\in{T}^0_\mathrm{poly}(A)$ and $n\geq2$ we have 
$\mathcal{U}^{\lambda}(\gamma_1,\cdots,\gamma_n)=0$
and if $\gamma_i\in{T}_\mathrm{poly}(A)\;\forall{i}=2,\cdots,n$ are polyvector fields and $\gamma_1\in{T}^0_\mathrm{poly}(A)$ is a linear vector field we have the vanishing
$\mathcal{U}^{\lambda}(\gamma_1,\cdots,\gamma_n)=0$. The previous two properties are equivalent to the vanishing of certain weights, we will discuss this vanishing lemmas in detail in section \ref{Global}. The vanishing Lemmas \ref{Global} are relevant for the globalisation of the local result in the sense of \cite{D}, Dolgushev's globalisation approach of Kontsevich's local formality \cite{K} theorem {\em via} Fedosov resolutions \cite{D}. The globalisation of Shoikhet-Tsygan formality has been constructed by Dolgushev in \cite{DT1}, \cite{DT2}, also here Fedosov resolutions are an essential ingredient.
\chapter{\textit{Kontsevich formality for the interpolation}}\label{Interpol}
Kontsevich stated (without proof!) $[2]$ that the formality theorem (see later on) holds actually true if one replaces in the definition of the angle function $\phi$ the $\arg$-function by
$$\phi(h_s,h_t):=\frac{1}{2\pi i}\ln\left(\frac{h_t-h_s}{h_t-\overline{h}_s}\right)$$
where $\ln$ denotes any complex logarithm. The actual proof of this claim is quite technical and people working in deformation quantization and related fields have been interested in a detailed proof.
Together with Anton Alekseev, Charles Torossian and Thomas Willwacher we have managed to prove the logarithmic formality theorem, {\em i.e.} we realised that the $L_\infty$-relations are true if we replace in the definition of the Kontsevich integral weights the usual $\arg(\cdot)$ propagator by the $\ln(\cdot)$ propagator. The proof is to calculate with Stokes theorem in the presence of certain mild singularities and the computations are inspired by \cite{FB1} and \cite{FB2}. Here I will mainly adapt the arguments developed in discussions with Rossi and Willwacher and the calculations are done in the original coordinates. Compared to \cite{AWRT} here we consider the interpolation propagator
\begin{equation}\label{InterpoPro}
\phi^\lambda(h_s,h_t):=\frac{1}{2\pi i}\left[\lambda\ln\left(\frac{h_t-h_s}{h_t-\overline{h}_s}\right)-(1-\lambda)\ln\left(\frac{\overline{h}_t-\overline{h}_s}{\overline{h}_t-{h_s}}\right)\right]
\end{equation}
with $\lambda$ any complex number. This interpolation family of propagators was introduced by Rossi and Willwacher in \cite{WR} where they elaborate on an involved conjecture of Etingof concerning Drinfeld associators. In their construction the real part $\phi^{\mathbb{R}}$ is an essential ingredient used to interpolate between the Drinfeld associators $\Phi_{\mathrm{AT}}$, $\Phi_{\mathrm{KZ}}$ and $\Phi_{\overline{\mathrm{KZ}}}$ (The Alekseev-Torossian associator, the Knizhnik-Zamolodchikov associator and the anti-$\mathrm{KZ}$ associator respectively). 

The fact that the interpolation integral weights
$$W^\lambda_\Gamma=\int_{C^+_{n,m}}\wedge_{e\in{E}_\Gamma}\d\phi^\lambda(e)$$
converge is {\em a priori} not obvious: It requires a careful analysis of products of exterior derivatives of the logarithmic angle function on all boundary strata. 
Our proof of the formality for the interpolation is to show with help of analytic methods that this weights are well-defined and satisfy the important quadratic relations of the original Kontsevich weights. The main a bit surprising observation is that although a single interpolation propagator does not extend to the boundary strata the relevant forms of top degree that we integrate actually do not admit singularities on the boundary, the singularities cancel.

For a field $\mathbb{K}\supseteq \mathbb C$, we denote by $A$ the ring $\mathbb{K}[x_1,...x_d]$ of formal power series in $d$ variables. By $T_\mathrm{poly}(A)$ we denote the graded vector space of totally skew-symmetric multi-derivations of $A$ and by $D_\mathrm{poly}(A)$ the graded vector space of multi-differential operators on $A$.

The following theorem is the analog of Kontsevich's formality theorem defined with the interpolation propagator, the proof of the analytic part of this theorem is content of this chapter:
\begin{theorem}\label{InterpolationFormality}
For all $\lambda\in\mathbb{C}$ there is an $L_{\infty}$-quasi-isomorphism
$$\mathcal{U}^{\lambda}:\left(T_\mathrm{poly}(A),0,[\cdot,\cdot]_\mathrm{S-N}\right)\leadsto\left(D_\mathrm{poly}(A),\mathrm{d}_\mathrm{H},[\cdot,\cdot]_\mathrm{G}\right)$$
defined analogous to Kontsevich's $L_{\infty}$-quasi-isomorphism
$\mathcal{U}^{1/2}$ , but with the usual $\arg\bigr(\frac{t-s}{t-\overline{s}}\bigr)$-propagator replaced by the interpolation propagator
$$\phi^\lambda(s,t):=\frac{\lambda}{2\pi\I}\ln\Bigr(\frac{t-s}{t-\overline{s}}\Bigr)-\frac{1-\lambda}{2\pi\I}\ln\Bigr(\frac{\overline{t}-\overline{s}}{\overline{t}-s}\Bigr)$$
This family also enjoys the same additional properties as Kontsevich's formality map listed below:
\begin{enumerate}
\item  $\mathcal{U}^{\lambda}$ is $GL(d,\mathbb{K})$-equivariant.
\item  One can replace $\mathbb{K}^d$ in the formality construction by its formal completion $\mathbb{K}^d_{formal}$ at the origin.
\item The first Taylor coefficient of $\mathcal{U}^{\lambda}$ coincides with the Hochschild-Kostant-Rosenberg-quasi-isomorphism of complexes $T_\mathrm{poly}(A)\rightarrow{D}_\mathrm{poly}(A)$ given by
$$\mathrm{HKR}(\partial_{i_1}\wedge\cdots\wedge\partial_{i_p})=\frac{1}{p!}\sum_{\sigma\in\mathfrak S_p}(-1)^\sigma \partial_{i_{\sigma(1)}}\otimes\cdots\otimes\partial_{i_{\sigma(p)}}$$
\item For vector fields $\gamma_i\in{T}^0_\mathrm{poly}(A)$ and $n\geq2$ we have
$$\mathcal{U}^{\lambda}(\gamma_1,\cdots,\gamma_n)=0$$
\item If $\gamma_i\in{T}_\mathrm{poly}(A)\;\forall{i}=2,\cdots,n$ are general polyvector fields and $\gamma_1\in{T}^0_\mathrm{poly}(A)$ is a linear vector field we have
$$\mathcal{U}^{\lambda}(\gamma_1,\cdots,\gamma_n)=0$$
\end{enumerate}
\end{theorem}
\begin{proof}[Proof:]
We want to replace in Kontsevich's formality construction sketched in chapter \ref{GenQua} the usual Kontsevich propagator $\phi(s,t):=\frac{1}{2\pi}\mathrm{\arg}\bigr(\frac{t-s}{t-\overline{s}}\bigr)$ by the more general interpolation propagators
$\phi^\lambda(s,t):=\frac{\lambda}{2\pi\I}\ln\bigr(\frac{t-s}{t-\overline{s}}\bigr)-\frac{1-\lambda}{2\pi\I}\ln\bigr(\frac{\overline{t}-\overline{s}}{\overline{t}-s}\bigr)$, here the source $s$ and target $t$ of an edge are in the upper half-plane $\mathbb{H}$ and $s\neq{t}$.

We first show that the integrals converge and are well-defined:
\section{Convergence of integrals for the interpolation}\label{Convergence}
The following arguments do not apply, whenever $\Gamma$ contains a vertex of type I with valency $1$ and at least one more vertex of type I or more then one vertex of type II: In this case the action of $\mathbb R^+\ltimes\mathbb R$ allows to fix another vertex of type I or two vertices of type II. Now degree reasons imply that the associated weight vanishes, because we have to integrate the $1$-valency vertex of type I over a $2$ dimensional space and there is only one $1$-form depending on the corresponding two coordinates. This vanishing is not a new statement and traces back to the original paper of Kontsevich. The case of a $1$-valent vertex of type I and exactly one vertex of type II will be considered in detail in \ref{HKRnormalization}. We now consider graphs where all vertices of type I are at least of valency $2$:

Because we integrate top degree forms over compact manifolds it is enough to show that we can integrate the forms locally near the boundary, {\em i.e.} that $\wedge_{e\in{E}_\Gamma}\d\phi_e$ extends to a form of maximal degree near a codimension $1$ boundary strata of $\partial{C}^+_{n,m}$:

First consider a boundary where some of the vertices $h_1,\cdots,h_N$ collapse in $\mathbb{H}\setminus\mathbb{R}$. As usual we specify coordinates near this boundary strata by $h_1=h=x+\I{y}$ with $y\neq0$, $h_2=h+\epsilon\exp(\I\varphi)$ with $\epsilon\geq0$, $\varphi\in[0,2\pi)$ and $h_i=h+\epsilon\exp(\I\varphi){z}_i$ with ${z}_i\in\mathbb{C}$ for $i=3,\cdots,N$ and set $z_1=0$, $z_2=1$. A propagator of an edge joining two of the collapsing vertices (in the following \textbf{internal edge})
\begin{align*}\phi^\lambda(h_i,h_j)=\d\Biggr[&\frac{\lambda}{2\pi\I}\ln\left(\frac{\epsilon\exp(\I\varphi)(z_i-z_j)}{2\I{y}+\epsilon\exp(-\I\varphi)\overline{z_i}-\epsilon\exp(\I\varphi)z_j}\right)\\&-\frac{1-\lambda}{2\pi\I}\ln\left(\frac{\epsilon\exp(-\I\varphi)(\overline{z_i}-\overline{z_j})}{-2\I{y}+\epsilon\exp(\I\varphi){z_i}-\epsilon\exp(-\I\varphi)\overline{z_j}}\right)\Biggr]\end{align*}
contains only one term with a singular behaviour at $\epsilon=0$. Moreover we have explicit the combination
$(2\lambda-1){\d\epsilon}/{\epsilon}+\I\d\varphi$, hence our form of top degree will not contain a $2$-form of the shape $\d\epsilon\d\varphi/\epsilon$ if we multiply two internal $1$-forms of the previous shape: Because of skew-symmetry of $1$-forms every such term will appear twice with opposite sign. The previous skew-symmetry argument was communicated to us by Willwacher and in the following we will use it sometimes.

Moreover in a $1$-form corresponding to an edge that connects a collapsing vertex with a vertex that does not collapse (in the following \textbf{external edge}) the term $\d\varphi$ is coupled to $\epsilon$: Because of the coupling in the polar coordinates $h_i=h+\epsilon\exp(\I\varphi)z_i$ partial derivatives of an external propagator function with respect to $\varphi$ are in $\epsilon$ regular functions multiplied by $\epsilon$. Hence multiplication with a term of this external propagators containing $\d\varphi$ cancels an internal $\d\epsilon/\epsilon$ singularity and we can integrate this forms near the boundary strata where several vertices collapse in $\mathbb{H}\setminus\mathbb{R}$.

Now consider the case where $N$ type I and $M$ type II vertices with $2\leq{N}+M\leq{n}+m-1$ collapse to a point $r\in\mathbb{R}$ and in some sense dually we can consider the case when some type I vertices move to $\infty$: We choose as usual polar coordinates $h_1=R\exp(\I\alpha)$ with $0<\alpha<\pi$ and $h_i=R\exp(\I\alpha){z}_i$ with ${z}_i\in\mathbb{C}$ for the other $i=2,\cdots,N$ diverging vertices. Propagators connecting vertices who both go to $\infty$ do not contain $R$. In propagators corresponding to an edge connecting a vertex that goes to $\infty$ with a vertex $h_k$ that does not diverge $\d{R}$ will only appear in the combinations
$$\d{R}\frac{{e}^{\I\alpha}{z}_i(h_k-\overline{h}_k)}{(h_k-R{e}^{\I\alpha}{z}_i)(\overline{h}_k-R{e}^{\I\alpha}{z}_i)}\quad\text{and}\quad\d{R}\frac{h_k({e}^{-\I\alpha}\overline{{z}}_i-{e}^{\I\alpha}{z}_i)}{(R{e}^{\I\alpha}{z}_i-h_k)(R{e}^{-\I\alpha}\overline{{z}}_i-h_k)}$$
where in the first case $z_k$ was the source and in the second case the target of the edge. In both situations for $R\rightarrow\infty$ this $1$-form goes like $\d{R}/R^2$ and by this is integrable.
\section{$L_\infty$ property for the interpolation }\label{REG}
The main observation in Kontsevich formality is to use Stokes theorem
$\int_{M}\d\alpha=\int_{\partial{M}}\alpha$
for $\alpha\in\Omega^{\dim(M)-1}(M)$ and consider the trivial operator
$$0=\sum_{\substack{{\Gamma\in\mathcal{G}(n,m)}\\{\vert{E}_\Gamma\vert=2n+m-3 }}}\left(\int_{C^+_{n,m}}{\d}\prod_{i=1}^{2n+m-3}{\d}\phi^\lambda({e_i})\right)U_\Gamma$$
$$=\sum_{\substack{{\Gamma\in\mathcal{G}(n,m)}\\{\vert{E}_\Gamma\vert=2n+m-3 }}}\left(\int_{\partial{C}^+_{n,m}}\prod_{i=1}^{2n+m-3}{\d}\phi^\lambda({e_i})\right)U_\Gamma$$
In the following we want to justify this central vanishing for the interpolation family in the sense of Kontsevich and show that the same codimension $1$ boundary combinatorics follow for every $\lambda\in\mathbb{C}$: We will specify charts $U_{1,2n+m-3}$ near the different codimension $1$ boundary strata and apply the local Stokes formula
\begin{align*}
\int_{U_{1,2n+m-3}}{\d}\prod_{i=1}^{2n+m-3}{\d}\phi^\lambda({e_i})&=\lim_{\epsilon\rightarrow0}\int_{U_{1,2n+m-3}(\epsilon)}{\d}\prod_{i=1}^{2n+m-3}{\d}\phi^\lambda({e_i})\\&=\lim_{\epsilon\rightarrow0}\int_{\partial{U}_{1,2n+m-3}(\epsilon)}\prod_{i=1}^{2n+m-3}{\d}\phi^\lambda({e_i})
\end{align*}
Here the charts $U_{1,2n+m-3}(\epsilon)$ restrict the boundary coordinate to values $\geq\epsilon\in\mathbb{R}^+$. For the singularities along the border we use a standard regularisation procedure similar to the methods used in \cite{FB1}, \cite{FB2}. As we will see in the following calculations the pullbacks of the form $\prod_{i=1}^{2n+m-3}{\d}\phi^\lambda({e_i})$ to regularised codimension $1$ boundaries ${\partial{U}_{1,2n+m-3}(\epsilon)}$ extend to a form of maximal degree along the $\epsilon=0$ border. For a detailed description how to glue the local results by a partition of unity we refer the reader to the article \cite{AWRT}.

\subsection*{The boundary strata S1)}
\subsection{The Schouten-Nijenhuis Lie bracket}\label{SN}

We discuss the case where two vertices $h_1$ and $h_2$ collapse in the upper half plane to a single vertex $h=x+\I{y}\in\mathbb{H}$ with $y>0$: We use the action of $\mathbb R^+\ltimes\mathbb R$ to fix the point $h$ where the two vertices collapse somewhere in $\mathbb{H}$ and by this eliminate $x$ and $y$ as coordinates: Set $h_1=h$ and use polar coordinates $h_2=h+\epsilon\exp(\I\varphi)$, the limit $\epsilon\rightarrow0$ describes the codimension $1$ boundary
$C_2\times{C}_{n-1,m}$
corresponding to the collapse, where we identify as usual $C_2$ with $S^1$. For example the picture looks as follows, where we draw a circle ( Kontsevich's  ``magnifying glass") around the two vertices $h_1=h$ and $h_2=h+\epsilon\exp(\I\varphi)$ to symbolise their collapse
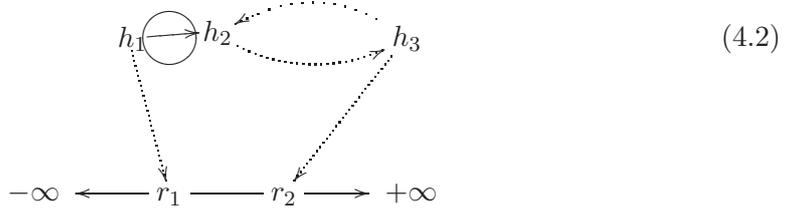
\begin{equation}\label{SNP}
\xymatrix{&*+<11pt,5.0em>[Fo]{\quad}\save[]+<-1.4em,0em>*{\;h_1}\ar@{.>}[d]\ar[r]+<-2.9em,0.2em>\restore&\save[]+<-2.4em,0.2em>*{\;h_2}\ar@{.>}@/_0.9pc/[r]\restore&h_3\ar@{.>}[dl]\ar@{.>}@/^-0.9pc/[l]+<-1.65em,0.5em>\\-\infty\;&{r_1}\ar@{-}[r]\ar@{->}[l]&{r_2}\ar@{->}[r]&\;+\infty}
\end{equation}
Notice that the appearing graph has $5$ edges, the manifold $C^+_{3,2}$ we integrate over is of dimension $3\cdot2+2-2=6$ and the missing degree of our form is compensated by the first $\d$ in the vanishing integral $0=\int_{C^+_{3,2}}{\d}\prod_{i=1}^{5}{\d}\phi^\lambda({e_i})$.

By elementary calculation the $1$-form corresponding to the internal edge
$$\I2\pi{\d}\phi^\lambda(h_1,h_2)=\lambda{\d}\ln\left(\frac{h_1-h_2}{\overline{h}_1-h_2}\right)-(1-\lambda){\d}\ln\left(\frac{\overline{h}_1-\overline{h}_2}{{h}_1-\overline{h}_2}\right)$$
transforms to
$$(2\lambda-1)\frac{{\d}\epsilon}{\epsilon}+\d\epsilon\left(\lambda\frac{e^{\I\varphi}}{-\epsilon{e}^{\I\varphi}-2\I{y}}-(1-\lambda)\frac{e^{-\I\varphi}}{-\epsilon{e}^{-\I\varphi}+2\I{y}}\right)$$
\begin{equation}
\label{ExPr}+\I{\d}\varphi+\epsilon\I\d\varphi\left(\lambda\frac{e^{\I\varphi}}{-\epsilon{e}^{\I\varphi}-2\I{y}}+(1-\lambda)\frac{e^{-\I\varphi}}{-\epsilon{e}^{-\I\varphi}-2\I{y}}\right)
\end{equation}
where we recall that $y$ got eliminated as a coordinate. Hence in this coordinates the propagator has the shape
$(2\lambda-1){{\d}\epsilon}/{\epsilon}$ plus terms regular in $\epsilon$.
The singular term $(2\lambda-1){{\d}\epsilon}/{\epsilon}$ and the other term proportional to $\d\epsilon$ are not intrinsic terms of a $\epsilon$ boundary {\em i.e.} orthogonal to the boundary, hence this terms do not contribute to the pullback to a regularised $\epsilon$ boundary given by
$$\I{\d}\varphi+\epsilon\I\d\varphi\left(\lambda\frac{e^{\I\varphi}}{-\epsilon{e}^{\I\varphi}-2\I{y}}+(1-\lambda)\frac{e^{-\I\varphi}}{-\epsilon{e}^{-\I\varphi}-2\I{y}}\right)$$
In the limit $\epsilon\rightarrow0$ this pullback form converges to the form $\I{\d}\varphi$.

It is clear that the pullback of a propagator corresponding to an external edge (an edge connecting a vertex that collapses with a vertex that does not collapse, we have drawn dotted edges for external vertices in \ref{SNP}) does not depend on $\varphi$ in the limit $\epsilon\rightarrow0$. In other words because of the coupling of $\epsilon$ and $\varphi$ in $h_2=x+\I{y}+\epsilon{e}^{\I\varphi}$ in every other $1$-form $\d\varphi$ will only appear as an in $\epsilon$ regular term multiplied by $\epsilon$, hence terms of external $1$-forms containing $\d\varphi$ vanish on the boundary $\epsilon=0$.

The integral over $C_2$, independent of $\lambda$, equals the integral of the normalized volume form on $S^1$ and the integral over $C_{n-1,m}$ is exactly the Kontsevich interpolation integral weight corresponding to the graph where the two first vertices $h_1$ and $h_2$ collapsed to a single type I vertex with propagator $\phi^\lambda$. This implies the usual contribution with the Schouten-Nijenhuis bracket in the following $L_\infty$-relations \ref{Kontsevich}.
\subsection{Kontsevich vanishing Lemma for the interpolation}\label{KV}
Now we discuss what happens when more than two vertices $h_1,h_2\dots{h_N}$ with $N\geq3$ collapse to a single vertex $h=x+\I{y}\in\mathbb{H}$

Let us W.L.O.G. rename this vertices so that we can assume that there is an edge joining $h_1$ and $h_2$. We again calculate with the section that fixes the point $h$ where the vertices collapse and now specify coordinates near the boundary
$ C_N\times C_{n+1-N,m}$
again as in \cite{K}: We set $h_1=h$ and write $h_2=h+\epsilon\exp(\I\varphi)$ for the first two collapsing vertices and for the remaining collapsing vertices we write $h_l=h+\epsilon\exp(\I\varphi){z}_l,\;N\geq{l}\geq3$ with ${z}_l\in{\mathbb{C}}\setminus\{0,1\}$ and ${z}_l\neq{z}_k$ if $l\neq{k}$ and consistent we set $z_1=0$ and $z_2=1$. In other words we choose for $C_{n+1-N,m}$ the section that fixes the point of collapse to $h$ and for $C_N$ the section that fixes the first vertex $z_1$ to zero and the second vertex $z_2$ on the unit circle $S^1$. Of course the integral is independent of the two fixed parameters $x,y$ and $x,y$ are not coordinates of $C_N\times{C}_{n+1-N,m}$.

Analog to the previous consideration the pullback to the boundary $\epsilon=0$ of forms corresponding to internal edges (edges between vertices that do both collapse) is
$$\I{\d}\varphi+\lambda{\d}\ln({z}_k-{z}_l)-(1-\lambda){\d}\ln(\overline{z}_k-\overline{z}_l)$$
As usual by dimensional reasons we are left with the situation where the internal sub graph of the $N$ collapsing vertices has $2N-3$ edges: Because we integrate a form of maximal degree we have to pick from the $1$-forms corresponding to internal edges partial derivatives with respect to the internal coordinates of $C_N$, in the limit $\epsilon\rightarrow0$ propagators corresponding to external edges do not depend on the internal $C_N$ coordinates $z_l$.

Recall that we assumed an edge joining $z_1$ and $z_2$ and we have to take $\d\varphi$ from this edge, because in other propagators the $\varphi$ dependence cancels. Because $\ln$ is holomorphic we see that the form we integrate over $C_N$ equals
\begin{equation}\label{RWRE}
\left(\lambda(1-\lambda)\right)^{N-2}\I{\d}\varphi\prod_{i=1}^{2N-4}{\d}\arg({z}_{s_i}-{z}_{t_i})
\end{equation}
where the two indices $s_i\neq{t_i}$ are determined by respectively the sources and targets of the $2N-3$ edges connecting collapsing vertices. Here we used the explicit shape
$\lambda{\d}\ln({z}_k-{z}_l)-(1-\lambda){\d}\ln(\overline{{z}}_k-\overline{z}_l)$
of the transformed forms: In our form of maximal degree every time $\lambda{\d}\ln({z}_k-{z}_l)$ gets multiplied with $(1-\lambda){\d}\ln(\overline{z}_i-\overline{z}_j)$ we also have the combination $\lambda{\d}\ln(\overline{{z}_k}-\overline{z}_l)$ multiplied with $(1-\lambda){\d}\ln({z}_i-{z}_j)$. This previous rewriting to the usual vanishing Lemma has been used by Rossi and Willwacher in their proof of a conjecture of Etingof on Drinfeld associators.

The previous formula \ref{RWRE} shows that the direction of internal edges (edges connecting two collapsing vertices) gets neglect-able on the boundary.

Again the argument that $\ln$ is holomorphic combined with the formulas $2\arg(\cdot)=\ln(\cdot)-\ln(\overline{\cdot})$ and $2\ln\vert\cdot\vert=\ln(\cdot)+\ln(\overline{\cdot})$
allows to rewrite the previous expression \ref{RWRE}
$$(-1)^{N-2}\left(\lambda(1-\lambda)\right)^{N-2}\I{\d}\varphi\prod_{i=1}^{2N-4}{\d}\ln\vert{{z}_{s_i}-{z}_{t_i}}\vert$$
The previous rewriting is Kontsevich's trick using logarithms and because the functions $\ln\vert{{z}_{s_i}-{z}_{t_i}}\vert$ do not depend on $\varphi$ and we can easily integrate the ${\d}\varphi$ part over $S^1$. 

From this point we could just argue that we have Kontsevich's famous vanishing Lemma. Because this vanishing observation is considered to be one of the most non-trivial points in Kontsevich formality and has also been discussed by other authors \cite{Kho}, let us explain a little bit for the reader how it works:

We integrate a form of maximal degree on $C_N\times{C}_{n+1-N,m}$ and it is clear that a non-trivial contribution can only appear if the valency of every type I vertex is $\geq2$.

\subsubsection*{Sub graphs containing a vertex of valency 2}
We could assume every collapsing vertex to be at least of valency $3$, the first sub graph of collapsing vertices satisfying this $3$-valency assumption is pictured below where we do not specify the neglect-able direction of edges on the $\epsilon=0$ boundary:\\

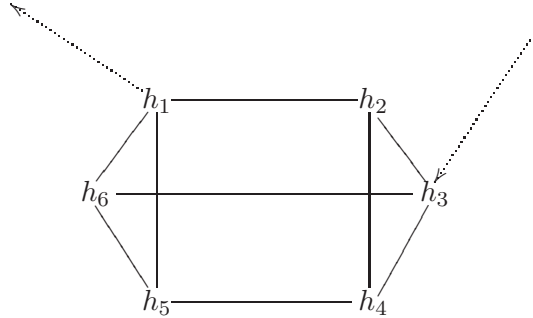
$${\xymatrix{&*+<11pt,13.5em>[]{}\save[]+<-3.6em,3.2em>*{h_1}\ar@{.>}[l]+<-5.6em,6.5em>\ar@{-}[r]+<0.3em,3.2em>\ar@{-}[dd]+<-3.6em,8.7em>\restore&\save[]+<0.7em,3.2em>*{\;h_2}\ar@{-}[dd]+<0.7em,8.7em>\restore&&\save[]+<0.7em,5.2em>*{}\ar@{.>}[dl]+<0.2em,9.6em>\restore&\\
\save[]+<-3em,9.2em>*{\;\;h_6}\ar@{-}[ur]+<-3.9em,2.8em>\ar@{-}[dr]+<-3.9em,8.7em>\restore&&&\save[]+<0.2em,9.2em>*{\;h_3\;}\ar@{-}[ul]+<1.0em,2.6em>\ar@{-}[dl]+<1.0em,8.4em>\ar@{-}[lll]+<-2.0em,9.2em>\restore&\\
&\save[]+<-3.6em,8.2em>*{h_5}\ar@{-}[r]+<0.1em,8.2em>\restore&\save[]+<0.7em,8.2em>*{\;h_4}\restore&{}&
}}$$\vspace{-2.2cm}\\
The proof for sub graphs containing a vertex of valency $2$ is a consequence of
$$0=\int_{z\in\mathbb{C}\backslash\{z_i,z_j\}\hspace{-1.5cm}}{\d}\left[\lambda\ln(z-z_j)-(1-\lambda)\ln(\overline{z}-\overline{z}_j)\right]{\d}\left[\lambda\ln(z-z_i)-(1-\lambda)\ln(\overline{z}-\overline{z}_i)\right]$$
and the proof of this vanishing Lemma goes quite along the same lines as the proofs for $2$-valency vertices contained in the appendix: We switch to polar coordinates, assume $\vert{z_i}\vert\leq\vert{z_j}\vert$ and split the integral to compute it with the geometric series and $\int_{0}^{2\pi}{\d}{\varphi}e^{i{n}\varphi}=2\pi\delta_0^n$. The difference compared to the computations \ref{2valent} for vertices of valency $2$ is that here we have instead of $\int^1_{\vert{z_2}\vert}\d{r}$ an integral $\int^\infty_{\vert{z_2}\vert}\d{r}$ and the boundary at $\infty$ does not contribute when integrating $1/r^{2+n}$ with ${n}\in\mathbb{N}$. There is also another well-known quite easy argument of Kontsevich to see the triviality of the previous integral by considering the map
$z\rightarrow{z_i+z_j-z}$
which is an orientation preserving involution. 

\subsubsection*{General sub graphs}
Let us W. L. O. G. assume an edge joining $z_1$ and $z_2$. To finally verify that for $N\geq3$ this integrals vanish without any valency restrictions we procced as in \cite{K}: We apply Stokes theorem and rewrite
$$\int_{\{{z}_1=0,{z}_2=1\;,\;{z}_3\dots{{z}}_N\in\mathbb{C}, {z}_k\neq{{z}}_l\}}\prod_{i=1}^{2N-4}\d\ln\vert{{z}_{s_i}-{z}_{t_i}}\vert$$
$$=\int_{\{{z}_1=0,{z}_2=1\;,\;{z}_3\dots{{z}}_N\in\mathbb{C}, {z}_k\neq{{z}}_l\}}\d\left(\ln\vert{{z}_{s_1}-{z}_{t_1}}\vert\prod_{i=2}^{2N-4}\d\ln\vert{{z}_{s_i}-{z}_{t_i}}\vert\right)$$
$$=\int_{\partial\{{{z}_1=0,{z}_2=1\;,\;{z}_3\dots{{z}}_N\in\mathbb{C}, {z}_k\neq{{z}}_l}\}}\ln\vert{{z}_{s_1}-{z}_{t_1}}\vert\prod_{i=2}^{2N-4}\d\ln\vert{{z}_{s_i}-{z}_{t_i}}\vert$$
The codimension $1$ boundary ${\partial\{{z}_1=0,{z}_2=1\;,\;{z}_3\dots{{z}}_N\in\mathbb{C}\}}$ consists of configurations where at least one vertex ${z}_j$ with $j\geq3$ goes to $\infty$ (this boundary corresponds to the collapse of the fixed vertices $z_1=0,z_2=1$) or configurations where more than two vertices collapse to a single vertex or one of the vertices $\{0,1\}$. We briefly discuss this boundary contributions:

We first discuss what happens when some vertices collapse to a single vertex ${z}$: Consider for example the case that the two vertices ${z}_{s_1}$ and ${z}_{t_1}$ collapse: We use again polar coordinates ${z}_{s_1}={z}_{t_1}+\epsilon\exp(\I\varphi)$ and the singularity $\ln(\epsilon)$ appears as a function but a form of top degree must contain ${\d}\varphi$ and $\varphi$ is coupled to $\epsilon$: Internal one forms do not depend on $\varphi$ because of the absolute norm $\vert\cdot\vert$ appearing in Kontsevich's trick using logarithms, hence in the pullback to an $\epsilon$ regularised boundary the form $\d\varphi$ will appear only in the pullback of an external edge and hence in the combination
$$\epsilon{\d}\varphi\left(\frac{e^{\I\varphi}}{{z}_k-{z}_{t_1}-\epsilon{e}^{\I\varphi}}-\frac{e^{-\I\varphi}}{\overline{z}_k-\overline{z}_{t_1}-\epsilon{e}^{-\I\varphi}}\right)$$
where ${z}_k\neq{z}_{t_1}$ is a not collapsing  vertex.
The reason for this is that we compute partial derivatives of the function $\ln\vert{z_k-z_{t_1}}\vert$, and this function is obviously invariant under the action of $S^1$ {\em i.e.} multiplication by $e^{\I\varphi}$ with $\varphi\in\mathbb{R}$. By L'H\^{o}spital's rule we have the well known limit
$lim_{\epsilon\rightarrow0}\ln(\epsilon)\epsilon=0$ and we see that the integrand vanishes on the boundary $\epsilon=0$.

For the other cases where more vertices collapse to a single vertex or where vertices collapse to ${z}_1=0$ and ${z}_1=1$ we can analogous to \label{KV} specify coordinates and the boundary contributions vanish similar by the same analytic inspection, dimensional arguments.

Dually consider the case where the vertex ${z}_{s_1}$ goes to $\infty$, we use polar coordinates ${z}_{s_1}=R\exp(\I\varphi)$ where $R\rightarrow\infty$. In the pullback to the $R$ boundary the $\d\varphi$ term will appear only in the combination
$$R{\d}\varphi\left(\frac{e^{\I\varphi}}{{z}_j-R{e}^{\I\varphi}}-\frac{e^{-\I\varphi}}{\overline{{z}}_j-R{e}^{-\I\varphi}}\right)=R{\d}\varphi\frac{e^{\I\varphi}\overline{{z}}_j-e^{-\I\varphi}{z}_j}{({z}_j-R{e}^{\I\varphi})(\overline{{z}}_j-R{e}^{-\I\varphi})}$$
Again by L'H\^{o}spital's rule we have
$lim_{R\rightarrow\infty}\ln(R)/R=0$,
hence the integrand vanishes on the boundary ${R}=\infty$. The same arguments also work if more than one vertex goes to $\infty$ because we can choose for the other vertices coordinates ${z}_i=R\exp(\I\varphi){z'}_i$ with ${z'}_i\neq1$ and ${z'}_i\neq{z'}_j$ if $i\neq{j}$ to de couple the $\varphi$ dependence in $\ln\vert{z_i}-z_j\vert$. The cases where other vertices diverge to $\infty$ are quite analogous.
\subsection*{The boundary strata S2)}\label{S2}
$N\in\mathbb{N}$ vertices of type I and $M\in\mathbb{N}$ vertices of type II with $N+M\geq2$ and $N+M\leq{n+m}-1$ collapse together to a single vertex $r$ on the real line, the codimension $1$ boundary strata is isomorphic to
$C_{N,M}\times{C}_{n-N,m+1-M}$.

We calculate with the section where the real part of one of the vertices that do not collapse is fixed to zero and the coordinate on $\mathbb{R}$ where the vertices collapse is fixed to some $r\in\mathbb{R}$. The coordinate identification we specify as follows: For the collapsing type I vertices we use the coordinates $h_1=r+\epsilon\I$ and $h_i=r+\epsilon{h'_i}$ for $i=2,\cdots,N$ and for the collapsing type II vertices the coordinates $r_i=r+\epsilon{r'_i}$ with ${r'_i}\in\mathbb{R}$ for $i=1,\cdots,M$ and in the limit the coordinates ${h'_i}$ get restricted to $\mathbb{H}$. In other words for the boundary we choose the section of $C_{N,M}$ where we fix $h'_1$ to $\I$ and the section of $C_{n-N,m-M+1}$ where we fix the new vertex, corresponding to the collapsed vertices, to $r$ and still the real part of one of the vertices that do not collapse is fixed to $0$.

Internal propagators, {\em i.e.} edges between collapsing vertices, do not depend on the point in $\mathbb{R}$ where the vertices collapse. Internal propagators do not depend on $\epsilon$ and also the other propagators are not singular in the limit $\epsilon\rightarrow0$. Hence as usual the integral factorizes in a product of two integrals for every interpolation propagator $\phi^\lambda$: Because we integrate forms of maximal degree on $C_{N,M}$ and $C_{n-N,m-M+1}$ this integrals can only be non-zero if the number of internal edges equals $2N+M-2$.
\subsubsection*{S2.1) A bad edge}
We assume that there is an external edge with source one of the collapsing vertices and target one of the vertices that do not collapse. If $s$ converges to $\mathbb{R}$ the $1$-form corresponding to this edge
${\d}\phi^\lambda(s,t)=\lambda{\d}\ln\left((s-t)/(\overline{s}-t)\right)-(1-\lambda){\d}\ln\left((\overline{s}-\overline{t})/(s-\overline{t})\right)$
vanishes. Roughly the argument for the vanishing is that the two quotients $(s-t)/(\overline{s}-t)$ and $(\overline{s}-\overline{t})/(s-\overline{t})$ both converge to $1$ if $s\rightarrow\mathbb{R}$, by uniform convergence we can commute partial differentials and limit, hence
$$\lim_{s\rightarrow\mathbb{R}}{\d}\phi^\lambda(s,t)=\lambda{\d}\ln\left(1\right)-(1-\lambda){\d}\ln\left(1\right)=0$$
\subsubsection*{S2.2) No bad edge}
Here clearly the integral decomposes as usual by Fubini's theorem into the product of the integral over $C_{N,M}$ of the interpolation form corresponding to the internal sub graph $\Gamma_{N,M}$ and the integral over $C_{n-N,m-M+1}$ corresponding to the external graph
$\Gamma_{n,m}/\Gamma_{N,M}$
obtained from $\Gamma_{n,m}$ by contraction of the collapsing vertices $\Gamma_{N,M}$ to a single type II vertex. As usual this boundary contribution corresponds in the following $L_\infty$-relations \ref{Kontsevich} to the terms where the Hochschild differential $\mathrm{d}_\mathrm{H}$ and Gerstenhaber bracket $[\cdot,\cdot]_G$ appear.
\section{$\mathrm{HKR}$ normalisation for the interpolation}\label{HKRnormalization}
One of the major properties of Kontsevich's original formality map is that it begins with the famous $\mathrm{HKR}$ quasi-isomorphism and this is also true for the interpolation propagator: We consider graphs with exactly one vertex of the type I and $m$ type II vertices
$$
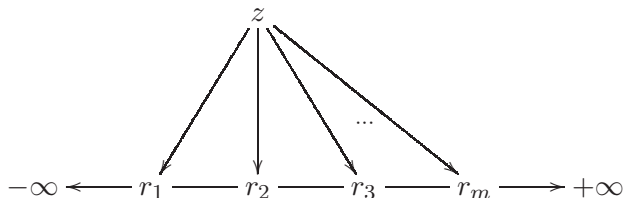
\vspace{-0.3cm}\xymatrix{*\txt{}&*\txt{}&z\ar@//[ddl]\ar@//[dd]\ar@//[ddr]\ar@//[ddrr]_{\cdots}\\&&&&&\\
*\txt{$-\infty\;$}&r_1\ar[l]\ar@{-}[r]&r_2\ar@{-}[r]&r_3\ar@{-}[r]&r_m\ar[r]&*\txt{$\;+\infty$}\\}$$
To compute this integral weights first notice that the interpolation propagator connecting type I with type II vertices exactly equals the original Kontsevich propagator because for $r\in\mathbb{R}$ we have
$$\ln\bigr((s-r)/(\overline{s}-r)\bigr)=\ln\bigr((s-r)/(\overline{s-r})\bigr)=2\I{\arg}\bigr((s-r)/(\overline{s}-r)\bigr)$$
We transform from the upper half plane to $\mathbb{D}_1(0)$ and fix the single vertex of type I to $0\in\mathbb{D}_1(0)$. Because the type II vertices are on $S^1$ and $\ln(1)=0$ by the previous argument for any interpolation propagator $\phi^\lambda$ the computation reduces to
$$\frac{1}{(2\pi\I)^m}\int_0^{2\pi}\I{\d}\varphi_m\underbrace{\int_0^{\varphi_m}\I{\d}\varphi_{m-1}\dots\underbrace{\int_0^{\varphi_3}\I{\d}\varphi_2\underbrace{\int_0^{\varphi_2}\I{\d}\varphi_1}_{\varphi_2}}_{{\varphi^2_3}/{2}}}_{{\varphi^{m-1}_m}/{(m-1)!}}=\frac{1}{m!}$$
where we used the usual $S^1$ coordinates for the ordered $m$ vertices of type II. This implies that if we integrate out the type II vertices in the above picture we yield a factor $1/m!$ for every interpolation propagator.
\section{Type I vertices of valency $2$ in the interpolation}\label{2valent}
By dimensional reasons a non-vanishing weight clearly forces the valency of a vertex to be $\geq2$ for type I and $\geq1$ for the type II. As we will see type I vertices of valency $2$ are quite well understood for the interpolation but there are only a few things for vertices with a valency $\geq3$ and there are cases where the valency of all vertices of type I can be at least $3$: On the one hand we integrate a form of top degree over $C^+_{n,m}$ and the dimension of $C^+_{n,m}$ is $2n+m-2$ but on the other hand every edge contributes with two to the overall valency and therefore the average valency of a type I vertex is smaller than
$4+(m-4)/n$.
This implies that there is a vertex with valency $\geq4+(m-4)/n$ and a vertex with valency $\leq4+(m-4)/n$. For example if $m\leq3$ there is a type I vertex corresponding to either a $3$-vector field, a $2$-vector field, a $1$-vector field or a function.

The following three sub-sections combined enable us to determine what happens if we integrate out a vertex of type I  with valency $2$ of any Kontsevich, Shoikhet integral weight performing a two dimensional integral. This is done by considering three cases, the outcome of the first sub-section \ref{Global} is well-known by different methods and is important in the globalisation of the interpolation formality \cite{WR}, hence necessary for the reader interested in the global formality theorems. The second sub-section \ref{Qua} just states another vanishing Lemma and \ref{Coupl} was added for completeness. For the proof we use a standard method, this method is quite similar has also been used by Merkulov \cite{M}. The strategy in principle allows to evaluate an arbitrary Kontsevich integral weight and the recipe briefly can be described as follows:
\begin{itemize}
\item First we transform the Kontsevich integrals from the upper-half-plane $\mathbb{H}$ to the unit disc $\mathbb{D}_1(0)$ with help of the M\"obius transform
$z\rightarrow\frac{z-\I}{z+\I}$.
Therefore we defined the logarithmic Kontsevich propagator $\phi^{\ln}$ on $\mathbb{D}_1(0)$ by 
$$\phi^{\ln}(w_s,w_t)=\frac{1}{2\pi\I}\ln\left(\frac{(1-\overline{w_s})(w_s-w_t)}{(1-w_s)(1-\overline{w_s}w_t)}\right)$$
instead of the usual Kontsevich propagator \cite{K} defined on $\mathbb{H}$.
\item For the computation on $\mathbb{D}_1(0)$ we switch to polar coordinates, multiply out the $1$-forms and split the integration domain so that we can expand geometric series
\begin{equation}\label{GS}\frac{1}{1-x}=\begin{cases}{\sum_{l=0}^\infty{x}^l\quad\hspace{1.72cm}\text{if}\;\vert{x}\vert<1}\\{-\sum_{l=0}^\infty1/x^{l+1}\;\;\;\quad\hspace{0.3cm}\text{if}\;\vert{x}\vert>1}\end{cases}\end{equation}
\item Finally to calculate the integrals we use the standard Stokes formulas
\begin{equation}\label{TI}
\int_{0}^{2\pi}{\d}{\varphi}e^{i{k}\varphi}=2\pi\delta_0^k\quad\forall\;{k}\in\mathbb{Z}\end{equation}
\begin{equation}\label{PI}\int_{a}^b{\hspace{-0.3cm}{\d}r}\hspace{0.1cm}r^{n}\ln^m(r)=\begin{cases}\sum_{j=0}^{m}\frac{(-1)^jj!}{(n+1)^{j+1}}\binom{m}{j}\left(a^{n+1}\ln^{m-j}(a)-b^{n+1}\ln^{m-j}(b)\right)\quad\text{for}\;n\neq-1\\{\ln^{m+1}(a/b)/(m+1)\hspace{2.13cm}\text{for}\;n=-1}\end{cases}\end{equation}
where $m\in\mathbb{N}$ and we obviously assume in some cases $a,b\neq0$.
\end{itemize}
In practice the oscillating integral formula \ref{TI} will couple certain indices of geometric series, for example for wheel-like graphs \ref{Wheel} this leads to $\int_0^1\d{x}_1\cdots\int_0^1\d{x}_n1/(1-x_1\cdots{x}_n)$. Application of with \ref{PI}  identifies this integrals with the $\zeta$ function evaluated at positive integers $n\geq2$ as noticed by Merkulov, we will give a detailed demonstration of the method in the appendix \ref{MerkWe}.

Of course for high dimensional Kontsevich integrals this method can still be a lot of work, but for example the method establishes the following helpful tool:
\begin{proposition}\label{MerkulovVanishing}
For an on $\mathbb{D}_1(0)$ convergent holomorphic power series
$$f(x)=\sum_{n=0}^{\infty}a_nx^n$$
we have the vanishing
\begin{equation}
\int_{w\in\mathbb{D}_1(0)\backslash\{p\}}\hspace{-1cm}{\d}\overline{w}{\d}w\;f(\overline{w}){\left(\frac{1}{w-p}+\frac{\overline{p}}{1-w\overline{p}}\right)}=0\end{equation}
\end{proposition}
\begin{proof} For $f$ as above we calculate
$$\int_{w\in\mathbb{D}_1(0)\backslash\{p\}}\hspace{-1cm}{\d}\overline{w}{\d}w{\frac{\overline{p}f(\overline{w})}{w\overline{p}-1}}=\frac{2\pi}{\I}\sum_{n=0}^{\infty}\frac{a_n\overline{p}^{n+1}}{n+1}=\frac{2\pi}{\I}\int_0^{\overline{p}}\hspace{-0.2cm}{\d}zf(z)$$
where we compute with the previous recipe and just used \ref{GS}, \ref{TI} and \ref{PI}. Maybe its interesting to notice that we can forget about the measure zero set $\{w=p\}$, switching to polar coordinates centred at the singularity the pole of order one cancels with the volume element. For the other integral we first split the integration domain
$$\int_{w\in\mathbb{D}_1(0)\backslash\{p\}}\hspace{-1cm}{\d}\overline{w}{\d}w{\frac{f(\overline{w})}{w-p}}=\int_{0\leq\vert{w}\vert\leq\vert{p}\vert}\hspace{-1cm}{\d}\overline{w}{\d}w{\frac{f(\overline{w})}{\frac{w}{p}-1}}\frac{1}{p}+\int_{{p}\vert<\vert{w}\vert\leq1}\hspace{-1cm}{\d}\overline{w}{\d}w{\frac{f(\overline{w})}{1-\frac{p}{w}}}\frac{1}{w}$$
and again calculate with \ref{GS},\ref{TI} and \ref{PI}.\end{proof}\subsection{Globalisation vanishing Lemma for the interpolation}\label{Global}
As mentioned we need for Dolgushev's version of the Fedosov globalisation of the local logarithmic Kontsevich-, Tsygan formality theorems the following vanishing Lemmas:
\begin{itemize}
\item For vector fields $\gamma_i\in{T}^0_\mathrm{poly}(A)$ and $n\geq2$ we have
$\mathcal{U}^{\lambda}(\gamma_1,...,\gamma_n)=0$.
\item If $\gamma_i\in{T}_\mathrm{poly}(A)\;\forall{i}=2,...,n$ are general polyvector fields and $y_1\in{T}^0_\mathrm{poly}(A)$ is a linear vector field we have
$\mathcal{U}^{\lambda}(\gamma_1,...,\gamma_n)=0$.
\end{itemize}
We sketch the argumentation of \cite{K} that reduces the situation to special graphs: From Kontsevich's construction \ref{KoOp} of differential operators corresponding to graphs it is clear that the second statement concerns vertices of valency $2$: If we differentiate an in the coordinates linear vector field more than once the result vanishes. Also the first statement is clearly a statement about vertices of valency $2$: To define a non-vanishing weight the number of edges should equal $2n+m-2$ and every vector field contributes with one edge, this implies $2n+m-2=n$ and hence $n\leq2$, the only graph is\\
\begin{equation}\label{2wheel}
\xymatrix{ 
w_1\ar@/^1.3pc/[rr]&&w_2\ar@/^1.3pc/[ll]}
\end{equation}\\
Hence this two vanishing Lemmas only concern the case where a type I vertex of valency $2$ has exactly one departing and one incoming edge. The action of $\mathbb R^+\ltimes\mathbb R$ allows to fix a vertex different from $w$ and by Fubini's theorem we can isolate a two dimensional integral. In the logarithmic case the statements now follow from the vanishing Lemma 
\begin{equation}\label{GlobVan}
\int_{w\in\mathbb{D}_1(0)\backslash\{w_1,w_2\}}{\d}\phi^1(w,w_1){\wedge}{{\d}\phi^1(w_2,w)}=0
\end{equation}
where $w_1,w_2$ are inside of the unit disk $\mathbb{D}_1(0)$, {\em i.e.} we used the M\"obius transform
$z\rightarrow{w}=(z-\I)/(z+\I)$
from the upper half plane onto $\mathbb{D}_1(0)$ and especially maps $\mathbb{R}$ to the unit circle $S^1$ and $\infty$ to $1$.

The previous integral corresponds to a single vertex of type I with exactly one incoming and one departing arrow pictured below
\begin{equation}\label{GV}
\xymatrix{ 
w_1\ar@/^1pc/[r]&w\ar@/_1pc/[r]&w_2}\end{equation}
This implies that the Kontsevich integral weight of a graph containing such a vertex of valency $2$ vanishes.  This vanishing Lemma is well-known, but we will present here a quite easy proof. To give the reader some alternatives for the globalisation we also sk{\em etc}h at the end of this section another proof of the globalisation vanishing Lemma that to the best of our knowledge originally traces back to C. Torossian, we thank T. Willwacher for pointing out this alternative proof of the globalisation. Our calculation goes as follows:

First notice that in the logarithmic case only the departing edge contains $\d\overline{w}$ and hence we consider
$$\int_{w\in\mathbb{D}_1(0)\backslash\{w_1,w_2\}}\hspace{-1cm}{\d}\overline{w}{\d}{w}\left(\frac{-1}{1-\overline{w}}+\frac{w_1}{1-\overline{w}w_1}\right)\left(\frac{-1}{w_2-w}+\frac{\overline{w_2}}{1-\overline{w_2}w}\right)$$
and now \ref{MerkulovVanishing} could be used to eliminate all the terms appearing in this integral but let us discuss the situation in more detail:
$$=2\I\int_{0}^1\hspace{-0.3cm}{r}{{\d}r}\int_{0}^{2\pi}\hspace{-0.5cm}{{\d}\varphi}\left(\frac{-1}{1-re^{-i\varphi}}+\frac{w_1}{1-w_1{r}e^{-i\varphi}}\right)\left(\frac{-1}{w_2-re^{i\varphi}}+\frac{\overline{w_2}}{1-\overline{w_2}re^{i\varphi}}\right)$$
where we used to polar coordinates ${\d}{w}{\d}\overline{w}=(2/\I)r{\d}r\d\varphi$. 
Now notice we can also forget about the measure zero set $\{w=w_1\}{\cup}\{w=w_2\}$, switching to polar coordinates centred at the singularities the pole of order one cancels with the volume element, hence the integrand is non-singular. 

The factors ${w_i}re^{i\varphi}$ are by assumption inside of the unit disk and for example we can evaluate the integral
$$\int_{0}^1{r}{{\d}r}\int_{0}^{2\pi}{{\d}\varphi}\frac{w_1}{1-w_1re^{-i\varphi}}\frac{\overline{w_2}}{1-\overline{w_2}re^{i\varphi}}$$
with help of the geometric series as follows:
$$w_1\overline{w_2}\int_{0}^1\hspace{-0.3cm}{r}{{\d}r}\int_{0}^{2\pi}\hspace{-0.5cm}{{\d}\varphi}\left(\sum_{n=0}^\infty{w_1}^nr^ne^{-in\varphi}\right)\left(\sum_{m=0}^\infty{\overline{w_2}}^mr^me^{im\varphi}\right)$$
Cauchy product multiplication of the geometric series, $\int_{0}^{2\pi}{\d}{\varphi}e^{i{n}\varphi}=2\pi\delta_0^n$ and $\int_{a}^b{{\d}r}r^{n}=\frac{1}{n+1}(a^{n+1}-b^{n+1})$ for $n\neq-1$ and $a,b\neq0$ if $n<-1$ yield
$$2\pi{w_1}\overline{w_2}\sum_{n=0}^\infty{w_1}^n{\overline{w_2}}^n\int_{0}^1{{\d}r}r^{2n+1}=2\pi{w_1}\overline{w_2}\sum_{n=0}^\infty\frac{(w_1\overline{w_2})^n}{2n+2}=-\pi\ln(1-w_1\overline{w_2})$$

The computation of for example the integral
$$\int_{0}^1{r}{{\d}r}\int_{0}^{2\pi}{{\d}\varphi}\frac{w_1}{1-re^{-i\varphi}w_1}\frac{-1}{w_2-re^{i\varphi}}$$
is quite analogous, but we have to be more careful with the geometric series and consider the cases $0<r<\vert{w_2}\vert$ and $\vert{w_2}\vert<r<1$ separately:
$$-\frac{w_1}{w_2}\int_{0}^{\vert{w_2}\vert}{r}{{\d}r}\int_{0}^{2\pi}{{\d}\varphi}\left(\sum_{n=0}^\infty{w_1}^nr^ne^{-in\varphi}\right)\left(\sum_{m=0}^\infty{{w_2}}^{-m}r^me^{im\varphi}\right)$$
$$+w_1\int^1_{\vert{w_2}\vert}{{\d}r}\int_{0}^{2\pi}{{\d}\varphi}\left(\sum_{n=0}^\infty{w_1}^nr^ne^{-in\varphi}\right)\left(\sum_{m=0}^\infty{{w_2}}^{m}r^{-m}e^{-im\varphi}\right)e^{-i\varphi}$$
Again the Cauchy product and $\int_{0}^{2\pi}{\d}{\varphi}e^{i{n}\varphi}=2\pi\delta_0^n$ yield that the second of this integrals vanishes and the first integral equals 
$$-2\pi\frac{w_1}{w_2}\sum_{n=0}^\infty{w_1}^n{w_2}^{-n}\int_{0}^{\vert{w_2}\vert}{{\d}r}r^{2n+1}=\pi\ln(1-w_1\overline{w_2})$$

The two remaining integrals are easy evaluated by just setting $w_1=1$ in the previous computation.

Essentially the same arguments of the previous and following computations yield
\begin{equation}\label{GobalizationVanishing}0=\int_{w\in\mathbb{D}_1(0)\backslash\{w_1,w_2\}}\hspace{-1.2cm}{\d}\phi^\lambda(w,w_1){\wedge}{{\d}\phi^\lambda(w_2,w)}
\end{equation}
for all interpolation propagators $\phi^\lambda$: With help of the vanishing lemma \ref{GlobVan} and its complex conjugate this integral can be rewritten as follows
\begin{align*}=\lambda(1-\lambda)\int_{w\in\mathbb{D}_1(0)\backslash\{w_1,w_2\}}\hspace{-1cm}{\d}w{\d}&\overline{w}\biggr[\left(\frac{1}{w-w_1}+\frac{1}{1-w}\right)\left(\frac{1}{\overline{w}-\overline{w_2}}+\frac{w_2}{1-w_2\overline{w}}\right)\\&-\left(\frac{1}{w-w_2}+\frac{\overline{w_2}}{1-\overline{w_2}w}\right)\left(\frac{1}{\overline{w}-\overline{w_1}}+\frac{1}{1-\overline{w}}\right)\biggr]\end{align*}
The vanishing Lemma \ref{MerkulovVanishing} and its complex conjugate reduce the computation to
\begin{align}\label{ReCo}=\lambda(1-\lambda)\int_{w\in\mathbb{D}_1(0)\backslash\{w_1,w_2\}}\hspace{-1cm}{\d}w{\d}&\overline{w}\biggr[\frac{1}{w-w_1}\left(\frac{1}{\overline{w}-\overline{w_2}}+\frac{w_2}{1-w_2\overline{w}}\right)\\&-\left(\frac{1}{w-w_2}+\frac{\overline{w_2}}{1-\overline{w_2}w}\right)\frac{1}{\overline{w}-\overline{w_1}}\biggr]\nonumber\end{align}
The previous computation
$$\int_{w\in\mathbb{D}_1(0)\backslash\{w_1,w_2\}}{\d}{w}{\d}\overline{w}\frac{w_1}{1-\overline{w}w_1}\frac{1}{w_2-w}=\frac{2}{\I}\pi\ln(1-w_1\overline{w_2})$$
helps to evaluate two of the terms that appear in the previous integral \ref{ReCo}. The other two terms, that appear and cancel this evaluated terms are
$$\int_{w\in\mathbb{D}_1(0)\backslash\{w_1,w_2\}}\hspace{-0.9cm}{\d}{w}{\d}\overline{w}\left(\frac{-1}{w_1-{w}}\frac{1}{\overline{w}-\overline{w_2}}+\frac{1}{w_2-w}\frac{1}{\overline{w}-\overline{w_1}}\right)=\frac{2}{\I}\pi\ln\left(\frac{1-\overline{w_1}{w_2}}{1-w_1\overline{w_2}}\right)$$
For this integral we first assume $\vert{w_1}\vert<\vert{w_2}\vert$ and switch to polar coordinates. The integral from $0$ to $\vert{w_1}\vert$ contributes with
$$\pi(\ln(1-\overline{w_1}/\overline{w_2})-\ln(1-w_1/w_2))$$
The integral from $\vert{w_1}\vert$ to $\vert{w_2}\vert$ vanishes and the integral from $\vert{w_2}\vert$ to $1$ contributes with
$$\pi(\ln(1-\overline{w_1}{w_2})-\ln(1-w_1\overline{w_2})-(\ln(1-\overline{w_1}/\overline{w_2})+\ln(1-w_1/w_2))$$
Here one has to be careful with signs,  for example that two times also a term $\ln(\vert{w_2}\vert)$ appears but this two terms cancel each other. The two $\ln(\vert{w_2}\vert)$ terms ``appear" in the computation of
$$\int^1_{\vert{w_2}\vert}{\hspace{-0.4cm}{\d}r}\int_{0}^{2\pi}{\hspace{-0.4cm}{\d}\varphi}\frac{1}{r}\Biggr[\frac{1}{1-w_1/(re^{i\varphi})}\frac{1}{1-\overline{w_2}/(re^{-i\varphi})}-\frac{1}{1-w_2/(re^{i\varphi})}\frac{1}{\overline{1-w_1/(re^{-i\varphi})}}\Biggr]$$
when integrating the first coefficients in the geometric series where we use $\ln(r)$ as a primitive of $1/r$ as usual. This two singular seeming terms do not really contribute because they cancel before the integration. The assumption $\vert{w_1}\vert<\vert{w_2}\vert$ makes the Taylor expansion for the logarithm convergent and finally we can drop the assumption because the integrand depends continuous on its arguments and symmetry.

To give the reader an alternative, more direct proof let us mention the following: This globalisation vanishing Lemma has also been proven by Rossi, Willwacher and Torossian \cite{WR} applying Stokes theorem on a fibre integral of the exterior derivative of the differential form corresponding to the graph\\
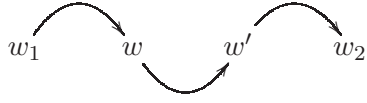
$$
\xymatrix{ 
w_1\ar@/^1.4pc/[r]&w\ar@/_1.4pc/[r]&w'\ar@/^1.4pc/[r]&w_2}$$\\
where the only codimension $1$ boundary contribution that appears is three times the graph \ref{GV}, hence this can independent of signs only sum up to $0$ if \ref{GV} vanishes, \ref{KV} contains quite similar computations with Stokes theorem.

\subsection{Interpolation vanishing of target vertices of type I and valency $2$}\label{Qua}
Shoikhet conjectured a certain compatibility of $\star$ products with Casimir functions \cite{S} and rephrased the conjecture states that the weight of any graph containing a vertex of type I corresponding to a function (equivalent with no outgoing edges) with arbitrary many edges pointing on it vanishes for the original Kontsevich propagator $\phi^{1/2}$. In the logarithmic $\phi^1$ respectively anti-logarithmic $\phi^0$ case this vanishing is quite clear because the form of maximal degree we integrate only contains $\d{w}$ respectively $\d\overline{w}$.

Here we do not prove the full conjecture of Shoikhet, but we claim in the $2$-valency case the vanishing lemma
\begin{equation}\label{TargetVanishing}\int_{w\in\mathbb{D}_1(0)\backslash\{w_1,w_2\}}{\d}\phi^\lambda(w_1,w){\wedge}{{\d}\phi^\lambda(w_2,w)}=0\end{equation}
 for the whole $\lambda$-interpolation family.

This integral can be rewritten as follows
\begin{align*}=\lambda(1-\lambda)\int_{w\in\mathbb{D}_1(0)\backslash\{w_1,w_2\}}\hspace{-1cm}{\d}w{\d}&\overline{w}\biggr[\left(\frac{1}{w-w_1}+\frac{\overline{w_1}}{1-\overline{w_1}w}\right)\left(\frac{1}{\overline{w}-\overline{w_2}}+\frac{w_2}{1-w_2\overline{w}}\right)\\&-\left(\frac{1}{w-w_2}+\frac{\overline{w_2}}{1-\overline{w_2}w}\right)\left(\frac{1}{\overline{w}-\overline{w_1}}+\frac{w_1}{1-w_1\overline{w}}\right)\biggr]\end{align*}
Formula \ref{MerkulovVanishing} reduces the computation analogous to \ref{ReCo} and we can argue with our calculations for the globalisation vanishing Lemma in the interpolation case.

The corresponding graphical picture is\\
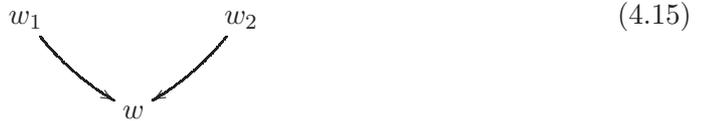
\begin{equation}\label{Morsewheel}
\xymatrix{ 
w_1\ar@/_0.3pc/[dr]&&w_2\ar@/^0.3pc/[dl]\\&{w}&}
\end{equation}\\
Hence the local Kontsevich formality morphism where we plug in a function $f$ as a I type argument only contains terms with cubic derivatives of this function, in other words if the function $f$ is quadratic then the operators $\mathcal{U}^\lambda(f,\gamma_2\cdots\gamma_n)$ vanish identically.

\subsection{Coupling by source vertices of type I and valency $2$}\label{Coupl}
When applying the previous computation method to evaluate the to \ref{Morsewheel} dual picture, with flipped direction of the edges, we yield
\begin{equation}\label{2C}\int_{w\in\mathbb{D}_1(0)\backslash\{w_1,w_2\}}\hspace{-1.5cm}{\d}\phi^\lambda(w,w_1){\wedge}{{\d}\phi^\lambda(w,w_2)}=\frac{1}{\pi}\arg\left((1-w_1\overline{w_2})\cdot\frac{1-w_2}{1-w_1}\right)\end{equation}
where we used $\ln(z)-\ln(\overline{z})=2\I\arg(z)$. More detailed the computation starts with
\begin{align*}=\int_{w\in\mathbb{D}_1(0)\backslash\{w_1,w_2\}}&\biggr(\frac{-{\d}\overline{w}}{1-\overline{w}}+\frac{{\d}w}{1-w}+\lambda\frac{-{\d}w}{w_1-w}+\lambda\frac{\overline{{\d}w}}{1-w_1\overline{w}}\\&+(1-\lambda)\frac{{\d}\overline{w}}{\overline{w_1}-\overline{w}}-(1-\lambda)\frac{\overline{w_1}{\d}w}{1-\overline{w_1}w}\biggr)\\&\biggr(\frac{-{\d}\overline{w}}{1-\overline{w}}+\frac{{\d}w}{1-w}+\lambda\frac{-{\d}w}{w_2-w}+\lambda\frac{\overline{{\d}w}}{1-w_2\overline{w}}\\&+(1-\lambda)\frac{{\d}\overline{w}}{\overline{w_2}-\overline{w}}-(1-\lambda)\frac{\overline{w_2}{\d}w}{1-\overline{w_2}w}\biggr)\end{align*}
This time we can just use some combinations of the formulas established in the other proofs for vertices of valency $2$, we do not need to calculate other integrals: At the end we yield the $\lambda$ independent r.h.s. of equation \ref{2C} because of
$\lambda-(1-\lambda)=-1$
for the uncoupled factors $1-w_i$ and because of
$\lambda^2+2\lambda(1-\lambda)+(1-\lambda)^2=1$
for the coupled factor $1-w_1\overline{w_2}$.

Let us discuss the compatibility of the previous result with the normalization: The compatibility of the result \ref{2C} with the $HKR$-normalization \ref{HKRnormalization} implies that
$$\frac{1}{(2\pi)^2}\int_{w\in\mathbb{D}_1(0)\backslash\{u_1,u_2\}}\hspace{-0.6cm}{\d}\phi^\lambda(w,u_1){\wedge}{{\d}\phi^\lambda(w,u_2)}=-\frac{\I}{\pi}\arg\left((1-u_1\overline{u_2})\cdot\frac{1-u_2}{1-u_1}\right)$$
should naturally equal $\I/2$, where we can choose any distinct unit vectors $u_1,u_2\neq1$ via the usual $\mathbb R^+\ltimes\mathbb R$ action that obviously acts in Kontsevich's original formulation. By assumption $u_1\neq{u_2}$ and $u_i\neq1=w(\infty)$, hence we do not have to consider the value zero where the argument function $\arg$ is not defined. Notice that
$\arg\left((1-u_1\overline{u_2})\cdot\frac{1-u_2}{1-u_1}\right)$
is by construction $\mathbb R^+\ltimes\mathbb R$ invariant. To see this we just rewrite this expression again with the M\"obius transform
$w(z)=(z-\I)/(z+\I)$ as the difference of two propagators $\ln((\overline{z_2}-z_1)/(z_2-\overline{z}_1)$. If we choose $u_1=-1$ in the previous formula we get
\begin{align}
\I/2=&-\frac{\I}{\pi}\arg\left((1+\overline{u_2})\cdot\frac{1-u_2}{1+1}\right)=-\frac{\I}{\pi}\arg\left(-\Im(u_2)\right)
\end{align}
This result just depends on the natural order of the two points on the real line.
\section{$L_\infty$-relations for $\mathcal{U}^\lambda$ and the interpolation $\star$ products}
In summary the previous computations showed that $\forall\lambda\in\mathbb{C}$ the interpolation integral weights satisfy the same intricate quadratic relations concerning their codimension $1$ boundary strata as the original weights where $\lambda=1/2$. Therefore for the complex interpolation $\phi^\lambda$ the formula
\begin{align*}\;0=&\hspace{-0.3cm}\sum_{\substack{{\Gamma\in\mathcal{G}(n,m)}\\{\vert{E}_\Gamma\vert=2n+m-3}}}\left(\int_{{C}^+_{n,m}}\d\prod_{i=1}^{2n+m-3}{\d}\phi^\lambda({e_i})\right)U_\Gamma\\&=\hspace{-0.3cm}\sum_{\substack{{\Gamma\in\mathcal{G}(n,m)}\\{\vert{E}_\Gamma\vert=2n+m-3}}}\left(\int_{\partial{C}^+_{n,m}}\prod_{i=1}^{2n+m-3}{\d}\phi^\lambda({e_i})\right)U_\Gamma\\&=\sum_{\substack{{\Gamma\in\mathcal{G}(n,m)}\\{\vert{E}_\Gamma\vert=2n+m-3}}}\left[\hspace{-1.5cm}\sum_{\substack{{\Gamma'\in\mathcal{G}(n-1,m)}\\{\hspace{1.5cm}\exists(v_i,v_j)\in{E}_\Gamma:\Gamma/(v_i,v_j)=\Gamma'}}}\hspace{-1.1cm}W^\lambda_{\Gamma'}+\hspace{-1.5cm}\sum_{\hspace{2cm}\substack{{A\sqcup{B}\subsetneq[n]\sqcup[m]\;,\Gamma_{A,B}\in\mathcal{G}_{A,B}}\\{\hspace{0.2cm}\vert{E}_{\Gamma_{A,B}}\vert=2\vert{A}\vert+\vert{B}\vert-2}}}\hspace{-1.7cm}W^\lambda_{\Gamma_{A,B}}W^\lambda_{\Gamma/\Gamma_{A,B}}\right]U_\Gamma\end{align*}
\begin{align*}=&\sum_{\substack{{\Gamma'\in\mathcal{G}(n-1,m)}\\{\vert{E}_\Gamma'\vert=2n+m-2}}}\hspace{-0.4cm}W^\lambda_{\Gamma'}\hspace{-1.7cm}\sum_{\substack{{\Gamma\in\mathcal{G}(n,m)}\\{\hspace{1.5cm}\exists(v_i,v_j)\in{E}_\Gamma:\Gamma/(v_i,v_j)=\Gamma'}}}\hspace{-1.9cm}U_\Gamma\\&+\hspace{-1.7cm}\sum_{\hspace{2cm}\substack{{A\sqcup{B}\subsetneq[n]\sqcup[m]\;,\Gamma_{A,B}\in\mathcal{G}_{A,B}}\\{\hspace{0.2cm}\vert{E}_{\Gamma_{A,B}}\vert=2\vert{A}\vert+\vert{B}\vert-2}}}\hspace{-1.7cm}W^\lambda_{\Gamma_{A,B}}W^\lambda_{\Gamma'_{[n]/A,\{\bullet\}\cup[m]/B}}\hspace{-1.7cm}\sum_{\substack{{\Gamma\in\mathcal{G}(n,m)}\\{\hspace{1.5cm}\Gamma/\Gamma_{A,B}=\Gamma'_{[n]/A,\{\bullet\}\cup[m]/B}}}}\hspace{-1.9cm}U_\Gamma\end{align*}
is also a legitimate rewriting as usual, where we rearranged the differential operators $U_\Gamma$ appearing in this impressive equations with help of the Leibniz rule in the sense of Kontsevich \cite{K}: We can finally identify this equations with the $L_\infty$-relations
\begin{align}\label{Kontsevich}0=&\sum_{i<j}\pm\mathcal{U}^\lambda_{n-1}([\gamma_i,\gamma_j]_\mathrm{S-N}\wedge\cdots\wedge\gamma_n)\\&-\mathrm{d}_\mathrm{H}\mathcal{U}^\lambda_n(\gamma_1\wedge\cdots\wedge\gamma_n)\\&+\frac{1}{2}\hspace{-0.5cm}\sum_{\substack{{k,l\geq1}\\{k+l=n}\;,\;{\sigma\in\sum_n}}}\hspace{-0.7cm}\frac{\pm1}{k!l!}[\mathcal{U}^\lambda_k(\gamma_{\sigma_1}\wedge\cdots\wedge\gamma_{\sigma_k}),\mathcal{U}^\lambda_l(\gamma_{\sigma_{k+1}}\wedge\cdots\wedge\gamma_{\sigma_n})]_\mathrm{G}\nonumber\end{align}
that characterise a $L_\infty$-morphism
$$\mathcal{U}^{\lambda}:\left(T_\mathrm{poly}(A),0,[\cdot,\cdot]_\mathrm{S-N}\right)\leadsto\left(D_\mathrm{poly}(A),\mathrm{d}_\mathrm{H},[\cdot,\cdot]_\mathrm{G}\right)$$
between the two Lie algebras we are mainly interested in deformation quantization.

We refer the reader to \cite{K} for a discussion of signs and the factorials
$\prod\vert{\text{Star}{(v_i)}}\vert!$
that appear in the definition of the formality map
$$\mathcal{U}^\lambda_n=\sum_{m\geq0}\sum_{\Gamma\in\mathcal{G}(n,m)}\prod_{i=1}^{n}\frac{1}{\vert\text{Star}(i)\vert!}W^\lambda_\Gamma{U}_\Gamma$$
Especially that Kontsevich put the right signs is delicate, the article \cite{AMM} contains a detailed discussion of the signs in Kontsevich's formality map.\end{proof}
As in the usual case $\mathcal{U}^{1/2}$ the $L_\infty$ property ensures that solutions of the Maurer-Cartan equation are mapped to solutions and as a consequence we have on $\mathbb{R}^d$ for every $\lambda\in\mathbb{C}$ and for any Poisson structure $\Pi$ the Kontsevich interpolation $\star^\lambda$ star products
\begin{equation}\label{KSTAR}
f\star^\lambda{g}=fg+\sum_{l=0}^{\infty}\frac{\hbar^l}{l!}\mathcal{U}^\lambda_n(\underbrace{\Pi\wedge\cdots\wedge\Pi}_{l\;\text{times}})(f\otimes{g})
\end{equation}
We refer to the original article \cite{K} for details and the precise equivalence statement.
\section{Some relations of the interpolation polynomials}
Clearly the interpolation weights are holomorphic polynomials in $\lambda$ and the maximal polynomial degree is restricted by the number of edges $\#{E}_\Gamma$. It seems there are not many relations among the weights of a given graph, we just could figure out the following identities:
\subsection{Derivation of the $L_\infty$-relations}
We have by iterated derivation of the original $L_\infty$-relations the derived equations
\begin{align}0=&\hspace{-1.0cm}\sum_{\substack{{\Gamma'\in\mathcal{G}(n-1,m)}\\{\hspace{1.5cm}\exists(v_i,v_j)\in{E}_\Gamma:\Gamma/(v_i,v_j)=\Gamma'}}}\hspace{-1.1cm}\partial^{k}_\lambda{W}^\lambda_{\Gamma'}\\&+\hspace{-1.9cm}\sum_{\hspace{2cm}\substack{{A\sqcup{B}\subsetneq[n]\sqcup[m]\;,\Gamma_{A,B}\in\mathcal{G}_{A,B}}\\{\hspace{0.2cm}\vert{E}_{\Gamma_{A,B}}\vert=2\vert{A}\vert+\vert{B}\vert-2}}}\sum_{l=0}^{k}\binom{k}{l}\partial^{l}_\lambda{W}^\lambda_{\Gamma_{A,B}}\partial^{k-l}_\lambda{W}^\lambda_{\Gamma/\Gamma_{A,B}}\end{align}
The main part of the previous formulas was proved in the previous section \ref{REG} and the iterated derivation of this polynomial equations is standard. To interpret this derived equations on the level of differential operators seems to be a non-trivial question.

However there are also some more bizarre, in some sense derived relations that are satisfied by the interpolation integral weights as proved by Rossi and Willwacher, for details and a interpretation of their relations we refer the reader to the articles \cite{WR} and \cite{WiGRT}.

\subsection{The conjugation symmetry of the interpolation polynomials}
By conjugation of the interpolation propagator we have for the interpolation weights the symmetry
\begin{equation}\label{Fun}
\overline{W^\lambda_\Gamma}={W}^{1-\overline{\lambda}}_\Gamma
\end{equation}
For the Hadamard factorisation
$W^\lambda_\Gamma=C^\Gamma\prod_{i=0}^{\#{E}_\Gamma}\left(z^\Gamma_i-\lambda\right)$
this just implies the symmetry
$z^\Gamma_i\leftrightarrow1-\overline{z}^\Gamma_i$
but we do not have formulas to determine the zeros and from a computational point of view the following is more practical:

Consider $W^\lambda_\Gamma$ as a power series {\em i.e.} 
\begin{equation}\label{PolExp}W^\lambda_\Gamma=\sum_{i=0}^{\#{E}_\Gamma}a^\Gamma_i\lambda^i\end{equation}
It is clear that ${a}^\Gamma_0$ is the weight of the anti-logarithmic propagator and \ref{Fun} implies
\begin{align}\label{Tay1}
\sum_{n=0}^{\#{E}_\Gamma}\overline{a}^\Gamma_n\overline{\lambda}^n\stackrel{!}{=}&\sum_{n=0}^{\#{E}_\Gamma}{a}^\Gamma_n(1-\overline{\lambda})^n=\sum_{n=0}^{\#{E}_\Gamma}\overline{\lambda}^n(-1)^n\left(\sum_{l=n}^{\#{E}_\Gamma}\binom{l}{n}{a}^\Gamma_l\right)
\end{align}
for the coefficients in $a^\Gamma_i$, hence we yield that \ref{Fun} is equivalent to
\begin{equation}\label{FunImp}
\overline{a}^\Gamma_n=(-1)^n\sum_{l=n}^{\#{E}_\Gamma}\binom{l}{n}{a}^\Gamma_l
\end{equation}
Clearly the previous equations do not determine the coefficients $a^\Gamma_n$ completely, but the formula in general reduces the computation a bit:

If ${\#{E}_\Gamma}$ is even the real parts of the coefficients ${a}^\Gamma_{\#{E}_\Gamma-2n-1}$ is determined by \ref{FunImp} if we know ${a}^\Gamma_{\#{E}_\Gamma-2n}\;\forall{n}\in\mathbb{N}$ and if ${\#{E}_\Gamma}$ is odd the real parts of the coefficients ${a}^\Gamma_{\#{E}_\Gamma-2n-2}$ is determined by \ref{FunImp} if we know ${a}^\Gamma_{\#{E}_\Gamma-2n-1}\;\forall{n}\in\mathbb{N}$.  

Clearly by the previous symmetry argument \ref{Fun} the original Kontsevich integral weights, {\em i.e.} we set in the interpolation $\lambda=1/2$, are real, this fact is well-known.

\subsection{Reality of the polynomials of type $(n,0)$ and $(n,1)$}
We claim that some more weights are actually purely real or purely imaginary
\begin{equation}\label{Symmetry}\overline{W^\lambda_\Gamma}=(-1)^{n-1}{W^{\overline\lambda}_\Gamma}={W^{1-\overline\lambda}_\Gamma}\quad\forall\;\Gamma\;\in\Gamma_{n,0}\cup\Gamma_{n,1}\end{equation}
This symmetry implies for the original Kontsevich integral weights the vanishing
$${W^{1/2}_\Gamma}=0\;\forall\;\Gamma\;\in\Gamma_{2n,0}\cup\Gamma_{2n,1}$$

Let us for simplicity assume $\Gamma\;\in\Gamma_{n,0}$. To show \ref{Symmetry} we 
use the universal method for computing Kontsevich integral weights: First we transform the integrals from $\mathbb{H}$ to $\mathbb{D}_1(0)$ with help of the mentioned M\"obius transform
$z\rightarrow\frac{z-\I}{z+\I}$.

We use the action of $\mathbb R^+\ltimes\mathbb R$ to fix a vertex of type $I$ at zero, the centre of $\mathbb{D}_1(0)$, this choice is convenient. Some of the $1/2\pi\I$ factors in the definition of the interpolation propagator
\begin{equation*}\label{Interpol}\phi^{\lambda}(w_s,w_t)=\frac{1}{2\pi\I}\left(\lambda\ln\left(\frac{(1-\overline{w_s})(w_s-w_t)}{(1-w_s)(1-\overline{w_s}w_t)}\right)-(1-\lambda)\overline{\ln\left(\frac{(1-\overline{w_s})(w_s-w_t)}{(1-w_s)(1-\overline{w_s}w_t)}\right)}\right)\end{equation*}
compensate the $2\pi\I$ coming from the natural polar coordinates on $\mathbb{D}_1(0)$, the remaining integral is real if $\lambda$ is real:

Every vertex of type $I$ has a rotational degree of freedom, only the vertex of type $I$ that we did fix not. As we will see every rotational degree $e^{\I\varphi}$ brings a factor $\I\d\varphi$ and integration
cancels $n-1$ of the $2n-2$ factors $\frac{1}{2\pi\I}$ in the definition of interpolation propagators. We inductively expand the integral kernel with power series. For this expansion first observe that the partial derivatives of $\phi^{\lambda}(w_s,w_t)$ with respect to $w_s$ or $w_t$ and their conjugates are essentially sums of geometric series and we have to split the integration domain appropriate inductively because of
$$\frac{1}{1-x}=\sum_{l=0}^\infty{x}^l\quad\text{if}\quad\vert{x}\vert<1\quad\text{and}\quad\frac{1}{1-x}=-\sum_{l=0}^\infty1/x^{l+1}\quad\text{if}\quad\vert{x}\vert>1$$
Notice that this division in sub cases only depends on the norm of $x$. The oscillating integral formula
\ref{TI}, for instance
$$\int_{0}^{2\pi}{\d}{\varphi}e^{i{k}\varphi}=2\pi\delta_0^k\quad\;\forall\quad{k}\in\mathbb{Z}$$
couples Taylor expansions of different series, we see that the real or imaginary part respectively cancel. For the inductive integration of the radial coordinate we just need the obvious reality of the integrals
\begin{equation*}\int_{a}^b{\hspace{-0.3cm}{\d}r}\hspace{0.1cm}r^{n}\ln^m(r)=\begin{cases}\sum_{j=0}^{m}\frac{(-1)^jj!}{(n+1)^{j+1}}\binom{m}{j}\left(a^{n+1}\ln^{m-j}(a)-b^{n+1}\ln^{m-j}(b)\right)\quad\text{for}\;n\neq-1\\{\ln^{m+1}(a/b)/(m+1)\hspace{2.13cm}\text{for}\;n=-1}\end{cases}\end{equation*}
For this step it is helpful to have a look at the calculation of the Merkulov $n$-wheels.

The reason why the previous integration argumentation does not work if $\Gamma\;\in\Gamma_{n,m}$ with $m>1$ is the following the vertices of type $II$ are ordered: The ordering of type $II$ vertices {\em a priori} makes it for example necessary to consider certain linear combinations of integrals
$\int_{0}^{2\pi}{\d}{\varphi}e^{i{k}\varphi}\varphi^l=\sum_{a=0}^{l-1}(-1)^aa!\binom{l}{a}(2\pi)^{l-a}/(\I{k})^{a+1}\quad\text{if}\quad{k}\neq0$.
Because of the mixed imaginary and real part it is not clear that in this linear combinations with non-trivial coefficients either real or imaginary part will cancel in general if $\lambda\neq1/2$.

\chapter{\textit{Shoikhet-Tsygan formality for the interpolation propagator}}
In the following we briefly introduce Tsygan's formality conjecture \cite{TsyganChains}, whose original proof for the ordinary Kontsevich propagator $ \phi^{1/2}$ is due to Shoikhet \cite{ShTsy}.

By $\Omega^{-\bullet}_A$ we denote the exterior algebra of the $\mathbb K$-module of K\"ahler differentials over $A$ with reversed grading, in other words $\Omega^{-\bullet}_A$ denotes the graded space of smooth exterior differential forms on $\mathbb{K}^d$ and the degree of a $k$-differential form is $-k$. Here we again denote by $\mathbb{K}\supseteq \mathbb C$ a field of characteristic zero and by $A$ the ring $\mathbb{K}[x_1,\cdots,x_d]$ of formal power series in $d$ variables. By $T_\mathrm{poly}(A)$ we denote the graded vector space of totally skew-symmetric multi-derivations of $A$ and by $D_\mathrm{poly}(A)$ the graded vector space of multi-differential operators on $A$.

If $\gamma$ is a $k$-polyvector field we denote by $i_\gamma:\Omega^\bullet_A\rightarrow\Omega^{\bullet-k}_A$ the natural contraction of polyvector fields with forms and define the Lie derivative by the Cartan formula
$$L_{\gamma}=\mathrm{d}_{DR}\circ{i}_{\gamma}+{i}_{\gamma}\circ\mathrm{d}_{DR}$$
This Lie derivative and the trivial differential endow $\Omega^{-\bullet}_A$ with a structure of dg Lie module over the dg Lie algebra $T_\mathrm{poly}(A)$ and we have the equation
$$[L_{\gamma_1},L_{\gamma_2}]=L_{[\gamma_1,\gamma_2]_{S-N}}$$

We further denote by $C_{-\bullet}(A,A)$ the Hochschild chain complex of $A$ with reversed grading, endowed with the Hochschild differential and the Lie derivative at the level of Hochschild cochains acting on Hochschild chains. For instance $C_{-k}(A,A)=A\otimes{A}^{\otimes{k}}$ with completed tensor products and the differential $\mathrm{b}$ is defined by
\begin{equation}\label{ChainsB}\mathrm{b}(a_0\otimes..\otimes{a}_k)\hspace{-0.05cm}=\hspace{-0.05cm}a_0{a_1}\otimes{a_2}\otimes..\otimes{a_k}-a_0\otimes{a_{1}a_2}\otimes..\otimes{a_k}+...\pm{a_k{a_0}}\otimes{a_1}\otimes..\otimes{a_{k-1}}\end{equation}
There is a homological Hochschild-Kostant-Rosenberg quasi-isomorphism of dg vector spaces from $C_{-\bullet}(A)$ to $\Omega^{-\bullet}_A$, defined by the map
$$\mu(a_0,a_1,\cdots,a_k)=\frac{1}{k!}a_0\mathrm{d}{a_1}\wedge\cdots\wedge\mathrm{d}{a_k}$$
This natural map is called Connes map and it is a quasi-isomorphism of complexes $H_i(C_\bullet(A,A))=\Omega^{i}_A$, where $\Omega^{-\bullet}_A$ is equipped with zero differential, this version of the $\mathrm{HKR}$ theorem is due to Teleman.

The Lie derivative operators on the level of complexes, i.e for a $\Psi\in{C}^{\bullet}(A)[1]$ an operator $L_\Psi$ acting on $C_{\bullet}(A,A)$ can be defined by the formula
\begin{align*}L_{\Psi}(a_0\otimes\cdots\otimes{a}_n)&=\hspace{-0.25cm}\sum_{j=n-k}^{n}\hspace{-0.15cm}(-1)^{n(j+1)}\Psi(a_{j+1}\otimes\cdots\otimes{a}_0\cdots)\otimes{a}_{j+k-n}\otimes\cdots\otimes{a}_j\\&\hspace{-0.4cm}+\sum_{i=0}^{n-k}(-1)^{(k-1)(i+1)}a_0\otimes\cdots\otimes\Psi(a_{i+1}\otimes\cdots\otimes{a}_{i+k})\otimes\cdots\otimes{a}_n\end{align*}
for $\Psi\in{Hom}(A^{\otimes{k}},A)$. If we denote by $m:A^{\otimes}\rightarrow{A}$ the multiplication we recover the chain Hochschild differential by
$L_m=b$ and moreover we have
$$[L_{\Psi_1},L_{\Psi_2}]=L_{[\Psi_1,\Psi_2]_{G}}$$

Kontsevich's formality quasi-isomorphism $\mathcal{U}$ allows to define an $\L_\infty$-module structure on $C_{-\bullet}(A)$ {\em via} the explicit formula for the corresponding Taylor components
$$\phi_0=\mathrm{b}$$
$$\phi_k((\gamma_1\wedge\cdots\wedge\gamma_k)\otimes\omega)=L_{\mathcal{U}(\gamma_1,\cdots,\gamma_k)}\omega$$

Tsygan formality conjecture states the existence of an $L_\infty$-quasi-isomorphism
$$\mathcal{S}:C_{-\bullet}\left(A\right)\leadsto\Omega^{-\bullet}_A$$
of $L_\infty$-modules over $T_\mathrm{poly}(A)$ compatible with $\mathcal{U}$. Tsygan formality \cite{TsyganChains} is equivalent to the existence of Taylor components
$$\mathcal{S}_k:\wedge^k{T_\mathrm{poly}}(A)\otimes{C}_{-\bullet}(A)\rightarrow\Omega^{-\bullet}_A[-k]$$
that satisfy for $k\geq0$ the equations
$$\mathcal{S}_0(\omega)=\mu(\omega)$$
$$L_{\mathcal{U}_0}=\mathrm{b}$$
\begin{align}\label{Tsygan}&\sum\pm\mathcal{S}_{k+1}([\gamma_{i_1},\gamma_{i_2}]_{S-N}\wedge\gamma_{j_1}\wedge\cdots\wedge\gamma_{j_k}\otimes\omega)\\&+\sum_{p+q=k+2}\pm\mathcal{S}_{p+1}(\gamma_{i_1}\wedge\cdots\wedge\gamma_{i_p}\otimes{L}_{\mathcal{U}_q(\gamma_{j_1}\wedge\cdots\wedge\gamma_{j_q})}\omega)\nonumber\\&+\sum\pm{L}_{\gamma_i}\mathcal{S}_{k+1}(\gamma_{j_1}\wedge\cdots\wedge\gamma_{j_{k+1}}\otimes\omega)=0\nonumber\end{align}

\section{Shoikhet's construction of Tsygan formality}
\subsection{Shoikhet's configuration spaces and their boundary strata}
We again denote by $\mathbb{D}_1(0)=\{z\in\mathbb{C}:\vert{z}\vert<1\}$ the complex unit disk. We will borrow the notation of Shoikhet's original article \cite{ShTsy}, with the only difference that we do not change the name of the centre $\textbf{0}$ of the unit Disk to $\textbf{1}$ and hope the reader will not get confused by this more classical choice of notation instead of the notation in \cite{ShTsy} that Shoikhet describes as bad from any point of view.

The M\"obius transform
$$z\rightarrow{w}(z)=\frac{z-\I}{z+\I}$$
restricts to a conformal isomorphism from $\mathbb{H}^+\setminus\I$ to $\mathbb{D}_1(0)\setminus\textbf{0}$, the real axis gets mapped to the unit circle $S^1$ and $\infty$ to $1$.\\

The space $Conf_{\textbf{0},n,m}$ is defined by
$$Conf_{\textbf{0},n,m}=\{w_1,\cdots,w_n\in\mathbb{D}_1(0)\setminus\textbf{0},u_1,\cdots,u_m\in{S}^1,w_i\neq{w}_j,u_i\neq{u}_j\;\text{if}\;i\neq{j}\}$$
For $2n+m\geq1$ we have a free action of the group of rotations
$$S^1=\{z\rightarrow{e^{\I\phi}}, \phi\in\mathbb{R}/2\pi\mathbb{Z}\}$$
on $Conf_{\textbf{0},n,m}$ and can again define a Fulton-MacPherson like compactifications $D^+_{\textbf{0},n,m}$ of the quotient
$$D_{\textbf{0},n,m}=Conf_{\textbf{0},n,m}/S^1$$
For $k\geq1$ we denote by $D_k$ the $k-1$ dimensional manifold
$$D_k=\{z_1;\cdots,z_k\in\mathbb{C},\;z_i\neq{z}_j\;\text{if}\;i\neq{j}\}/\{z\rightarrow{p}z,\;p\in\mathbb{R}^+\}$$
We describe this compactified configuration spaces by specifying the codimension $1$ boundary strata of ${D}^+_{\textbf{0},n,m}$:
There are three different type of boundary strata, let us mention the helpful reference \cite{CaRo}:
\begin{itemize}
\item[S1)] The case where $N\geq2$ type I vertices collapse together to a single type I vertex in $\mathbb{D}_1(0)\setminus\textbf{0}$ with $2N-3\geq0$ and $2(n-N+1)+m-1\geq0$. This boundary strata is isomorphic to
$$C_N\times{D}_{\textbf{0},n+1-N,m}$$
\item[S2)] The case where $N\geq2$ type I vertices collapse together to $\textbf{0}$ with $2N-1\geq0$ and $2(n-N)+m-1\geq0$. This boundary strata is isomorphic to
$$D_N\times{D}_{\textbf{0},n-N,m}$$
\item[S3)] The case where $N\in\mathbb{N}$ type I vertices and $M\in\mathbb{N}$ type II vertices
with $N+M-2\geq0$ and $2(n-N)+m-M\geq0$ collapse together on $S^1$. This boundary strata is isomorphic to
$$C_{N,M}\times{D}_{\textbf{0},n-N,m+1-M}$$
\end{itemize}

\subsection{Shoikhet graphs and operators}
The set $\mathcal{G}(\textbf{0},n,m)$ of admissible graphs in Shoikhet's construction is analogous to the graphs $\mathcal{G}(n,m)$ in Kontsevich's construction with the difference that here we have a marked vertex $\textbf{0}$, the centre of the disk that is not a target of any edge, hence $\textbf{0}$ can be considered as the input of a differential form as we will see in the following construction of Shoikhet's operators corresponding to graphs:

We will define for every admissible graph $\Gamma$ with $n$ type I vertices and $m$ of type II (hence we can plug in $n$ polyvector fields $\gamma_i$ of degree $\#\text{Star}(i)$ and $m$ functions) an out-coming $\#\text{Star}(\textbf{0})$-differential form $\Omega^\Gamma_{\#\text{Star}(\textbf{0})}$. Here Shoikhet uses an analog of the construction in the Kontsevich formality map:

As usual we have to specify such a $C^\infty(\mathbb{R}^d)$ module linear map on the generators
$$\Omega^\Gamma_{\#\text{Star}(\textbf{0})}(\partial_{k_1}\wedge\cdots\wedge\partial_{k_{\#\text{Star}(\textbf{0})}})=\sum_{I:E_\Gamma\setminus\text{Star}(\textbf{0})\rightarrow\{1,\cdots,d\}}\Omega^{k_1,\cdots,k_{\text{Star}(\textbf{0})}}_I$$
We can extend the map $I:E_\Gamma\setminus\text{Star}(\textbf{0})\rightarrow\{1,\cdots,d\}$ to a map $\tilde{I}:E_\Gamma\rightarrow\{1,\cdots,d\}$ by setting $\tilde{I}(\textbf{0},\cdot)=k_s$ if the edge $(\textbf{0},\cdot)$ has the label $e_\textbf{0}^s$ in the graph $\Gamma$. For a vertex $v\neq\textbf{0}$ again as in the Kontsevich case we define
$$\tilde{\Psi}_v=\left(\prod_{e\in{E}_\Gamma,\;e=(\cdot,v)}\partial_{I(e)}\right)\Psi_v$$
With
$$\Omega^{k_1,\cdots,k_{\text{Star}(\textbf{0})}}_I=\prod_{v\in{V}_\Gamma\setminus\textbf{0}}\tilde{\Psi}_v$$
we define finally the function
$$\Omega^\Gamma_{\#\text{Star}(\textbf{0})}(\partial_{k_1}\wedge\cdots\wedge\partial_{k_{\#\text{Star}(\textbf{0})}})=\sum_I\Omega^{k_1,\cdots,k_{\text{Star}(\textbf{0})}}_I$$
\subsection{Shoikhet's propagator}
Shoikhet's propagator distinguishes between the marked vertex $\textbf{0}$ and other vertices as source. In the proof for the interpolation we will give alternative and more explicit formulas for Shoikhet's propagator that are well suited to introduce the interpolation propagator, but let us mention that in the original construction Shoikhet's propagator is defined more geometrical as follows:

\begin{itemize}
\item If $w_s\neq\textbf{0}$ we define 
$\phi_S^\lambda(w_s,w_t)$
as the angle between the two geodesics $(w_s,w_t)$ and $(w_s,\textbf{0})$ with respect to the Poincar\'{e} metric on $\mathbb{D}_1(0)$ and where the angle is counted from $(w_s,\textbf{0})$ to $(w_s,w_t)$ counterclockwise.
\item If the source of an edge is the central vertex $\textbf{0}$ we define the propagator $\phi_S^\lambda(\textbf{0},w_t)$ by the angle between $(\textbf{0},w_t)$ 
and $(\textbf{0},u_1)$ where $u_1\in{S}^1$ is the first vertex of type II. For pictures illustrating the situation we refer to \cite{ShTsy}.
\end{itemize}

\subsection{Shoikhet's $L_\infty$-quasi-isomorphism of $L_\infty$-modules}

Finally the maps $\mathcal{S}^\lambda_k:\wedge^k{T_\mathrm{poly}}(A)\otimes{C}_{-\bullet}(A)\rightarrow\Omega^{-\bullet}_A[-k]$ that determine the Shoikhet-Tsygan $L_\infty$-quasi-isomorphism
of $L_\infty$-modules over $T_\mathrm{poly}(A)$ can be defined by
$$\mathcal{S}_n=\sum_{m\geq0}\sum_{\Gamma\in\mathcal{G}(\textbf{0},n,m)}\frac{1}{\vert\text{Star}(\textbf{0})\vert!}\prod_{i=1}^{n}\frac{1}{\vert\text{Star}(k)\vert!}W_\Gamma\Omega^\Gamma_{\#\text{Star}(\textbf{0})}$$
where the integral weight of a graph $\Gamma\in\mathcal{G}(\textbf{0},n,m)$ is given by
$$W_\Gamma=\int_{D^+_{\textbf{0},n,m}}\wedge_{e\in{E}_\Gamma}\d\phi^\lambda_S(e)$$

\section{A proof of the Shoikhet-Tsygan formality theorem for the interpolation propagator}
In the following section we prove that Tsygan formality holds true for the aforementioned family $\mathcal U^\lambda$ of $L_{\infty}$-quasi-isomorphisms from $T_\mathrm{poly}(A)$ to $D_\mathrm{poly}(A)$: In other words, we modify Shoikhet's construction to yield a family $\mathcal S^\lambda$ of $L_\infty$-quasi-isomorphisms from $C_{-\bullet}(A)$ to $\Omega_A^{-\bullet}$ compatible with $\mathcal U^\lambda$ in the sense \ref{Tsygan}. More precise we want to do this by just replacing the integral weights in Shoikhet's construction by weights defined analogous to the weights in the case of the interpolation \ref{Interpol}. For this family $\mathcal S^\lambda$ one has to prove convergence of the corresponding integral weights and apply Stokes Theorem to prove that this modified weights still satisfy the Tsygan $L_\infty$-relations.
\begin{theorem}\label{InterpolationTsyganFormality}
For all $\lambda\in\mathbb{C}$ there is a Shoikhet-Tsygan $L_{\infty}$-quasi-isomorphism
$$\mathcal{S}^{\lambda}:C_{-\bullet}(A)\leadsto\Omega^\bullet_A$$
of $L_\infty$-modules over ${T}_\mathrm{poly}(A)$ compatible with $\mathcal{U}^\lambda$ and enjoying the properties:
\begin{enumerate}
\item  $\mathcal{S}^{\lambda}$ is $GL(d,\mathbb{K})$-equivariant.
\item  One can replace $\mathbb{K}^d$ in the construction by its formal completion $\mathbb{K}^d_{formal}$ at the origin.
\item The zeroth structure map of $\mathcal{S}^{\lambda}$ coincides with the homological Hochschild-Kostant-Rosenberg-quasi-isomorphism
$$\mu(a_0,a_1,\cdots,a_k)=\frac{1}{k!}a_0\mathrm{d}{a_1}\wedge\cdots\wedge\mathrm{d}{a_k}$$
\item If  $\gamma_0\in{T}^0_\mathrm{poly}(A)$ is linear in the coordinates of $\mathbb{R}^d$ and any set of polyvector fields $\gamma_i\in{T}_\mathrm{poly}(A)\;\forall{i}=2,\cdots,n$ and any Hochschild chain $a\in{C}_\bullet(A)$ we have
$$\mathcal{S}^{\lambda}(\gamma_0,\gamma_1,\cdots,\gamma_n;a)=0$$
\end{enumerate}
\end{theorem}
\

\begin{proof}[Proof:] First it is well-known that the two properties $i)$ and $ii)$ rely on Shoikhet's universal graph construction and not on the integral weights used in the construction of the Shoikhet-Tsygan formality map.
\subsection{The Shoikhet interpolation propagator}
Here we give the details how to adjust Shoikhet's construction in the interpolation to get compatibility, this question was one of the remaining tasks opened by the interpolation of Kontsevich formality.

We again distinguish between the marked vertex $\textbf{0}$ and other vertices as source vertices of an edge and describe the Shoikhet interpolation propagator as follows:
\begin{itemize}
\item If $w_s\neq\textbf{0}$ we define the propagator as the difference
$$\phi_S^\lambda(w_s,w_t)=\phi^\lambda(w_s,w_t)-\phi^\lambda(w_s,\textbf{0})$$
of two Kontsevich propagators, explicitly we have
$$\phi_S^\lambda(w_s,w_t)=\frac{1}{2\pi\I}\left[\lambda\ln\left(\frac{w_s-w_t}{w_s(1-\overline{w}_sw_t)}\right)-(1-\lambda)\ln\left(\frac{\overline{w}_s-\overline{w}_t}{\overline{w}_s(1-w_s\overline{w}_t)}\right)\right]$$
\item If the source of an edge is the central vertex $\textbf{0}$ we define the propagator by
$$\phi_S^\lambda(\textbf{0},w_t)=\frac{1}{2\pi\I}\left[\lambda\ln\left(\frac{w_t}{u_1}\right)-(1-\lambda)\ln\left(\frac{\overline{w}_t}{\overline{u}_1}\right)\right]$$
where $u_1=e^{\I\varphi_1}\in{S}^1$ is the first vertex of type II.
\end{itemize}

Notice that the interpolation propagator $\phi_S^\lambda$ is obviously invariant under the $S^1$ action on $Conf_{\textbf{0},n,m}$ and descends to a form on $D^+_{\textbf{0},n,m}$.
\subsection{Convergence of integrals for the interpolation}
First the convergence of the Shoikhet integral weights defined with the interpolation propagator can be justified quite analogous to the previous argumentation \ref{Convergence} in the case of the Kontsevich formality map with the interpolation propagator, essentially no new arguments are needed. For the convenience of the reader let us consider the maybe different seeming case of the marked vertex \textbf{0}:

We again describe by the polar coordinates ($w_1=\epsilon{e}^{\I\varphi}$ for the first vertex and $w_i=\epsilon{e}^{\I\varphi}z_i$ with $z_i\in\mathbb{C}$ for the other collapsing vertices) the boundary strata corresponding to the collapse $\epsilon\rightarrow0$ of the collapsing vertices to the marked vertex \textbf{0}. By computation
$$\d\phi_S^\lambda(\textbf{0},w_1)=\frac{1}{2\pi\I}\left[\I\d\varphi+(2\lambda-1)\frac{\d\epsilon}{\epsilon}-\lambda\frac{\d{u_1}}{u_1}+(1-\lambda)\frac{\d\overline{u}_1}{\overline{u}_1}\right]$$
where $u_1\in{S}^1$ is the first vertex of type II.

We multiply out the form of maximal degree: The terms $\d{u_1}/u_1$ and its conjugate are non-singular terms, because $u_1\in{S}^1$. The singular term $\d\epsilon/\epsilon$ again only appears always in the same linear combination with $\d\varphi$. Hence the appearing singular terms cancel each other in pairs by skew-symmetry of $1$-forms or the singularity $\d\epsilon/\epsilon$ gets compensated because we multiply it by a non-singular $1$-form proportional to $\epsilon\d\varphi$, where this non-singular $1$-form can come from edges connecting collapsing with not collapsing vertices. As in the previous computations the fact that the $1$-forms connecting collapsing with not collapsing edges are proportional to $\epsilon$ reflects the chain rule and the coupling in the coordinates proportional to $\epsilon{e}^{\I\varphi}$.

\subsection{Tsygan $L_\infty$-relations for the interpolation}
Now we come to the involved algebraic relations \ref{Tsygan} that the interpolation weights should satisfy. As in Shoikhet's proof we want to justify and rewrite the quite trivial vanishing
$$0=\int_{{D}^+_{\textbf{0},n,m}}{\d}\prod_{i=1}^{2n+m-2}{\d}\phi^\lambda({e_i})=\int_{\partial{D}^+_{\textbf{0},n,m}}\prod_{i=1}^{2n+m-2}{\d}\phi^\lambda({e_i})$$
where ${\d}\prod_{i=1}^{2n+m-2}{\d}\phi^\lambda({e_i})$ is the exterior derivative of a form corresponding to some Shoikhet graph $\Gamma\in\mathcal{G}(\textbf{0},n,m)$ with $2n+m-2$ edges. We now discuss quite analogous to \ref{REG} the regularised contribution of the different boundary strata.
\subsubsection{The boundary strata S1)}
We consider the case where $N\geq2$ type I vertices with $2N-3\geq0$ and $2(n-N+1)+m-1\geq0$ collapse together to a single vertex in $\mathbb{D}_1(0)\setminus\textbf{0}$, the codimension $1$ boundary strata isomorphic to
$$C_N\times{D}_{\textbf{0},n+1-N,m}$$

Let us W.L.O.G. enumerate the collapsing vertices and assume that there is an edge between the two collapsing vertices $h_1$ and $h_2$.

We use the section that fixes the angle of the collapsing vertex $w_1$ to some $\alpha\in[0,2\pi)$. We choose the coordinates $h_1=r{e}^{\I\alpha}$ with $r\in(0,1)$, $\phi$ fixed and $h_2=r{e}^{\I\alpha}+\epsilon{e}^{\I\varphi}$ for the first two vertices and $h_i=r{e}^{\I\alpha}+{e}^{\I\varphi}\epsilon{z_i}$ for the remaining collapsing vertices, quite analogous to the calculations in \ref{KV}. In other words we choose the section of $C_n$ where the first collapsing vertex is fixed to $0$ and the absolute square of the second collapsing vertex is set to $1$ and the section of ${D}_{\textbf{0},n+1-N,m}$ where the angle of the new vertex (corresponding to the collapsed vertices) is fixed to $\alpha$. Notice that the coordinate $r$ corresponds to an external coordinate of ${D}_{\textbf{0},n+1-N,m}$.

Because in the limit $\epsilon\rightarrow0$ external edges do not depend on the internal coordinates of $C_N$ it is clear that the integral factorizes and the argumentation goes along the same lines as in the proof \cite{ShTsy}: Because the Shoikhet interpolation propagator can be written as the difference
$\phi_S^\lambda(w_s,w_t)=\phi^\lambda(w_s,w_t)-\phi^\lambda(w_s,\textbf{0})$
of two Kontsevich propagators and the Kontsevich vanishing lemma \ref{KV} for the interpolation we only get a non-trivial contribution in the case $N=2$.

By the same arguments as in section \ref{SN} and in the sense of \cite{ShTsy} this contribution corresponds to the term in \ref{Tsygan} where the Schouten-Nijenhuis Lie bracket appears.
\subsubsection{The boundary strata S2)}
Now consider the case where $N\geq2$ vertices with $2(n-N)+m-1\geq0$ , lets say the vertices $w_{1},\cdots,w_{N}$, of type I collapse together to $\textbf{0}$ with codimension $1$ boundary strata isomorphic to
$$D_N\times{D}_{\textbf{0},n-N,m}$$
and the integral again factorizes in two integrals.

First notice that the condition $2(n-N)+m-1\geq0$ ensures that there is at least one vertex left that does not collapse to $\textbf{0}$ and to make computations easy we explicitly use the section that fixes the angle of the first type II vertex $u_1\in{S}^1$. For the first collapsing vertices we write $w_1=\epsilon{e}^{\I\varphi}$ and $w_i=\epsilon{e}^{\I\varphi}{z_i}$ for the remaining collapsing vertices, in other words we specified for $D_{\textbf{0},n-N,m}$ the section that fixes the argument of an external vertex to some $\alpha$ and for $D_N$ the section that fixes the absolute square of the first coordinate to the point $1$. 

We now can adopt Shoikhet's proof: Because we have the Kontsevich vanishing Lemma for the whole interpolation \ref{KV} it is clear that we only have to consider the case where one type I vertex $w_{1}$ converges to $\textbf{0}$.

Now consider what happens if the vertex $w_1=\epsilon\exp(\I\varphi)$ collapsing to $\textbf{0}$ is the target of an edge coming from a external vertex $w_k\neq\textbf{0}$.
$$\phi_S^\lambda(w_k,\epsilon{e}^{\I\varphi})=\lambda\ln\left(\frac{w_k-\epsilon{e}^{\I\varphi}}{w_k(1-\overline{w}_k\epsilon{e}^{\I\varphi})}\right)-(1-\lambda)\ln\left(\frac{\overline{w}_k-\epsilon{e}^{-\I\varphi}}{\overline{w}_k(1-w_k\epsilon{e}^{-\I\varphi})}\right)$$
It is quite obvious that the restriction of this function vanishes in the limit $\epsilon\rightarrow0$ because
$$\lim_{\epsilon\rightarrow0}\frac{w_k-\epsilon{e}^{\I\varphi}}{w_k(1-\overline{w}_k\epsilon{e}^{\I\varphi})}=\frac{\overline{w}_k-\epsilon{e}^{-\I\varphi}}{\overline{w}_k(1-w_k\epsilon{e}^{-\I\varphi})}=1$$
and $\ln(1)=0$. Commuting the limit $\epsilon\rightarrow$ and the partial differentials this shows that we only have to consider graphs where no edge ends at a vertex collapsing to $\textbf{0}$.

\subsubsection*{S2.1) There is an edge from $\textbf{0}$ to $w_{1}$}
On the one hand the edge from $\textbf{0}$ to $w_{1}=\epsilon\exp(\I\varphi)$ corresponds to the $1$-form
$${\d}\phi_S^\lambda(\textbf{0},\epsilon{e}^{\I\varphi})=(2\lambda-1)\frac{{\d}\epsilon}{\epsilon}+\I{\d}\varphi+\lambda{\d}\ln\left(\frac{1}{u_1}\right)-(1-\lambda){\d}\ln\left(\frac{1}{\overline{u}_1}\right)$$
and the pullback to the boundary $\epsilon={0}$ is
$\I{\d}\varphi$ where we recall that we fixed $u_1\in{S}^1$.

On the other hand consider $1$-forms corresponding to edges with source $w_{1}$ and target an external vertex $w_k\neq\textbf{0}$, for instance
$$\phi_S^\lambda(\epsilon{e}^{\I\varphi},w_k)=\lambda\ln\left(\frac{\epsilon{e}^{\I\varphi}-w_k}{\epsilon{e}^{\I\varphi}(1-\epsilon{e}^{-\I\varphi}w_k)}\right)-(1-\lambda)\ln\left(\frac{\epsilon{e}^{-\I\varphi}-\overline{w}_k}{\epsilon{e}^{-\I\varphi}(1-\epsilon{e}^{\I\varphi}\overline{w}_k)}\right)$$
By computation the $1$-form ${\d}\phi_S^\lambda(\epsilon{e}^{\I\varphi},w_k)$ pulled back to the boundary $\epsilon=0$ equals
$\I{\d}\varphi+\lambda{{\d}w_k}/{w_k}-(1-\lambda){{\d}\overline{w}_k}/{w_k}$.

We integrate a form of maximal degree over $D_1\times{D}_{n-1,m}$ and clearly propagators corresponding to external edges do not depend on $\varphi$ in the limit $\epsilon\rightarrow0$. Although the $1$-form $\d\phi_S^\lambda(\epsilon{e}^{\I\varphi},w_k)$ contains the term $\d\varphi$ dimensional reasoning implies that we have to choose the $\d\varphi$ term from the previous internal $1$-form ${\d}\phi_S^\lambda(\textbf{0},w_1)$ to get a non-trivial form of maximal degree on $D_1$ by skew-symmetry of $1$-forms. For example if there is only one edge departing from the collapsing vertex we have to consider on the $\epsilon=0$ boundary the $2$-form
$$\I\d\varphi\wedge\left(\lambda\frac{{\d}w_k}{w_k}-(1-\lambda)\frac{{\d}\overline{w}_k}{w_k}\right)$$
This pulled back $2$-form is in fact non-singular, the integration of the canonical volume form $\d\varphi$ over $D_1\cong{S}^1$ is trivial and the factor $\left(\lambda\frac{{\d}w_k}{w_k}-(1-\lambda)\frac{{\d}\overline{w}_k}{w_k}\right)$ contributes to the integral over ${D}_{n-1,m}$, recall we fixed $u_1\in{S}^1$.

In the sense of Shoikhet \cite{ShTsy} this shows that in the case S2.1) we end up with the usual contribution to \ref{Tsygan} corresponding to a summand of $\d{i}_{\gamma_i}\mathcal{S}^\lambda_{k+1}(\gamma_{j_1}\wedge\cdots\wedge\gamma_{j_{k+1}}\otimes\omega)$ if we replace the original propagator $\phi^{1/2}$ by $\phi^\lambda$.

\subsubsection*{S2.2) There is no edge from $\textbf{0}$ to $w_{1}$}
Because we assume that there is no edge connecting $\textbf{0}$ and $w_{1}=\epsilon\exp(\I\varphi)$ we can only build up a non-vanishing form of maximal degree on the factor $D_1$ if we choose the $S^1$ volume element $\d\varphi$ from pulled back forms $\I{\d}\varphi+\lambda{{\d}w_k}/{w_k}-(1-\lambda){{\d}\overline{w}_k}/{w_k}$ corresponding to edges departing from the collapsing vertex.

We again choose the section that fixes the argument of the first type II vertex $u_1\in{S}^1$. Multiplying out the pulled back $1$-forms we end up with the usual combinatorics if we replace the original propagator $\phi^{1/2}$ by $\phi^\lambda$.

In the sense of \cite{ShTsy} this shows that the case S2.2) contributes to \ref{Tsygan} with a summand of ${i}_{\gamma_i}\d\mathcal{S}^\lambda_{k+1}(\gamma_{j_1}\wedge\cdots\wedge\gamma_{j_{k+1}}\otimes\omega)$ and for every $\phi^\lambda$ in the case S2) we have the same combinatorics and end up in the usual situation:

In summary in the sense of Shoikhet \cite{ShTsy} the total contribution of the cases S2.1) and S2.2) to \ref{Tsygan} can be identified with the term where the Lie derivative appears. We refer to \cite{ShTsy} for a detailed discussion with a nice example how the calculus works.
\subsubsection{The boundary strata S3)}
Finally consider the case where $N\in\mathbb{N}$ type I vertices and $M\in\mathbb{N}$ type II vertices with $N+M\geq2$ collapse together to a single vertex $u$ on $S^1$, the boundary strata corresponding to this collapse is isomorphic to
$$C_{N,M}\times{D}_{\textbf{0},n-N,m+1-M}$$

Analogous to the proof of Kontsevich formality \ref{S2} this boundary contribution factorizes in two integrals because there are no singularities along this border.

For the convenience of the reader we now show that the internal integral over $C_{N,M}$ equals the Kontsevich integral with the interpolation propagator.

We calculate with the section that fixes the point where the vertices collapse to some $u\in{S}^1$. For the collapsing type I vertices we use the coordinates $w_1=u(1+\epsilon\I)$ and $w_i=u(1+\epsilon\I{h_i})$ for $i=2,\cdots,N$ with $h_i=x_i+\I{y}_i$ and $x_i,y_i\in\mathbb{R}$, for type II vertices we specify the coordinates by $u_i=u{e}^{\epsilon\I{r_i}}$ for $i=1,\cdots,M$. In other words we use the section of $C_{N,M}$ where we fix $h_1$ to $\I$ and the section of ${D}_{\textbf{0},n-N,m+1-M}$ where we fix the new vertex, corresponding to the collapsed vertices, to $u\in{S^1}$. 

Notice the condition $\vert{w}_i\vert<1$ is equivalent to
$\epsilon{(x_i^2+y^2_i)}/2<{y}_i$
and in the limit ${{\epsilon}\rightarrow0}$ the coordinates $h_i$ get restricted to $\mathbb{H}$ and quite similar coordinates $r_i$ of type II vertices get restricted to $\mathbb{R}$ with respect to their order, hence we indeed have the usual identification with $C_{N,M}$.

For example for internal edges between type I vertices by elementary computation
$$\lim_{{\epsilon}\rightarrow0}\phi_S^\lambda\bigr(u(1+\epsilon\I{h_1}),u(1+\epsilon\I{h_2})\bigr)=\lambda\ln\left(\frac{h_1-h_2}{\overline{h}_1-h_2}\right)-(1-\lambda)\ln\left(\frac{\overline{h}_1-\overline{h}_2}{h_1-\overline{h}_2}\right)$$
holds and by expanding $\exp(\epsilon\I{r_i})$ as a Taylor series analogous formulas hold if type II vertices are involved, in other words in the limit ${{\epsilon}\rightarrow0}$ we recover on $C_{N,M}$ the family of interpolation propagators $\phi^\lambda$.

This implies that if we replace $\mathcal{S}$ by $\mathcal{S}^\lambda$ and $\mathcal{U}$ by $\mathcal{U}^\lambda$ as well we also yield in the case S3) the same combinatorics. Hence in the sense of \cite{ShTsy} the boundary S3) contributes to \ref{Tsygan} with the term where the Lie derivative on the level of complexes appears.

\subsection{Homological $HKR$-normalization for the interpolation}
Also for the interpolation it is clear that we have the usual homological Teleman $HKR$ quasi-isomorphism normalization, for instance
$$\mathcal{S}_0^\lambda(a_1\otimes\cdots\otimes{a}_{m})=\frac{\mu(a_1\otimes\cdots\otimes{a_m})}{(m-1)!}=\frac{1}{(m-1)!}a_1\d{a_2}\wedge\cdots\wedge\d{a_m}$$
This is because we have for $u_t\in{S}^1$ the original Shoikhet propagators
\begin{align*}
\phi_S^\lambda(\textbf{0},u_t)&=\frac{1}{2\pi\I}\left[\lambda\ln\left(\frac{u_t}{u_1}\right)-(1-\lambda)\ln\left(\frac{\overline{u}_t}{\overline{u}_1}\right)\right]\\&=\frac{1}{2\pi}\begin{cases}\arg\left(\frac{u_t}{u_1}\right)\;\text{if}\;t>1\\0\;\text{if}\;t=1\end{cases}
\end{align*}
We discussed the integration of the $m-1$ type II vertices in detail in  \ref{HKRnormalization}.
\subsection{Tsygan interpolation formality globalisation}
For the Dolgushev globalisation in the sense of Fedosov we need the vanishing Lemma
$$0=\int_{w\in\mathbb{D}_1(0)\backslash\{w_1,w_2\}}\hspace{-1.2cm}{\d}\phi_S^\lambda(w,w_1){\wedge}{{\d}\phi_S^\lambda(w_2,w)}$$
With the interpolation globalisation vanishing Lemma \ref{GobalizationVanishing} for $2$-valent vertices proved in section \label{2valent} we can argue exactly as in \cite{DT1}:

Again by writing the Shoikhet propagator as the difference
$\phi_S^\lambda(w_s,w_t)=\phi^\lambda(w_s,w_t)-\phi^\lambda(w_s,\textbf{0})$
we can reduce the situation to the calculation of the integral
$$=\int_{w\in\mathbb{D}_1(0)\backslash\{w_1,w_2\}}\hspace{-0.2cm}\left({\d}\phi^\lambda(w,w_1)-\d\phi^\lambda(w,\textbf{0})\right){\wedge}\left(\d\phi^\lambda(w_2,w)-\d\phi^\lambda(w_2,\textbf{0})\right)$$
Multiplying out the $2$-form and integrating the two terms containing ${\d}w{\d}\overline{w}$ it is not difficult to realise that the globalisation vanishing Lemma for Tsygan formality is a consequence of the globalisation Lemma for the interpolation Kontsevich formality.\end{proof}

\subsection{Formality of cyclic chains for the interpolation}
Willwacher in \cite{WillChains} proved another formality conjecture raised by Tsygan \cite{TsyganChains}. His proof relies on a nice compatibility of Shoikhet's formality with the deRham differential. To give the precise compatibility statement we need to introduce the deformed differential $b+\hbar{B}$ where $b$ was defined in \ref{ChainsB} and $B$ is defined by
$$B(a_0\otimes\cdots\otimes{a}_n):=\sum_{j=1}^{n}(-1)^{nj}1\otimes{a_j}\otimes\cdots\otimes{a_n}\otimes{a_0}\otimes\cdots\otimes{a_{j-1}}$$
with $a_{-1}:=a_n$ for simplification of notation. Willwacher proved the for a long time not noticed, remarkable compatibility $\mathcal{S}\circ{B}=\d\circ\mathcal{S}$. Willwacher pointed out that his proof of formality of chains should carry over to the interpolation: The proof of Willwacher is very technical, we refrain from going into the details here and refer the reader to the original article \cite{WillChains}. However, most of the calculations in \cite{WillChains} do not depend on the propagator, the two properties of the propagator used there namely are
\begin{itemize}
\item The $HKR$-normalization that $\phi_S^\lambda(\textbf{0},u)=\phi_S^{1/2}(\textbf{0},u)$ for any $u\in{S}^1$ holds for any interpolation propagator $\phi_S^\lambda$.
\item The transitivity equation
$$\phi_S^\lambda(x,y)+\phi_S^\lambda(y,z)=\phi_S^\lambda(x,z)$$
of the central interpolation propagator defined for $w_s,w_t\in\overline{\mathbb{D}_1(0)}$ by
$$\phi_S^\lambda(w_s,w_t):=\frac{1}{2\pi\I}\left[\lambda\ln\left(\frac{w_s}{w_t}\right)-(1-\lambda)\ln\left(\frac{\overline{w}_s}{\overline{w_t}}\right)\right]$$
\end{itemize}
Because of this two compatibilities we can {\em verbatim} copy the proof of Willwacher in \cite{WillChains} for the interpolation propagator without changes and claim
\begin{proposition}$$\mathcal{S}^{\lambda}\circ{B}=\d\circ\mathcal{S}^{\lambda}$$
\end{proposition}
\chapter{Appendix}
\section{The weight of the first wheel-like graph}
 In~\cite[Subsubsection 1.4.2]{K}, Kontsevich provides a general formula for an associative product $\star$ up to second order w.r.t.\ $\hbar$, but this formula does not really coincide with the operator $B_2(\cdot,\cdot)$ in \cite[Subsubsection 1.4.2]{K}.
At the end of~\cite[Chapter 7]{BCKT}, the Kontsevich integral weight of this first wheel-like graph is also computed, but one of the key arguments in the computation relies either on a deep result of~\cite{FSh,Sh} (which in turn implies the vanishing of the integral weights of all wheels with spokes pointing inwards to the centre), or on a direct computation by Cattaneo-Felder, hinted to in~\cite{Sh} (and motivating the conjecture proved therein).
However, the computations in~\cite[Chapter 7]{BCKT} are surely not elementary, despite their brevity.
Therefore, Rossi and I felt that no harm is done by writing down our own computations: Despite being therein nothing really new in the result, we use a slightly different way of applying Stokes' Theorem (namely, we use that Kontsevich's propagator is exact on a certain boundary stratum of its domain of definition).

The previous pictured wheel-like graph \ref{Graph2} of type $(2,2)$, does not contribute to the formula in~\cite[Subsubsection 1.4.2]{K}, but it should contribute: The ``method" for the calculation of this special weight is just to try to apply Stokes theorem several times, this method is not very intuitive, technical integrals remain and it may not work for arbitrary weights. Let us here mention that the method used in chapter \ref{MerkWe} is easier to adapt for the calculation of an arbitrary weight, but the universal method described in \ref{MerkWe} is less elegant and tedious.

Because of the action of $\mathbb R^+\ltimes\mathbb R$ we can fix $z_2$ at $\I$ for the calculation 
$$w_\Gamma=\frac{1}{(2\pi)^4}\int_{{\{r_1,r_2\in\mathbb{R},r_1\neq{r_2},z\in\mathbb{H},z\neq\I\}}^+}\hspace{-1.5cm}{\d}\phi(z,\I){{\d}\phi(z,r_1){\d}\phi(\I,z){\d}\phi(\I,r_2)}$$
\subsection*{First application of Stokes theorem}
Since $d\phi(z,r_1)$ is an exact form we can use Stokes-theorem and we now discuss the contribution of the different codimension $1$ boundaries:

The boundary term ${z\rightarrow\mathbb{R}}$ with $z=\lim_{\epsilon\rightarrow0}(r+\epsilon\I)$, with $r\in\mathbb{R}$ vanishes because
$$\lim_{\epsilon\rightarrow0}{\d}\phi(r+\epsilon\I,\I)={\d}\arg(1)=0$$
The boundary term ${z\rightarrow\I}$ with $z=\lim_{\epsilon\rightarrow0}(\I+\epsilon\exp(\I\varphi))$ vanishes because of
$$\lim_{\epsilon\rightarrow0}{\d}\phi(\I,\I+\epsilon\exp(\I\varphi))={\d}\varphi=\lim_{\epsilon\rightarrow0}{\d}\phi(\I+\epsilon\exp(\I\varphi),\I)$$
and skew-symmetry. The boundary term ${z\rightarrow\infty}$ with $z=\lim_{R\rightarrow\infty}(R\exp(\I\varphi))$ vanishes because of
$$\lim_{R\rightarrow\infty}{\d}\phi(\I,R\exp(\I\varphi))={\d}\arg(1)=0$$
Hence the only boundary strata of codimension $1$ that contributes is the boundary term $r_1\rightarrow{r_2}$.

Because of 
$\phi(\I,r)=-2\left[\frac{\pi}{2}-\arctan(r)\right]$
and
$\phi(z,r)=-2\left[\frac{\pi}{2}-\arctan\left(\frac{r-\Re(z)}{\Im(z)}\right)\right]$
we have
$$w_\Gamma=\frac{2}{(2\pi)^4}\int_{{\{r\in\mathbb{R},z\in\mathbb{H},z\neq\I\}}^+}\hspace{-1.3cm}{{\d}\phi(z,\I){\d}\phi(\I,z)}\left(\frac{2{\d}r}{1+r^2}\right)\left[\frac{\pi}{2}-\arctan\left(\frac{r-\Re(z)}{\Im(z)}\right)\right]$$
\subsection*{Application of the Residue Theorem}
We will be able to go on with the computation with help of the formula
\begin{equation}\label{eq-res}
f(\alpha,\beta)=\int_\mathbb R\frac{\arctan(\alpha x+\beta)}{1+x^2}{\d}x=\pi\arctan\!\left(\frac{\beta}{\alpha+1}\right)\
\end{equation}
where we assume $\alpha,\ \beta\in \mathbb R$ and $\alpha>0$. To prove this formula we use the fact that the integrand in~\eqref{eq-res} is an absolutely integrable function on $\mathbb R$; viewed as a function on $\alpha$, $\beta$ and $x$, it is obviously smooth w.r.t.\ the three variables in their respective domains of definition.
The partial derivative of the integrand w.r.t.\ $\beta$ is
$\frac{1}{(1+x^2)(1+(\alpha x+\beta)^2)}$,
which is also an absolutely integrable function w.r.t.\ $x$ over $\mathbb R$, hence
$$\partial_\beta f(\alpha,\beta)=\int_\mathbb R \frac{1}{(1+x^2)\left(1+(\alpha x+\beta)^2\right)}\mathrm{d}x$$
The integral on the r.h.s. of the previous identity can be computed by means of the Residue Theorem (it is obvious that all assumptions apply to the case at hand): We apply the Residue Theorem to $\mathbb H^+\sqcup \mathbb R$, thus the only relevant poles of the integrand are $z=i$ and $z=(i-\beta)/\alpha$. 
Some manipulations yield then
$$\partial_\beta f(\alpha,\beta)=\pi\frac{(\alpha+1)}{(\alpha+1)^2+\beta^2}=\pi\partial_{\beta}\arctan\!\left(\frac{\beta}{\alpha+1}\right)$$
whence $f(\alpha,\beta)=\pi\arctan\!\left(\frac{\beta}{\alpha+1}\right)+g(\alpha)$.
Clearly $f(\alpha,0)=0$ and this implies $g(\alpha)=0$, therefore the integral~\eqref{eq-res} equals $\pi\arctan\!\left(\frac{\beta}{\alpha+1}\right)$.\\

Plugging the expression on the r.h.s. of the previous equality, the fact that $\arctan$ is an odd function and elementary manipulations yield
$$w_\Gamma=\frac{2}{(2\pi)^3}\int_{{\{z\in\mathbb{H},z\neq\I\}}^+}{{\d}\phi(z,\I){\d}\phi(\I,z)\bigr[\pi-\arg(z+\I)\bigr]}$$
The formulas $\arg(\overline{s})=2\pi-\arg(s)$ and $\arg(-{s})=\pi+\arg(s)$ for $s\in\mathbb{H}$ imply
$${{\d}\phi(z,\I){\wedge}{\d}\phi(\I,z)}=-2{{\d}\arg(-z+\I){\wedge}{\d}\arg(z+\I)}$$
and by this
$w_\Gamma=w_\Gamma^1+w_\Gamma^2$ where we define $w_\Gamma^1$ and $w_\Gamma^2$ by
$$w_\Gamma^1=\frac{-2}{(2\pi)^2}\int_{\partial{\{z\in\mathbb{H},z\neq\I\}}^+}\arg(z+\I){\d}\arg(-z+\I)$$
$${w}_\Gamma^2=\frac{2}{(2\pi)^3}\int_{\partial{\{z\in\mathbb{H},z\neq\I\}}^+}\arg^2(z+\I){\d}\arg(-z+\I)$$
\subsection*{Second application of Stokes theorem}
For the boundary term $z\rightarrow\mathbb{R}$ we write $z=\lim_{\epsilon\rightarrow0}(r+\epsilon\I)$ and calculate
$$\lim_{\epsilon\rightarrow0}{\d}\arg(-r+\I(1-\epsilon))=-\lim_{\epsilon\rightarrow0}{\d}{\;\arctan\left(\frac{1-\epsilon}{r}\right)}=\frac{{\d}r}{1+r^2}$$
For the boundary $z\rightarrow\infty$ we write $z=\lim_{R\rightarrow\infty}(R\exp(\I\varphi))$ and have
$$\lim_{R\rightarrow\infty}{\d}\arg(-R\exp(\I\varphi)+\I)={\d}{\varphi}$$
For the collapse $z\rightarrow\I$ we write $z=\lim_{\epsilon\rightarrow0}(\I+\epsilon\exp(\I\varphi))$ and have
$$\lim_{\epsilon\rightarrow0}{\d}\arg(-\epsilon\exp(\I\varphi))={\d}{\varphi}$$

Now we choose the orientation of the boundary strata in the usual way
$$w_\Gamma^1=\frac{-2}{(2\pi)^2}\left[\int_{-\infty}^{\infty}{\d}r\left(\frac{1}{1+r^2}\right)\left(\frac{\pi}{2}-\arctan(r)\right)+\int_{0}^{\pi}{\d}\varphi\varphi-\int_{0}^{2\pi}{\d}\varphi\frac{\pi}{2}\right]$$
where the first integral corresponds to the first mentioned boundary term $z\rightarrow\mathbb{R}$, the second integral corresponds to the second mentioned boundary term $z\rightarrow\infty$ and the third integral corresponds to the third mentioned boundary term $z\rightarrow\I$.

The triviality of the integral $w_\Gamma^1$ can also be seen by means of an involution argument as in~\cite[Section 6]{K}, but this elementary computation also yields that $w_\Gamma^1$ vanishes:
$$w_\Gamma^1=\frac{-2}{(2\pi)^2}\Biggr[\frac{\pi}{2}\arctan(r)\vert_{-\infty}^{\infty}+{\frac{\varphi^2}{2}}\rvert_{0}^{\pi}-\frac{\pi}{2}{{\varphi}}\rvert_{0}^{2\pi}\Biggr]$$
$$=\frac{-2}{(2\pi)^2}\left[\frac{\pi}{2}\pi+\frac{\pi^2}{2}-\frac{\pi}{2}{2{\pi}}\right]$$
The elementary computation of $w_\Gamma^2=w_\Gamma$ is quite analogous:
$$w_\Gamma=\frac{2}{(2\pi)^3}\left[\int_{-\infty}^{\infty}\hspace{-0.2cm}{\d}r\left(\frac{1}{1+r^2}\right)\left(\frac{\pi}{2}-\arctan(r)\right)^2+\int_{0}^{\pi}\hspace{-0.2cm}{\d}\varphi\varphi^2-\int_{0}^{2\pi}\hspace{-0.2cm}{\d}\varphi\left(\frac{\pi}{2}\right)^2\right]$$
where the first integral corresponds to the first mentioned boundary term and so on. This integrals are easy to solve and we can finish the computation:
$$w_\Gamma=\frac{2}{(2\pi)^3}\Biggr[\left(\left(\frac{\pi}{2}\right)^2\arctan(r)+\frac{1}{3}\arctan^3(r)\right)\vert_{-\infty}^{\infty}+{\frac{\varphi^3}{3}}\rvert_{0}^{\pi}-\left(\frac{\pi}{2}\right)^2{{\varphi}}\rvert_{0}^{2\pi}\Biggr]$$
$$=\frac{2}{(2\pi)^3}\left[\left(\frac{\pi}{2}\right)^2\pi+\frac{2}{3}\left(\frac{\pi}{2}\right)^3+\frac{\pi^3}{3}-\left(\frac{\pi}{2}\right)^2{2{\pi}}\right]$$
$$=\frac{1}{24}$$

\section{\textit{The ``shadow" of the Merkulov $n$-wheels}}\label{MerkWe}
As noticed in \cite{AWRT} and \cite{WR} the logarithmic Kontsevich formality map $\mathcal{U}^{\ln}$ has nice number theoretic properties, because it corresponds to the original Drinfeld associator known as Knizhnik-Zamolodchikov associator \cite{Dri} and hence to multiple $\zeta$ values.

Merkulov showed that the weight of the $n$-wheel with spokes pointing outwards in the logarithmic case equals $\zeta(n)$. The calculation of this  $n$-wheel weights is contained in his article \cite{M} in the proof that for the logarithmic formality map $\mathcal{U}^{\ln}$ the characteristic class is given by
$$\exp\left(\sum_{n=2}^{\infty}\frac{\zeta(n)}{n}(x/2\pi\I)^n\right)$$

Here we show an alternative calculation of the Merkulov $n$-wheels. Our calculation slightly differs from Merkulov's original computation because we do not use a certain freedom to fix a vertex of the first type. We do not fix this vertex because we are interested in some formulas that correspond on the number theoretic side to an obvious $\mathbb R^+\ltimes\mathbb R$ invariance of the Kontsevich propagator. To the best of our knowledge the number theoretic meaning of the $\mathbb R^+\ltimes\mathbb R$ invariance has not been studied before in the literature and to consider this invariance was inspired by the discussion at the end of section \ref{Coupl}.

 The invariance of the logarithmic  Kontsevich propagator \cite{KoM} under the action of the group of holomorphic transformations of $\mathbb{CP}^1$ preserving the upper half-plane and the point $\infty$ yields for Merkulov's $n$-wheels a constant containing certain multiple harmonic series and polylogarithms, more precise
\begin{proposition}\label{MerkulovWheel} For every $2\leq{n}\in\mathbb{N}$ and $\vert{w}\vert<1$ 
$$\hspace{-0.5cm}\sum_{\hspace{0.5cm}\substack{{m,m'\in\mathbb{N}}\\{m+m'<n}}}\hspace{-0.6cm}(-1)^{m'}\binom{n}{m}\binom{n-m}{m'}\hspace{-4cm}\sum_{\substack{{l_i\in\mathbb{N}^+}\\{\hspace{3.4cm}l_1\leq{l_n}\leq{l_{n-1}\leq{\cdots\leq{l_{m+m'+2}}\leq{l_{m+m'+1}}}}}\\{\hspace{3.4cm}{l_{m+m'+1}}>l_{m+m'}>\cdots>l_2>l_1\;\text{if}\;m+m'>0}}}\hspace{-3.9cm}\frac{\vert{w}\vert^{2\left(-l_1+\sum_{i=m+1}^{m+m'+1}l_i\right)}}{\prod_{i=1}^{n}l_i}\hspace{-2.5cm}\prod_{\hspace{+1.5cm}i\in\{1,2,\cdots,n\}\backslash\{2,3,\cdots,m+m'+1\}}\hspace{-2.6cm}\left(1-\vert{w}\vert^{2l_i}\right)$$
is a constant function given by the value $\zeta(n)$.\end{proposition}

\begin{proof} The constant ``shadow" \ref{MerkulovWheel} corresponds to integration of the vertices $w_1,\cdots,w_n$ on the circle of the Merkulov $n$-wheel we pictured in the introduction \ref{Wheel}.

 Because we integrate a form of maximal degree and the propagators $\phi^{\ln}(w,w_i)$ do not contain $\overline{w}_i$ we have to evaluate
\begin{align*}\int_{w_1,\cdots{w}_n\in\mathbb{D}_1(0)}\hspace{-1.7cm}{\d}{w_{1}}{\d}\overline{w_{1}}{\cdots}{\d}{w_{n}}{\d}\overline{w_{n}}&\prod_{i=1}^n\left(\frac{1}{w_i-{w}}+\frac{\overline{w}}{{1-w_i\overline{w}}}\right)\\&\hspace{-0.5cm}\left(-\frac{1}{1-\overline{w_n}}+\frac{w_{1}}{1-\overline{w}_{n}w_{1}}\right)\prod_{i=2}^n\left(-\frac{1}{1-\overline{w}_{i-1}}+\frac{w_{i}}{1-\overline{w}_{i-1}{w_{i}}}\right)\end{align*}
where the point $w$ is the centre of the wheel and we convenient rescaled the interpolation propagator \ref{InterpoPro} by the factor $2\pi\I$, for the proof of \ref{MerkulovWheel} this rescaling is irrelevant, but let us mention that due to Merkulov the actual weight is $(-1)^{n(n-1)/2}\zeta(n)/(2\pi\I)^n$.

Notice that in
$$\left(-\frac{1}{1-\overline{w_n}}+\frac{w_{1}}{1-\overline{w}_{n}w_{1}}\right)\prod_{i=2}^n\left(-\frac{1}{1-\overline{w}_{i-1}}+\frac{w_{i}}{1-\overline{w}_{i-1}{w_{i}}}\right)$$
every $w_i\;\forall{i}=1\cdots{n}$ appears only once in the power series expansions and the vanishing lemma \ref{MerkulovVanishing} helps to eliminate some terms appearing in any Merkulov $n$-wheel, more precise it reduces the integral to the computation of
$$\int_{w_1,\cdots{w}_n\in\mathbb{D}_1(0)}\hspace{-1.7cm}{\d}{w_{1}}{\d}\overline{w_{1}}{\cdots}{\d}{w_{n}}{\d}\overline{w_{n}}\frac{w_{1}}{1-\overline{w}_{n}w_{1}}\prod_{i=2}^n\frac{w_{i}}{1-\overline{w}_{i-1}{w_{i}}}\prod_{i=1}^n\left(\frac{1}{w_i-{w}}+\frac{\overline{w}}{{1-w_i\overline{w}}}\right)$$
For example if we would now set $w=0$ then \ref{GS}, \ref{TI} and \ref{PI} quite immediately yield Merkulov's observation that the weight of the $n$-wheel equals $\zeta(n)$, as a hint notice that $w_i$ and $\overline{w_i}$ appear exactly once in two different geometric series and \ref{TI} couples the summation indices of the power series expansions.

Now we will show what happens if we evaluate this integral without setting $w=0$:
Two of the terms in the previous integral, for instance
\begin{align*}&\int_{w_1,\cdots{w}_n\in\mathbb{D}_1(0)}\hspace{-1.7cm}{\d}{w_{1}}{\d}\overline{w_{1}}{\cdots}{\d}{w_{n}}{\d}\overline{w_{n}}\frac{w_{1}}{1-\overline{w}_{n}w_{1}}\prod_{i=2}^n\frac{w_{i}}{1-\overline{w}_{i-1}{w_{i}}}\prod_{i=1}^n\frac{\overline{w}}{{1-w_i\overline{w}}}\\&\int_{w_1,\cdots{w}_n\in\mathbb{D}_{\vert{w}\vert}(0)}\hspace{-1.7cm}d{w_{1}}d\overline{w_{1}}{\cdots}d{w_{n}}d\overline{w_{n}}\frac{w_{1}}{1-\overline{w}_{n}w_{1}}\prod_{i=2}^n\frac{w_{i}}{1-\overline{w}_{i-1}{w_{i}}}\prod_{i=1}^n\frac{1}{{w_i-{w}}}\end{align*}
vanish if we integrate out the $n$ vertices, because here the intersection of all the equations resulting from
$$\int_{0}^{2\pi}{\d}{\varphi}e^{i{n}\varphi}=2\pi\delta_0^n$$
is contradicting contributions.

We multiply out the remaining formula and compute
$$\int_{\substack{{}\\{}\\{w_1,\cdots,{w}_m\in\mathbb{D}_1(0)}\\{w_{m+1},\cdots,{w}_{m+m'}\in\mathbb{D}_{\vert{w}\vert}(0)}\\{w_{m+m'+1},\cdots,{w}_n\in\mathbb{D}_1(0)\backslash\mathbb{D}_{\vert{w}\vert}(0)}}}\hspace{-3.5cm}{\d}{w_{1}}{\d}\overline{w_{1}}{\cdots}{\d}{w_{n}}{\d}\overline{w_{n}}\frac{w_{1}}{1-\overline{w}_{n}w_{1}}\prod_{i=2}^n\frac{w_{i}}{1-\overline{w}_{i-1}{w_{i}}}\prod_{i=1}^m\frac{\overline{w}}{{1-w_i\overline{w}}}\prod_{i=m+1}^{n}\frac{1}{w_i-{w}}$$
$$\hspace{-0.7cm}=\left(\frac{2\pi}{\I}\right)^n\hspace{-0.2cm}(-1)^{m'}\hspace{-4.1cm}\sum_{\substack{{l_i\in\mathbb{N}}\\{\hspace{3.4cm}l_1\leq{l_n}\leq{l_{n-1}\leq{l_{n-2}\cdots\leq{l_{m+m'+2}}\leq{l_{m+m'+1}}}}}\\{\hspace{3.4cm}{l_{m+m'+1}}>l_{m+m'}\cdots>l_2>l_1\;\text{if}\;m+m`>0}}}\hspace{-4.1cm}\frac{\vert{w}\vert^{2(l_{m+1}-l_1)}}{\prod_{i=1}^{n}(1+l_i)}\vert{w}\vert^{2(m'+\sum_{i=m+2}^{m+m'+1}l_i)}\hspace{-3.2cm}\prod_{\hspace{+2.6cm}i\in\{1,2,\cdots,n\}\backslash\{2,3,\cdots,m+m'+1\}}\hspace{-3.1cm}\left(1-\vert{w}\vert^{2(1+l_i)}\right)$$
where we again calculate with help of \ref{GS}, \ref{TI} and \ref{PI}. A little bit more precise \ref{TI} implies a system of equations that couples some power series expansions of the shape
\begin{align*}
1+l_1-l_2+k_1&=0\\&
\;\vdots\\
1+l_{m+m'}-l_{m+m'+1}+k_{m+m'}&=0\\
l_{m+m'+1}-l_{m+m'+2}-k_{m+m'+1}&=0\\&
\;\vdots\\
l_{n}-l_{1}-k_n&=0
\end{align*}
where the first block corresponds to some integrals $\int_0^1$ and some integrals $\int_0^{\vert{w}\vert}$ and the second block to integrals $\int_{\vert{w}\vert}^1$.

Finally each of this integrals has to be multiplied with the combinatorial factor
$$\binom{n}{m}\binom{n-m}{m'}$$
and we yield for every $2\leq{n}$ the invariant ``{shadow}" of the Merkulov $n$-wheel stated in \ref{MerkulovWheel}.\end{proof}

In \ref{MerkulovWheel} among other terms
$$\sum_{i=0}^{n}(-1)^i\binom{n}{i}\mathrm{Li}_n(\vert{w}\vert^{2i})$$
appears for $m=m'=0$, where
$$\mathrm{Li}_n(w)=\sum_{l=1}^{\infty}{w^l}/{l^n}\;\text{for}\;w\in\mathbb{D}_1(0)$$
denotes as usual the polylogarithm.

As mentioned the weight of the Merkulov $n$-wheels is well-known, but the idea of \ref{MerkulovWheel} is to demonstrate a bit the universal method of the computation of Kontsevich integrals and deduce from the obvious invariance of the propagator other statements that maybe are not so obvious, as discussed in the following example.
\subsection{Example:``Shadow" of the $2$-wheel and $\sum_{l\in\mathbb{N}^+}{1}/{l(m+l)}$ for $m\in\mathbb{N}^+$}
Here the general formula for the $n$-wheels
for the logarithmic propagator $\phi^1$ implies that the following function does not depend on $w\in\mathbb{D}_1(0)$:
\begin{equation}\label{2WheelDis}
\sum_{m=1}^\infty\frac{1-2\vert{w}\vert^{2m}+{\vert{w}\vert^{4n}}}{2m^2}-\hspace{-0.2cm}\sum_{\substack{{l_i\in\mathbb{N}}\\{l_2>l_1}}}\hspace{-0.1cm}\frac{\left(2-\vert{w}\vert^{2(1+l_1)}\right)\vert{w}\vert^{2(1+l_2)}}{(1+l_1)(1+l_2)}+\hspace{-0.1cm}\sum_{\substack{{m\in\mathbb{N}^+}\\{l\in\mathbb{N}}}}\hspace{-0.1cm}\frac{\vert{w}\vert^{2m}}{(1+l)(m+1+l)}
\end{equation}
By inspecting the four power series expansions displayed in \ref{2WheelDis} we yield for all $m\in\mathbb{N}^+$
\begin{equation}\label{MerkIde}
\sum_{l\in\mathbb{N}^+}\frac{1}{l(m+l)}=\frac{2}{m}\sum_{l=1}^{{m-1}}\frac{1}{l}-\hspace{-0.3cm}\sum_{\substack{{l_1>l_2\in\mathbb{N}^+}\\{l_1+l_2=m}}}\hspace{-0.1cm}\frac{1}{l_1l_2}+\begin{cases} -\frac{1}{m^2}\;,\;$m$\;\text{even}\\\frac{1}{m^2}\quad,\;$m$\;\text{odd}\end{cases}
\end{equation}
The r.h.s of \ref{MerkIde} is clearly a finite sum of rational terms and hence also the l.h.s of \ref{MerkIde}, namely the infinite series
$\sum_{l\in\mathbb{N}^+}{1}/{l(m+l)}$
in fact is a rational number for all $m\in\mathbb{N}^+$, but notice that this rationality property is violated for $m=0$ because of Euler's formula
$\zeta(2)=\sum_{l=1}^\infty\frac{1}{l^2}=\pi^2/6$
and Lindemann's famous transcendence result.

It is well-known that the harmonic numbers satisfy
\begin{equation}\label{HamIde}
\sum_{l\in\mathbb{N}^+}\frac{1}{l(m+l)}=\frac{1}{m}\sum_{{l}=1}^{{m}}{1}/{l}
\end{equation}
For the proof one uses telescope sums and we thank Pieter Moree for pointing out this easier identity with help of the harmonic numbers.

By rewriting
$$\sum_{\substack{{l_1>l_2\in\mathbb{N}^+}\\{l_1+l_2=m}}}\frac{1}{l_1l_2}=\sum_{l=1}^{\lfloor{(m-1)/2}\rfloor}\frac{1}{l}+\frac{1}{m-l}$$
and elementary manipulations it is not difficult to see that both previous identities \ref{MerkIde} and \ref{HamIde} for the series 
$\sum_{l\in\mathbb{N}^+}{1}/({l(m+l)})$
are indeed equivalent. However, one can consider it as a slight conceptional advantage of our calculation \ref{MerkulovWheel} that we interpret the formulas \ref{HamIde} as a incarnation of the obvious scaling and real translation invariance of the Kontsevich propagator in the special case of the Merkulov $2$-wheel because we have some generalisations to $n$-wheels.

\vfill{\begin{center}
\textcircled{c} Copyright by Johannes L\"offler, 2015, All Rights Reserved
\end{center}}
\end{document}